\documentclass[11pt]{article}

\usepackage{fullpage}
\usepackage{circuitikz}
\usepackage{amsmath,amsfonts,amsthm,mathrsfs,mathpazo,xspace,graphicx}
\usepackage{endnotes}
\usepackage{color}
\usepackage{times}
\usepackage{amssymb,latexsym}
\usepackage{cite}
\usepackage{mdframed}
\usepackage{authblk}
\usepackage{braket}
\usepackage{esint}
\usepackage{xcolor}
\usepackage{bm}
\usepackage{caption}
\usepackage{subcaption}
\usepackage[all]{xy}
\usepackage[linesnumbered,ruled,vlined]{algorithm2e}
\SetKwFor{RepTimes}{repeat}{times}{end}
\usepackage[mathscr]{euscript}

\usepackage{tikz}
\usetikzlibrary{fit}

\usepackage{array}
\newcolumntype{C}{>{{}}c<{{}}}

\definecolor{ddgreen}{rgb}{255,0,0}
\definecolor{damethyst}{rgb}{0.4, 0.2, 0.6}
\usepackage[colorlinks, linkcolor=red,citecolor=green]{hyperref}

\newcommand{\Mod}[1]{\ (\mathrm{mod}\ #1)}

\newcommand{\expref}[2]{\texorpdfstring{\hyperref[#2]{#1~\ref{#2}}}{#1~\ref{#2}}}

\tikzset{sArrow/.style={->,>=stealth,thick}}
\tikzset{gArrow/.style={->,>=stealth,thick,gray}}
\tikzset{arrowLabel/.style={auto}}
\tikzset{blocResource/.style={draw,minimum width=4cm,minimum height=1.9cm}}
\tikzset{brnode/.style={minimum width=3.4cm,minimum height=.2cm}}
\tikzset{largeResource/.style={draw,minimum width=3.736cm,minimum height=3.25cm}}
\tikzset{medResource/.style={draw,minimum width=3.736cm,minimum height=1.75cm}}
\tikzset{thinResource/.style={draw,minimum width=1.618*2cm,minimum height=1cm}}
\tikzset{protocol/.style={draw,minimum width=1.545cm,minimum height=2.5cm}}
\tikzset{protocolLong/.style={draw,minimum height=1cm,minimum width=2.8cm}}
\tikzset{pnode/.style={minimum width=1cm,minimum height=.5cm}}
\tikzset{simulator/.style={draw,minimum width=2.8cm,minimum height=1.7cm}}

\DeclareMathAlphabet\mathbfcal{OMS}{cmsy}{b}{n}

\newtheorem{construction}{Construction}

\newtheorem{definition}{Definition}
\newtheorem{theorem}{Theorem}
\newtheorem{corollary}{Corollary}
\newtheorem{lemma}{Lemma}

\newtheorem{claim}{Claim}
\newtheorem*{remark}{Remark}
\newtheorem*{conjecture}{Conjecture}

\newmdtheoremenv[backgroundcolor=gray!10,
                 linewidth=0pt,
                 innerleftmargin=4pt,
                 innerrightmargin=4pt,
                 innertopmargin=-2pt,
                 innerbottommargin=4pt,
            splitbottomskip=4pt]{protocol}[prot]{Protocol}

\newmdtheoremenv[backgroundcolor=gray!10,
                 linewidth=0pt,
                 innerleftmargin=4pt,
                 innerrightmargin=4pt,
                 innertopmargin=-2pt,
                 innerbottommargin=4pt,
            splitbottomskip=4pt]{experiment}[exp]{Experiment}

\usepackage{mathtools}

\DeclarePairedDelimiter{\ip}{\langle}{\rangle}

\newcommand{\N}{\mathbb{N}}

\newcommand{\Z}{\mathbb{Z}}

\renewcommand{\vec}[1]{\mathbf{#1}}
\newcommand{\id}{\mathbb{1}}

\newcommand{\ketbra}[2]{\left|#1\right\rangle\!\!\left\langle #2\right|}

\newcommand{\Tr}{\mathrm{Tr}}
\newcommand{\tr}{\mathrm{tr}}

\newcommand{\proj}[1]{\ensuremath{|#1\rangle \langle #1|}}

\renewcommand{\rho}{\varrho}

\newcommand{\POVM}{\ensuremath{\mathsf{POVM}}\xspace}

\newcommand{\FT}{\ensuremath{\mathsf{FT}}\xspace}

\newcommand{\ot}{\otimes}
\newcommand{\eps}{\varepsilon}

\newcommand{\PKE}{\ensuremath{\mathsf{PKE}}\xspace}
\newcommand{\DualPKE}{\ensuremath{\mathsf{DualPKE}}\xspace}
\newcommand{\DualFHE}{\ensuremath{\mathsf{DualFHE}}\xspace}
\newcommand{\DualFHECD}{\ensuremath{\mathsf{DualFHE}_\mathsf{CD}}\xspace}
\newcommand{\QKD}{\ensuremath{\mathsf{QKD}}\xspace}

\newcommand{\hmin}{H_{\min}}

{\begin{framed}\begin{small}}
{\end{small}\end{framed}}

\newcommand{\negl}{\ensuremath{\operatorname{negl}}\xspace}
\newcommand{\from}{\ensuremath{\leftarrow}}
\newcommand{\bit}{\{0,1\}}
\newcommand{\sk}{\mathsf{sk}\xspace}
\newcommand{\pk}{\mathsf{pk}\xspace}
\newcommand{\vk}{\mathsf{vk}\xspace}
\newcommand{\ct}{\mathsf{CT}\xspace}

\newcommand{\rand}{\raisebox{-1pt}{\ensuremath{\,\xleftarrow{\raisebox{-1pt}{$\scriptscriptstyle\$$}}\,}}}

\DeclareMathOperator{\supp}{supp}

\newcommand{\KeyGen}{\ensuremath{\mathsf{KeyGen}}\xspace}
\newcommand{\Enc}{\ensuremath{\mathsf{Enc}}\xspace}
\newcommand{\Dec}{\ensuremath{\mathsf{Dec}}\xspace}

\newcommand{\Eval}{\ensuremath{\mathsf{Eval}}\xspace}
\newcommand{\Extract}{\ensuremath{\mathsf{Extract}}\xspace}
\newcommand{\Exp}{\ensuremath{\mathsf{Exp}}\xspace}
\newcommand{\Adv}{\ensuremath{\mathsf{Adv}}\xspace}
\newcommand{\aux}{\ensuremath{\mathsf{aux}}\xspace}

\newcommand{\poly}{\operatorname{poly}}
\newcommand{\guess}{\operatorname{guess}}

\newcommand{\algo}{\mathcal}
\newcommand{\Verify}{\ensuremath{\mathsf{Verify}}\xspace}

\newcommand{\Vrfy}{\ensuremath{\mathsf{Vrfy}}\xspace}

\newcommand{\Del}{\ensuremath{\mathsf{Del}}\xspace}

\newcommand{\QPT}{\ensuremath{\mathsf{QPT}}\xspace}

\newcommand{\LWE}{\ensuremath{\mathsf{LWE}}\xspace}
\newcommand{\SIS}{\ensuremath{\mathsf{SIS}}\xspace}
\newcommand{\ISIS}{\ensuremath{\mathsf{ISIS}}\xspace}
\newcommand{\CPTP}{\mathsf{CPTP}}

\newcommand{\FHE}{\ensuremath{\mathsf{FHE}\xspace}}
\newcommand{\HE}{\ensuremath{\mathsf{HE}\xspace}}
\newcommand{\NAND}{\ensuremath{\mathsf{NAND}\xspace}}

\newcommand{\PPT}{\ensuremath{\mathsf{PPT}}\xspace}

\newcommand{\PKECD}{\ensuremath{\mathsf{PKE}_{\mathsf{CD}}}\xspace}
\newcommand{\DualPKECD}{\ensuremath{\mathsf{DualPKE_\mathsf{CD}}}\xspace}
\newcommand{\FHECD}{\ensuremath{\mathsf{FHE}_\mathsf{CD}}\xspace}
\newcommand{\HECD}{\ensuremath{\mathsf{HE}_\mathsf{CD}}\xspace}

\newcommand{\INDCPA}{\ensuremath{\mathsf{IND\mbox{-}CPA}}\xspace}

\newcommand{\INDCPACD}{\ensuremath{\mathsf{IND\mbox{-}CPA\mbox{-}\mathsf{CD}}}\xspace}

\newif\ifnotes\notesfalse

\begin{document}

\title{Quantum Proofs of Deletion for Learning with Errors}

\sloppy

\author{Alexander Poremba\footnote{\href{aporemba@caltech.edu}{aporemba@caltech.edu}}}
\affil{California Institute of Technology}
%\date{December 14, 2020}
\maketitle

\begin{abstract}
Quantum information has the property that measurement is an inherently destructive process.
This feature is most apparent in the principle of complementarity, which states that mutually incompatible observables cannot be measured at the same time. Recent work by Broadbent and Islam (TCC 2020) builds on this aspect of quantum mechanics to realize a cryptographic notion called \emph{certified deletion}.
While this remarkable notion enables a classical verifier to be convinced that a (private-key) quantum ciphertext has been deleted by an untrusted party, it offers no additional layer of functionality. 

In this work, we augment the proof-of-deletion paradigm with fully homomorphic encryption ($\mathsf{FHE}$). 
We construct the first fully homomorphic encryption scheme with certified deletion -- an interactive protocol which enables an untrusted quantum server to compute on encrypted data and, if requested, to simultaneously prove data deletion to a client. Our scheme has the desirable property that verification of a deletion certificate is \emph{public}; meaning anyone can verify that deletion has taken place.
Our main technical ingredient is an interactive protocol by which a quantum prover can convince a classical verifier that a sample from the Learning with Errors ($\mathsf{LWE}$) distribution in the form of a quantum state was deleted.
As an application of our protocol, we construct a \emph{Dual-Regev} public-key encryption scheme with certified deletion, which we then extend towards a (leveled) $\mathsf{FHE}$ scheme of the same type. We introduce the notion of \emph{Gaussian-collapsing} hash functions -- a special case of collapsing hash functions defined by Unruh (Eurocrypt 2016) -- and we prove the security of our schemes under the assumption that the Ajtai hash function satisfies a certain
\emph{strong} Gaussian-collapsing property in the presence of leakage.

Our results enable a form of everlasting cryptography and give rise to new privacy-preserving quantum cloud applications, such as private machine learning on encrypted data with certified data deletion.
\end{abstract}

\newpage
\tableofcontents
\newpage

\section{Introduction}
\label{sec:intro}

Data protection has become a major challenge in the age of cloud computing and artificial intelligence. The European Union, Argentina, and California recently introduced new data privacy regulations which grant individuals the right to request the deletion of their personal data by \emph{media companies} and other \emph{data collectors} -- a legal concept that is commonly referred to as the \emph{right to be forgotten}~\cite{GargGV20}. While new data privacy regulations have been put into practice in several jurisdictions, formalizing data deletion remains a fundamental challenge for cryptography. A key question, in particular, prevails:\ \\
\ \\
\emph{How can we certify that user data stored on a remote cloud server has been deleted?}\ \\
\ \\
Without any further assumptions, the task is clearly impossible to realize in conventional cloud computing. This is due to the fact that there is no way of preventing the data collector from generating and distributing additional copies of the user data. Although it impossible to achieve in general, \emph{proofs-of-secure-erasure}~\cite{cryptoeprint:2010:217,tcc-2011-23513} can achieve a limited notion of data deletion under \emph{bounded memory assumptions}. Recently, Garg, Goldwasser and Vasudevan~\cite{GargGV20} proposed rigorous definitions that attempt to formalize the \emph{right to be forgotten}
from the perspective of classical cryptography. However, a fundamental challenge in the work of Garg et al.~\cite{GargGV20} lies in the fact that the data collector is always assumed to be \emph{honest}, which clearly limits the scope of the formalism.

A recent exciting idea is to use quantum information in the context of data privacy~\cite{Coiteux_Roy_2019,Broadbent_2020}.  Contrary to classical data, it is fundamentally impossible to create copies of an unknown quantum state thanks to the \emph{quantum no-cloning theorem}~\cite{Wootters1982Single}.
Broadbent and Islam~\cite{Broadbent_2020} proposed a quantum encryption scheme which enables a user to certify the deletion of a quantum ciphertext. Unlike classical proofs-of-secure-erasure, this cryptographic notion of certified deletion is achievable unconditionally in a fully malicious adversarial setting~\cite{Broadbent_2020}.
All prior protocols for certified deletion enable a client to delegate data in the form of plaintexts and ciphertexts with no additional layer of functionality. A key question raised by Broadbent and Islam~\cite{Broadbent_2020} is the following:\ \\
\ \\
\emph{Can we enable a remote cloud server to compute on encrypted data, while simultaneously allowing the server to prove data deletion to a client?}\ \\
\ \\
This cryptographic notion can be seen as an extension of fully homomorphic encryption schemes~\cite{Rivest1978,homenc,BrakerskiVaikuntanathan2011} which allow for arbitrary computations over encrypted data. 
Prior work on certified deletion makes use of very specific encryption schemes that seem incompatible with such a functionality; for example, the private-key encryption scheme of Broadbent and Islam~\cite{Broadbent_2020} requires a classical \emph{one-time pad}, whereas the authors in \cite{hiroka2021quantum} use a particular \emph{hybrid encryption} scheme in the context of public-key cryptography.
While
homomorphic encryption enables a wide range of applications including private queries to a search engine and machine learning classification on encrypted data~\cite{eprint-2014-25801}, a fundamental limitation remains: once the protocol is complete, the cloud server is still in possession of the client's encrypted data. This may allow adversaries to break the encryption scheme retrospectively, i.e. long after the execution of the protocol.
This potential threat especially concerns data which is required to remain confidential for many years, such as medical records or government secrets.

\emph{Fully homomorphic encryption with certified deletion} seeks to address this limitation as it allows a quantum cloud server to compute on encrypted data while simultaneously enabling the server to prove data deletion to a client, thus effectively achieving a form of \emph{everlasting security}~\cite{10.1007/978-3-540-70936-7_3,hiroka2021certified}.

\subsection{Main results}

Our contributions are the following.

\paragraph{Quantum superpositions of $\LWE$ samples.} We use Gaussian states to encode samples from the Learning with Errors ($\LWE$) distribution~\cite{Regev05} for the purpose of \emph{certified deletion} while simultaneously preserving their full cryptographic functionality. Because verification of a deletion certificate amounts to checking whether it is a solution to the \emph{(inhomogenous) short integer solution} problem~\cite{DBLP:conf/stoc/Ajtai96}, our encoding results in encryption schemes with certified deletion which are publicly verifiable -- in contrast to prior work based on hybrid encryption and BB84 states~\cite{Broadbent_2020,hiroka2021certified}.
% In some sense, our encoding is the first technique that does not follow hybrid approach
Our technique suggests a generic template for \emph{certified deletion} protocols which can be applied to many other cryptographic primitives based on $\LWE$.

\paragraph{Gaussian-collapsing hash functions.} To analyze the security of our quantum encryption schemes based on Gaussian states, we introduce the notion of \emph{Gaussian-collapsing} hash functions -- a special class of so-called \emph{collapsing} hash functions defined by Unruh~\cite{cryptoeprint:2015/361}. Informally, a hash function $h$ is \emph{Gaussian-collapsing} if it is computationally difficult to distinguish a superposition of Gaussian-weighted pre-images under $h$ from a single (measured) pre-image. We prove that the \emph{Ajtai collision-resistant hash function}~\cite{DBLP:conf/stoc/Ajtai96} is Gaussian-collapsing
assuming the quantum subexponential hardness of decisional \LWE.

\paragraph{Dual-Regev public-key encryption with certified deletion.} Using Gaussian superpositions, we construct a public-key encryption scheme with certified deletion which is based on the \emph{Dual-Regev} scheme introduced by Gentry, Peikert and Vaikuntanathan~\cite{cryptoeprint:2007:432}. We prove the security of our scheme under the assumption that
Ajtai’s hash function satisfies a certain strong Gaussian-collapsing property in the presence of leakage.

\paragraph{(Leveled) fully homomorphic encryption with certified deletion.} We construct the first (leveled) fully homomorphic encryption ($\FHE$) scheme with certified deletion based on our aforementioned \emph{Dual-Regev} encryption scheme with the identical security guarantees. Our $\FHE$ scheme is based on the (classical) \emph{dual homomorphic encryption} scheme used by Mahadev~\cite{mahadev2018classical}, which is a variant of the $\FHE$ scheme by Gentry, Sahai and Waters~\cite{GSW2013}. Our protocol supports the evaluation of polynomial-sized Boolean circuits on encrypted data and, if requested, also enables the server to prove data deletion to a client.

\subsection{Overview}\label{sec:overview}

How can we certify that sensitive information has been deleted by an untrusted party? Quantum information allows us to achieve a cryptographic notion called \emph{certified deletion}~\cite{Coiteux_Roy_2019,Fu_2018,Broadbent_2020}. The main idea behind this concept is the \emph{principle of complementarity}. This feature allows us to encode information in two mutually incompatible bases -- a notion that has no counterpart in a classical world.

Broadbent and Islam~\cite{Broadbent_2020} construct a private-key quantum encryption scheme with certified deletion using a BB84-type protocol that closely resembles the standard quantum key distribution ($\QKD$) protocol~\cite{BB84,Tomamichel2017largelyself}.
The crucial idea behind the scheme is that the information which is necessary to decrypt is encoded in the \emph{computational basis}, whereas \emph{certifying deletion} requires a measurement in the incompatible \emph{Hadamard basis}.
The scheme in~\cite{Broadbent_2020} achieves a rigorous notion of certified deletion security: once the ciphertext is successfully deleted, the plaintext $m$ remains hidden even if the private key is later revealed.

Using a standard \emph{hybrid encryption scheme}, Hiroka, Morimae, Nishimaki and Yamakawa~\cite{hiroka2021quantum} extended the scheme in~\cite{Broadbent_2020} to both public-key and attribute-based encryption with certified deletion via the notion of \emph{receiver non-committing} ($\mathsf{RNC}$) encryption~\cite{10.5555/1756169.1756191,10.1145/237814.238015}. The security proof in\cite{hiroka2021quantum} relies heavily on the fact that the classical public-key encryption is \emph{non-committing}, i.e. it comes with the ability to equivocate ciphertexts to encryptions of arbitrary plaintexts.
As a complementary result, the authors also gave a public-key encryption scheme with certified deletion which is \emph{publicly verifiable} assuming the
existence of one-shot signatures and extractable witness encryption. This property enables anyone to verify a deletion certificate using a publicly available verification key.

All prior protocols for certified deletion enable a client to delegate data in the form of ciphertexts with no additional layer of functionality. In this work, we answer a question raised by Broadbent and Islam~\cite{Broadbent_2020} affirmatively, namely whether it is possible to construct a \emph{homomorphic} quantum encryption scheme with certified deletion. This cryptographic notion is remarkably powerful as it would allow a quantum cloud server to compute on encrypted data, while simultaneously enabling the server to prove data deletion to a client. So far, however, none of the encryption schemes with certified deletion can enable such a functionality. Worse yet, the hybrid encryption paradigm appears insufficient in order to construct homomorphic encryption with certified deletion (see \expref{Section}{sec:related}), and thus an entirely new approach is necessary. 

Our techniques deviate from the hybrid encryption paradigm of previous works~\cite{Broadbent_2020,hiroka2021certified}
and allow us to construct the \emph{first} homomorphic quantum encryption scheme with certified deletion which has the desirable feature of being publicly verifiable. The main technical
ingredient of our scheme is an interactive protocol by which a quantum prover can convince a classical verifier that a sample from the \emph{Learning with Errors}~\cite{Regev05} distribution in the form of a quantum state was deleted.

\paragraph{Quantum superpositions of $\LWE$ samples.}
The \emph{Learning with Errors} $(\LWE)$ problem was introduced by Regev~\cite{Regev05} and has given rise to numerous cryptographic applications, including public-key encryption~\cite{cryptoeprint:2007:432}, homomorphic encryption~\cite{cryptoeprint:2011/344,GSW2013} and attribute-based encryption~\cite{cryptoeprint:2014/356}.

The problem is described as follows. Let $n,m \in \N$ and $q\geq 2$ be a prime modulus, and $\alpha \in (0,1)$ be a noise ratio parameter. In its decisional formulation, the $\LWE_{n,q,\alpha q}^m$ problem asks to distinguish between
a sample $(\vec A \rand \Z_q^{n \times m},\vec s \cdot \vec A+ \vec e \Mod{q})$ from the $\LWE$ distribution and a uniformly random sample $(\vec A \rand \Z_q^{n \times m},\vec u \rand \Z_q^m)$. Here, $\vec s \rand  \Z_q^n$ is a uniformly random row vector and $\vec e \sim D_{\Z^m,\alpha q}$ is a row vector which is sampled according to the discrete Gaussian distribution $D_{\Z^m,\alpha q}$. The latter distribution assigns probability proportional to $\rho_r(\vec x)=e^{-\pi \|\vec x\|^2/r^2}$ to every lattice point $\vec x \in \Z^m$, for $r = \alpha q >0$.

%Our work assumes the hardness of $\LWE_{n,q,\alpha q}^m$ with subexponential parameter $1/\alpha$, and thus relies on the worst-case hardness of approximating short vector problems (e.g. the shortest independent vectors problem, $\mathsf{SIVP}$) in lattices to within a subexponential factor in $n$~\cite{Regev05,10.1145/3055399.3055489}.

How can we certify that a (possibly malicious) party has deleted a sample from the $\LWE$ distribution? 
The main technical insight of our work is that one can encode $\LWE$ samples as \emph{quantum superpositions} for the purpose of certified deletion while simultaneously preserving their full cryptographic functionality. 
Superpositions of $\LWE$ samples have been considered by Grilo, Kerenidis and Zijlstra~\cite{Grilo_2019} in the context of quantum learning theory and by Alagic, Jeffery, Ozols and Poremba~\cite{cryptography4010010}, as well as by Chen, Liu and Zhandry~\cite{chen2021quantum}, in the context of quantum cryptanalysis of $\LWE$-based cryptosystems.\\
Let us now describe the main idea behind our constructions. Consider the Gaussian superposition,\footnote{
A standard tail bound shows that the discrete Gaussian $D_{\Z^m,\sigma}$ is essentially only supported on $\{\vec x \in \Z^m : \|\vec x\|_\infty \leq \sigma \sqrt{m}\}$. We choose $\sigma \ll q/\sqrt{m}$ and consider the domain $\Z^m \cap (-\frac{q}{2},\frac{q}{2}]^m$ instead. For simplicity, we also ignore that $\ket{\hat\psi}$ is not normalized.}
    $$
    \ket{\hat\psi}_{XY} = \sum_{\vec x \in \Z_q^m} \rho_\sigma(\vec x) \ket{\vec x}_X \otimes \ket{\vec A \cdot \vec x \Mod{q}}_Y.
    $$
Here, we let $\sigma = 1/\alpha$ and use $\Z_q^m$ to represent $\Z^m \cap (-\frac{q}{2},\frac{q}{2}]^m$.
By measuring system $Y$ in the computational basis with outcome $\vec y \in \Z_q^n$, the state $\ket{\hat\psi}$ \emph{collapses} into the quantum superposition
\begin{align}\label{eq:dual-state-intro}
    \ket{\hat\psi_{\vec y}} = \sum_{\substack{\vec x \in \Z_q^m:\\ \vec A \vec x= \vec y \Mod{q}}} \rho_\sigma(\vec x) \ket{\vec x}.
\end{align}
Note that the state $\ket{\hat\psi_{\vec y}}$ is now a superposition of \emph{short} Gaussian-weighted solutions $\vec x \in \Z_q^m$ subject to the constraint $\vec A \cdot \vec x = \vec y \Mod{q}$. In other words, by measuring the above state in the computational basis, we obtain a solution to the so-called \emph{(inhomogenous) short integer solution} $(\ISIS)$ problem specified by $(\vec A,\vec y)$ (see \expref{Definition}{def:ISIS}). The quantum state $\ket{\hat\psi_{\vec y}}$ in Eq.~\eqref{eq:dual-state-intro} has the following \emph{duality property}; namely, by applying the (inverse) $q$-ary quantum Fourier transform we obtain the state
\begin{align}\label{eq:primal-state-intro}
\ket{\psi_{\vec y}} =\sum_{\vec s \in \Z_q^n} \sum_{\vec e \in \Z_q^m} \rho_{\frac{q}{\sigma}}(\vec e) \, \omega_q^{-\ip{\vec s,\vec y}} \ket{\vec s \vec A + \vec e \Mod{q}},
\end{align}
where $\omega_q = e^{2 \pi i/q}$ is the primitive $q$-th root of unity. We make this statement more precise in \expref{Lemma}{lem:switching}. Throughout this work, we will refer to $\ket{\psi_{\vec y}}$ and $\ket{\hat\psi_{\vec y}}$ as the \emph{primal} and \emph{dual} Gaussian state, respectively. 
Notice that the resulting state $\ket{\psi_{\vec y}}$ is now a quantum superposition of samples from the $\LWE$ distribution. This relationship was first observed in the work of Stehlé et al.~\cite{cryptoeprint:2009/285} who gave quantum reduction from $\SIS$ to $\LWE$ based on Regev's reduction~\cite{Regev05}, and was later implicitly used by Roberts~\cite{Roberts19} and Kitagawa et al.~\cite{kitagawa2021secure} to construct quantum money and secure software leasing schemes.

Our quantum encryption schemes with certified deletion exploit the fact that a measurement of $\ket{\psi_{\vec y}}$ in the \emph{Fourier basis} yields a short solution to the $\ISIS$ problem specified by $(\vec A,\vec y)$, whereas ciphertext information which is necessary to decrypt is encoded using $\LWE$ samples in the \emph{computational basis}.

\paragraph{Dual-Regev public-key encryption with certified deletion.}
 
The key ingredient of our homomorphic encryption scheme with certified deletion is the \emph{Dual-Regev} public-key encryption scheme introduced by Gentry, Peikert and Vaikuntanathan~\cite{cryptoeprint:2007:432}. 
%Unlike Regev's original $\PKE$ scheme in~\cite{Regev05}, the Dual-Regev $\PKE$ scheme has the property that the ciphertext takes the form of a regular sample from the $\LWE$ distribution together with an additive shift $b \cdot  \lfloor\frac{q}{2} \rfloor$ that depends on the plaintext $b \in \bit$. 
Using Gaussian states, we can encode Dual-Regev ciphertexts for the purpose of certified deletion while simultaneously preserving their full cryptographic functionality. 
Our scheme Dual-Regev scheme with certified deletion consists of the following efficient algorithms:
\begin{itemize}
\item To generate a pair of keys $(\sk,\pk)$, sample a random matrix $\vec A \in \Z_q^{n \times (m+1)}$ together with a particular short trapdoor vector $\vec t \in \Z^{m+1}$ such that $\vec A \cdot \vec t = \vec 0 \Mod{q}$, and let $\pk = \vec A$ and $\sk = \vec t$.

\item To encrypt $b \in \bit$ using the public key $\pk=\vec A$, generate the following for a random $\vec y \in \Z_q^n$:
$$
\vk \leftarrow (\vec A,\vec y), \quad\quad
\ket{\ct} \leftarrow \sum_{\vec s \in \Z_q^n} \sum_{\vec e \in \Z_q^{m+1}} \rho_{\frac{q}{\sigma}}(\vec e) \, \omega_q^{-\ip{\vec s,\vec y}} \ket{\vec s \vec A + \vec e +b \cdot (0,\dots,0, \lfloor\frac{q}{2} \rfloor)},
$$
where $\vk$ is a public verification key and $\ket{\ct}$ is the quantum ciphertext for $\sigma = 1/\alpha$.

\item To decrypt a ciphertext $\ket{\ct}$ using the secret key $\sk$, measure in the computational basis to obtain an outcome $\vec c \in \Z_q^{m+1}$, and output $0$, if $\vec c^T \cdot \sk\in \Z_q$
is closer to $0$ than to $\lfloor\frac{q}{2}\rfloor$,
and output $1$, otherwise.
\end{itemize}

To delete the ciphertext $\ket{\ct}$, we simply perform measurement in the Fourier basis. In \expref{Corollary}{cor:switching}, we show that the Fourier transform of the ciphertext $\ket{\ct}$ results in the \emph{dual} quantum state
\begin{align}\label{eq:dual-with-phase}
\ket{\widehat{\ct}}=\sum_{\substack{\vec x \in \Z_q^{m+1}:\\ \vec A \vec x = \vec y \Mod{q}}}\rho_{\sigma}(\vec x) \, \omega_q^{\ip{\vec x,b \cdot (0,\dots,0,  \lfloor\frac{q}{2} \rfloor)}} \,\ket{\vec x}.
\end{align}
Notice that a Fourier basis measurement of $\ket{\ct}$ necessarily erases all information about the plaintext $b \in \bit$ and results in a \emph{short} vector $\pi \in \Z_q^{m+1}$ such that $\vec A \cdot \pi = \vec y \Mod{q}$. In other words, to verify a deletion certificate we can simply check whether it is a solution to the $\ISIS$ problem specified by the verification key $\vk=(\vec A,\vec y)$.
Our scheme has the desirable property that verification of a certificate $\pi$ is public; meaning anyone in possession of $(\vec A,\vec y)$ can verify that $\ket{\ct}$ has been successfully deleted. Moreover, due to the tight connection between worst-case lattice problems and the average-case $\ISIS$ problem~\cite{DBLP:journals/siamcomp/MicciancioR07,cryptoeprint:2007:432}, it is computationally difficult to produce a valid deletion certificate from $(\vec A,\vec y)$ alone.

To formalize security, we use the notion of \emph{certified deletion security} (i.e. $\INDCPACD$ security) \cite{Broadbent_2020,hiroka2021certified} which roughly states that, once deletion of the ciphertext is successful, the plaintext remains hidden even if the secret key is later revealed (see \expref{Definition}{def:INDCPACD}).
We prove the security of our schemes under the assumption that the Ajtai \emph{collision-resistant} hash function $h_{\vec A}(\vec x) = \vec A \cdot \vec x \Mod{q}$ satisfies a certain strong \emph{collapsing property} in the presence of leakage.

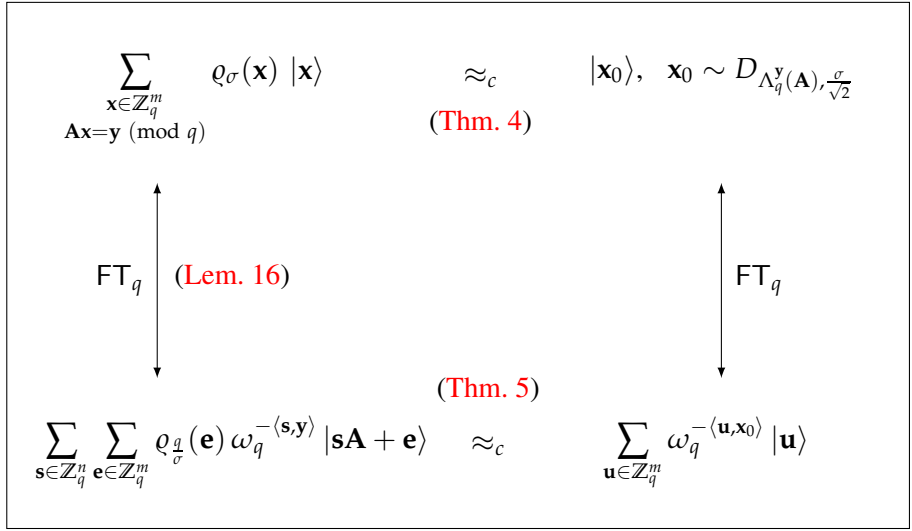
\begin{figure}[!]
\centering
\begin{tikzpicture}[> = latex, scale = 1]

\draw (0,0) rectangle (12,7);

\node at (2.5,5.7) {$ \displaystyle\sum_{\substack{\vec x \in \Z_q^m\\ \vec A \vec x = \vec y \Mod{q}}}\rho_{\sigma}(\vec x) \,\ket{\vec x} $};

\node at (9.5,6) {$\ket{\vec x_0}$, \,\, $\vec x_0 \sim D_{\Lambda_q^{\vec y}(\vec A),\frac{\sigma}{\sqrt{2}}}$};

\node at (3,1) {$\displaystyle \sum_{\vec s \in \Z_q^n} \sum_{\vec e \in \Z_q^m} \rho_{\frac{q}{\sigma}}(\vec e) \, \omega_q^{-\ip{\vec s,\vec y}} \ket{\vec s \vec A + \vec e}$};

\node at (6.3,6) {$\approx_c$};
\node at (6.4,1.1) {$\approx_c$};

\node at (9.3,1) {$\displaystyle \sum_{\vec u \in \Z_q^m} \omega_q^{-\ip{\vec u,\vec x_0}}\ket{\vec u}$};

\draw[<->] (2,2) -- (2,4.6);
\draw[<->] (9.5,2) -- (9.5,4.6);

\node at (1.5,3.3) {$\FT_q$};
\node at (10,3.3) {$\FT_q$};

\node at (3,3.35) {(\expref{Lem.}{lem:switching})};

\node at (6.3,5.4) {(\expref{Thm.}{thm:GaussCollapse})};

\node at (6.4,1.85) {(\expref{Thm.}{thm:pseudorandom-SLWE})};

\end{tikzpicture}
\caption{Technical overview of the main quantum states and their properties used throughout this work. The computational indistinguishability property holds under the (subexponential) quantum hardness of the (decisional) $\LWE$ assumption (\expref{Definition}{def:decisional-lwe}). Here, $\Lambda_q^{\vec y}(\vec A) = \{ \vec x \in \Z^m  :  \vec A \cdot \vec x = \vec y \Mod{q}\}$ denotes a particular coset of the $q$-ary lattice $\Lambda_q^\bot(\vec A) = \{\vec x \in \Z^m: \, \vec A \cdot\vec x = \vec 0 \Mod{q}\}$ defined in \expref{Section}{sec:lattices}.}\label{fig:diagram}
\end{figure}

\paragraph{Gaussian-collapsing hash functions.}

Unruh~\cite{cryptoeprint:2015/361} introduced the notion of collapsing hash functions in his seminal work on computationally binding quantum commitments. Informally, a hash function $h$ is called \emph{collapsing} if it is computationally difficult to distinguish between a superposition of pre-images, i.e. $\sum_{\vec x: \, h(\vec x)=\vec y} \alpha_{\vec x} \ket{\vec x}$, and a single measured pre-image $\ket{\vec x_0}$ such that $h(\vec x_0) = \vec y$.
Motivated by the properties of the dual Gaussian state in Eq.~\eqref{eq:dual-state-intro}, we consider a special class of hash functions which are \emph{collapsing} with respect to Gaussian superpositions. We say that a hash function $h$ is $\sigma$-\emph{Gaussian-collapsing} (formally defined in \expref{Definition}{def:gaussian-collapsing}), for some $\sigma >0$, if the following states are computationally indistinguishable:
    $$
    \sum_{\vec x:\,\, h(\vec x) = \vec y} \rho_\sigma(\vec x) \ket{\vec x} \quad \approx_c \quad \ket{\vec x_0}, \,\,\, \textrm{s.t. } \,\,h(\vec x_0) = \vec y.\,\,
    $$
Here, $\vec x_0$ is the result of a computational basis measurement of the the Gaussian superposition (on the left).
Notice that any collapsing hash function $h$ is necessarily also \emph{Gaussian-collapsing}, since a superposition of Gaussian-weighted vectors constitutes a special class of inputs to $h$. Liu and Zhandry~\cite{cryptoeprint:2019/262} implicitly showed that the \emph{Ajtai hash function} $h_{\vec A}(\vec x) = \vec A \cdot \vec x \Mod{q}$ is collapsing -- and thus \emph{Gaussian-collapsing} -- via the notion of \emph{lossy functions} and (decisional) $\LWE$. In \expref{Theorem}{thm:GaussCollapse}, we give a simple and direct proof of the Gaussian-collapsing property assuming (decisional) $\LWE$, which might be of independent interest. 

The fact Ajtai's hash function is Gaussian-collapsing has several implications for the security of our schemes. Because our Dual-Regev ciphertext corresponds to the Fourier transform of the state in Eq.~\eqref{eq:dual-with-phase}, the Gaussian-collapsing property immediately implies the semantic (i.e., $\INDCPA$) security under decisional $\LWE$ (see \expref{Theorem}{thm:pseudorandom-SLWE}).
We refer to~\expref{Figure}{fig:diagram} for an overview of our Gaussian states and their properties.

To prove the stronger notion of $\INDCPACD$ security of our Dual-Regev scheme with certified deletion, we have to show that, once deletion has taken place, the plaintext remains hidden even if the secret key (i.e., a short trapdoor vector $\vec t$ in the kernel of $\vec A$) is later revealed. 
In other words, it is sufficient to show that Ajtai's hash function satisfies a particular \emph{strong Gaussian-collapsing property} in the presence of leakage; namely, once an adversary $\algo A$ produces a valid short certificate $\pi$ with the property that $\vec A\cdot \pi = \vec y \Mod{q}$, then $\algo A$ cannot tell whether the input at the beginning of the experiment corresponded to a Gaussian superposition of pre-images or a single (measured) pre-image, even if $\algo A$ later receives a short trapdoor vector $\vec t$ in the kernel of $\vec A$.
Here, it is crucial that $\algo A$ receives the trapdoor vector $\vec t$ only \emph{after} $\algo A$ provides a valid pre-image witness $\pi$, otherwise $\algo A$ could trivially distinguish the two states by applying the Fourier transform and using the trapdoor $\vec t$ to distinguish between a superposition of $\LWE$ samples and a uniform superposition.

Unfortunately, we currently do not know how to prove the \emph{strong} Gaussian-collapsing property of the Ajtai hash function from standard assumptions (such as $\LWE$ or $\ISIS$). The problem emerges when we attempt to give a reduction between the $\INDCPACD$ security of our Dual-Regev public-key encryption
scheme with certified deletion and the $\LWE$ (or $\ISIS$) problem. In order to simulate the $\INDCPACD$ game successfully, we have to eventually forward a short trapdoor vector $\vec t \in \Z^{m+1}$ (i.e. the secret key) to the adversary once deletion has taken place. Notice, however, that the reduction has no way of obtaining a short trapdoor vector $\vec t$ such that $\vec A \cdot \vec t = \vec 0 \Mod{q}$ as it is trying to break the underlying $\LWE$ (or $\ISIS$) problem with respect to $\vec A$ in the first place (!) Recently, Hiroka, Morimae, Nishimaki and Yamakawa~\cite{hiroka2021certified} managed to overcome similar technical difficulties using the notion of \emph{receiver non-committing} ($\mathsf{RNC}$) encryption~\cite{10.5555/1756169.1756191,10.1145/237814.238015} in the context of \emph{hybrid encryption} in order to produce a \emph{fake} secret key. In our case, we cannot rely on similar techniques involving $\mathsf{RNC}$ encryption as it seems difficult to reconcile with homomorphic encryption, which is the main focus of this work.
Instead, we choose to formalize the strong Gaussian-collapsing property of the Ajtai hash function as a simple and falsifiable conjecture in \expref{Conjecture}{conj:SGC}.
We prove the following result in \expref{Theorem}{thm:Dual-Regev-PKE-CD} (assuming that \expref{Conjecture}{conj:SGC} holds):\\
\ \\ 
\textbf{Theorem} (informal): \emph{The Dual-Regev $\PKE$ scheme with certified deletion (see \expref{Construction}{cons:dual-regev-cd}) is $\INDCPACD$-secure under the strong Gaussian-collapsing assumption in \expref{Conjecture}{conj:SGC}.}\\
\ \\
To see why \expref{Conjecture}{conj:SGC} is plausible, consider the following natural attack. Given as input either
a Gaussian superposition of pre-images or a single (measured) pre-image, we perform the quantum Fourier transform, reversibly shift the outcome by a fresh $\LWE$ sample\footnote{To \emph{smudge} the Gaussian error of the initial superposition, we can choose an error from a discrete Gaussian distribution which has a significantly larger standard deviation.} and store the result in an auxiliary register. If the input corresponds to a superposition, we obtain a separate $\LWE$ sample which is \emph{re-randomized}, whereas if the input is a single (measured) pre-image, the outcome remains random. 
Hence, if the aforementioned procedure succeeded without disturbing the initial quantum state, we could potentially provide a valid certificate $\pi$ and also distinguish the auxiliary system with access to the trapdoor. However, by shifting the state by another $\LWE$ sample, we have necessarily entangled the two systems in a way that prevents us from finding a valid certificate via a Fourier basis measurement. We make this fact more precise in \expref{Section}{sec:uncertainty}, where we prove a general \emph{uncertainty relation for Fourier basis projections} (\expref{Theorem}{thm:uncertainty}) that rules out a large class of attacks, including the \emph{shift-by-$\LWE$-sample} attack described above.

Next, we extend our Dual-Regev scheme towards a (leveled) $\FHE$ scheme with certified deletion.

\paragraph{Dual-Regev fully homomorphic encryption with certified deletion.} Our (leveled) $\FHE$ scheme with certified deletion is based on the (classical) Dual-Regev leveled $\FHE$ scheme used by Mahadev \cite{mahadev2018classical} -- a variant of the scheme due to Gentry, Sahai and Waters~\cite{GSW2013}. 
Let $n,m \in \N$, let $q\geq 2$ be a prime modulus, and let $\alpha \in (0,1)$ be the noise ratio with $\sigma = 1/\alpha$. Let $N = (n+1) \lceil\log q\rceil$ and let $\vec G \in \Z_q^{(m+1) \times N}$ denote the \emph{gadget matrix} (defined in \expref{Section}{sec:FHE_construction}) designed to convert a binary representation of a vector back to its $\Z_q$ representation. The scheme consists of the following efficient algorithms:
\begin{itemize}
\item To generate a pair of keys $(\sk,\pk)$, sample a random matrix $\vec A \in \Z_q^{(m+1) \times n}$ together with a particular short trapdoor vector $\vec t \in \Z^{m+1}$ such that $\vec t\cdot\vec A = \vec 0 \Mod{q}$, and let $\pk = \vec A$ and $\sk = \vec t$.

\item To encrypt a bit $x\in \bit$ using the public key $\vec A \in \Z_q^{(m+1) \times n}$, generate the following pair consisting of a verification key and ciphertext for a random $\vec Y \in \Z_q^{n \times N}$ with columns $\vec y_1,\dots,\vec y_N \in \Z_q^{n}$:
$$
\vk \leftarrow (\vec A,\vec Y), \quad\,\,
\ket{\ct} \leftarrow \sum_{\vec S \in \Z_q^{n \times N}} \sum_{\vec E \in \Z_q^{(m+1)\times N}} \rho_{q/\sigma}(\vec E) \, \omega_q^{-\Tr[\vec S^T \vec Y]} \ket{\vec A\cdot \vec S + \vec E + x \cdot \vec G},
$$
where $\vec G \in \Z_q^{(m+1)\times N}$ denotes the \emph{gadget matrix} and where $\sigma = 1/\alpha$.

%\item To apply a $\NAND$ gate on ciphertexts $\ct_0$ and $\ct_1$, output the ciphertext $\vec G - \ct_0 \cdot \vec G^{-1}(\ct_1) \Mod{q}$.

\item To decrypt a quantum ciphertext $\ket{\ct}$ using the secret key $\sk$, measure in the computational basis to obtain an outcome $\vec C \in \Z_q^{(m+1)\times N}$ and compute $c = \sk^T \cdot \vec c_N \in \Z_q$, where $\vec c_N \in \Z_q^{m+1}$ is the $N$-th column of $\vec C$, and then output $0$, if $c$
is closer to $0$ than to $\lfloor\frac{q}{2}\rfloor$,
and output $1$, otherwise.
\end{itemize}
We remark that deletion and verification take place as in our Dual-Regev scheme with certified deletion.

%Let us now describe how to perform homomorphic operations on the encrypted data. 
Our $\FHE$ scheme supports the evaluation of polynomial-sized Boolean circuits consisting entirely of $\NAND$ gates, which are universal for classical computation. 
Inspired by the classical homomorphic $\NAND$ operation of the Dual-Regev scheme~\cite{GSW2013,mahadev2018classical}, we define an analogous quantum operation $U_\NAND$ in \expref{Definition}{def:homomorphic-NAND-gate} which allows us to apply a $\NAND$ gate directly onto Gaussian states. When applying homomorphic operations, the new ciphertext maintains the form of an $\LWE$ sample with respect to the same public key $\pk$, albeit for a new $\LWE$ secret and a new (non-necessarily Gaussian) noise term of bounded magnitude.
Notice, however, that the resulting ciphertext is now a highly entangled state since the unitary operation $U_\NAND$ induces entanglement between the $\LWE$ secrets and Gaussian error terms of the superposition.
This raises the following question: How can a server perform homomorphic computations and, if requested, afterwards prove data deletion to a client? In some sense, applying a single homomorphic $\NAND$ gates breaks the structure of the Gaussian states in a way that prevents us from obtaining a valid deletion certificate via a Fourier basis measurement.
Our solution to the problem involves a single additional round of interaction between the quantum server and the client in order to \emph{certify deletion}.

After performing a Boolean circuit $C$ via a sequence of $U_\NAND$ gates starting from the ciphertext $\ket{\ct} = \ket{\ct_1} \otimes \dots \otimes \ket{\ct_\ell}$ in system $C_{\mathsf{in}}$ corresponding to an encryption of $x = (x_1,\dots,x_\ell) \in \bit^\ell$, the server simply sends the quantum system $C_{\mathsf{out}}$ containing an encryption of $C(x)$ to the client. Then, using the secret key $\sk$ (i.e., a trapdoor for the public matrix $\pk$), it is possible for the client to \emph{extract} the outcome $C(x)$ from the system $C_{\mathsf{out}}$ with overwhelming probability without significantly damaging the state. We show that it is possible to rewind the procedure in a way that results in a state which is negligibly close to the original state in system $C_{\mathsf{out}}$. At this step of the protocol, the client has learned the outcome of the homomorphic application of the circuit $C$ while the server is still in possession of a large number of auxiliary systems (denoted by $C_\aux$) which mark intermediate applications of the gate $U_\NAND$. 
We remark that this is where the standard $\FHE$ protocol ends.
In order to enable \emph{certified deletion}, the client must now return the system $C_{\mathsf{out}}$ to the server. Having access to all three systems $C_{\mathsf{in}}C_{\mathsf{aux}} C_{\mathsf{out}}$, the server is then able to undo the sequence of homomorphic $\NAND$ gates in order to return to the original product state in system $C_{\mathsf{in}}$ (up to negligible trace distance). Since the ciphertext in the server's possession is now approximately a simple product of Gaussian states, the server can perform a Fourier basis measurement of systems $C_{\mathsf{in}}$, as required. Once the protcol is complete, it is therefore possible for the client to know $C(x)$ and to be convinced that data deletion has taken place. We prove the following in \expref{Theorem}{thm:FHE-CD-security}.\\ 
\\ 
\textbf{Theorem} (informal): \emph{Our Dual-Regev (leveled) $\FHE$ scheme with certified deletion (\expref{Construction}{cons:FHE-cd}) is $\INDCPACD$-secure under the strong Gaussian-collapsing assumption in \expref{Conjecture}{conj:SGC}.}

\paragraph{Open problems.} Our results leave open many interesting future research directions. For example, is it possible to prove \expref{Conjecture}{conj:SGC} -- and thus the $\INDCPACD$ security of our constructions -- from the hardness of $\LWE$ or $\ISIS$?  Another interesting direction is the following. Since the verification of our proofs of deletion only requires classical computational capabilities, this leaves open the striking possibility that all communication that is required for fully homomorphic encryption with certified deletion can be dequantized entirely, similar to work of Mahadev~\cite{mahadev2018classical} on delegating quantum computations, as well as recent work on classically-instructed parallel remote state preparation by Gheorghiu, Metger and Poremba~\cite{gheorghiu2022quantum}.

\subsection{Applications}

\paragraph{Data retention and the right to be forgotten.}

The European Union, Argentina, and California recently introduced new data privacy regulations -- often referred to as the \emph{right to be forgotten}~\cite{GargGV20} -- which grant individuals the right to request the deletion of their personal data by media companies. However, formalizing data deletion still remains a fundamental challenge for cryptography.
Our fully homomorphic encryption scheme with certified deletion achieves a rigorous notion of \emph{long-term data privacy}: it enables a remote quantum cloud server to compute on encrypted data and -- once it is deleted and publicly verified -- the client's data remain safeguarded against a future leak that reveals the secret key.

\paragraph{Private machine learning on encrypted data.} Machine learning algorithms are used for wide-ranging classification tasks, such as medical predictions, spam detection and face recognition. While
homomorphic encryption enables a form of privacy-preserving machine learning ~\cite{eprint-2014-25801}, a fundamental limitation remains: once the protocol is complete, the cloud server is still in possession of the client's encrypted data. This threat especially concerns data which is required to remain confidential for many years. Our results remedy this situation by enabling private machine learning on encrypted data with certified data deletion.

\paragraph{Everlasting cryptography.} 
Assuming that the server has not broken the computational assumption before data deletion has taken place, our results could potentially transform a long-term $\LWE$ assumption~\cite{Regev05} into a temporary one, and thus effectively achieve a form of \emph{everlasting security}~\cite{10.1007/978-3-540-70936-7_3,hiroka2021certified}.

\subsection{Related work}\label{sec:related}

The first work to formalize a notion resembling \emph{certified deletion} is due to
Unruh \cite{Unruh2013} who proposed a quantum timed-release
encryption scheme that is \emph{revocable}. The protocol allows a user to \emph{return} the ciphertext of a quantum timed-release encryption scheme, thereby losing all access to the data. Unruh's security proof exploits the \emph{monogamy of entanglement} in order to guarantee that the quantum revocation process necessarily erases all information about the plaintext. 
Subsequently, Coladangelo, Majenz and Poremba~\cite{coladangelo2020quantum} adapted this property to $\emph{revocable}$ programs in the context of \emph{secure software leasing}, a weaker notion of \emph{quantum copy-protection} which was proposed by Ananth and La Placa~\cite{ananth2020secure}.

Fu and Miller~\cite{Fu_2018} gave the first quantum protocol that proves deletion of a single bit using classical interaction alone. Subsequently, Coiteux-Roy and Wolf~\cite{Coiteux_Roy_2019} proposed a $\QKD$-like conjugate coding protocol that enables certified deletion of a classical plaintext, albeit without a complete security proof.

Independently of~\cite{Coiteux_Roy_2019}, Broadbent and Islam~\cite{Broadbent_2020} construct a private-key quantum encryption scheme with a rigorous definition of certified deletion using a BB84-type protocol that closely resembles the standard quantum key distribution protocol~\cite{BB84,Tomamichel2017largelyself}.
There, the ciphertext (without the optional quantum error correction part) consists of random BB84 states $\ket{x^\theta} = H^{\theta_1} \ket{x_1} \otimes \dots \otimes H^{\theta_n} \ket{x_n}$ together with a one-time pad encryption of the form $f(x_{| \theta_i=0}) \oplus m \oplus u$, where $u$ is a random string (i.e. a one-time pad key), $f$ is a two-universal hash function and $x_{| \theta_i=0}$ is the substring of $x$ to which no Hadamard gate is applied. The main idea behind the scheme is that the information which is necessary to decrypt is encoded in the \emph{computational basis}, whereas \emph{certifying deletion} requires a \emph{Hadamard basis} measurement. Therefore, if the verification of a deletion certificate is successful, $x_{| \theta_i=0}$ must have high entropy, and thus $f(x_{| \theta_i=0})$ is statistically close to uniform (i.e. $f$ serves as an extractor).
The private-key quantum encryption scheme of Broadbent and Islam~\cite{Broadbent_2020} achieves the notion of \emph{certified deletion security}: once the ciphertext is successfully deleted, the plaintext $m$ remains hidden even if the private key $(\theta,f,u)$ is later revealed.

Using a standard \emph{hybrid encryption scheme}, Hiroka, Morimae, Nishimaki and Yamakawa~\cite{hiroka2021quantum} extended the scheme in~\cite{Broadbent_2020} to both public-key and attribute-based encryption with certified deletion via the notion of \emph{receiver non-committing} ($\mathsf{RNC}$) encryption~\cite{10.5555/1756169.1756191,10.1145/237814.238015}; for example, to obtain a public-key encryption scheme with certified deletion, one simply outputs a quantum ciphertext of the \cite{Broadbent_2020} scheme together with a classical (non-committing) public-key encryption of its private key.
Given access to the $\mathsf{RNC}$ secret key, it is therefore possible to decrypt the quantum ciphertext. Crucially, the hybrid encryption scheme also inherits the certified deletion property of the~\cite{Broadbent_2020} scheme; namely, once deletion has taken place, the plaintext remains hidden even if the $\mathsf{RNC}$ secret key is later revealed. The security proof in~\cite{hiroka2021quantum} relies heavily on the fact that the classical public-key encryption is \emph{non-committing}, i.e. it comes with the ability to equivocate ciphertexts to encryptions of arbitrary plaintexts. To obtain a homomorphic encryption scheme with certified deletion, one would have to instantiate the hybrid encryption scheme with a classical (non-committing) homomorphic encryption scheme which is not known to exist. While generic transformations for non-committing encryption have been studied~\cite{cryptoeprint:2018/974}, they tend to be incompatible with basic homomorphic computations.
Moreover, it is unclear whether the candidate hybrid approach for homomorphic encryption is even secure: for all we know, a malicious adversary could use homomorphic evaluation to decouple the quantum part from the classical part of the ciphertext in order to obtain a classical encryption of the plaintext, thereby violating certified deletion security.

Hiroka, Morimae, Nishimaki and Yamakawa~\cite{hiroka2021certified} studied \emph{certified everlasting zero-knowledge proofs} for $\mathsf{QMA}$ via the notion of \emph{everlasting security} which was first formalized by M\"{u}ller-Quade and Unruh~\cite{10.1007/978-3-540-70936-7_3}. 
A recent paper by Coladangelo, Liu, Liu and Zhandry~\cite{coladangelo2021hidden} introduces \emph{subspace coset states} in the context of unclonable crytography in a way that loosely resembles our use of primal and dual Gaussian states.

In a subsequent and independent work, Khurana and Bartusek~\cite{https://doi.org/10.48550/arxiv.2207.01754} consider generic transformations for encryption schemes with certified deletion. Similar to Broadbent and Islam~\cite{Broadbent_2020}, they use a hybrid approach via BB84 states to construct (privately verifiable) public-key, attribute-based and homomorphic encryption schemes with \emph{certified everlasting security}: once deletion is successful, the security notion guarantees that the plaintext remains hidden even if the adversary is henceforth computationally unbounded.

\paragraph{Previous version of this paper.} We remark that a prior version of this paper was posted to \emph{arXiv}\footnote{\url{https://arxiv.org/abs/2203.01610v1}} and presented as unpublished work at QIP 2022. This paper contains substantial new improvements to the previous constructions: compared to the prior version of the paper which presented security proofs in the \emph{semi-honest} adversarial model, this work features security proofs in a fully malicious setting under the plausible \emph{strong Gaussian-collapsing property} of the Ajtai hash function, and also offers revised Dual-Regev encryption schemes with certified deletion that enable \emph{public verification} of deletion certificates.

\paragraph{Acknowledgments.} 
The author would like to thank Urmila Mahadev for pointing out an attack on an earlier version of our protocols, and for the idea behind the proof of \expref{Theorem}{thm:GaussCollapse}.
The author would also like to thank Thomas Vidick, Prabhanjan Ananth and Vinod Vaikuntanathan for many insightful discussions.  The author is also grateful for useful comments made by anonymous reviewers. The author
is partially supported by AFOSR YIP award number FA9550-16-1-0495 and the Institute for Quantum Information and Matter (an NSF Physics Frontiers Center; NSF Grant PHY-1733907), and
is also grateful for the hospitality of the Simons Institute for the Theory of Computing, where part
of this research was carried out.

\section{Preliminaries}
\label{sec:prelim}

\paragraph{Notation.}

For $x \in \mathbb{C}^n$, we denote the $\ell^2$ norm by $\| \vec x \|$. For $\vec x \in \mathbb{C}^n$, we occasionally also use the max norm $\| \vec x \|_\infty = \max_i |x_i|$.
We denote the expectation value of a random variable $X$ which takes values in $\algo X$ by $\mathbb{E}[X] = \sum_{x \in \algo X} x \Pr[ X = x]$.  The notation $x \rand \algo X$ denotes sampling of $x$ uniformly at random from $\algo X$, whereas $x \sim D$ denotes sampling of an element $x$ according to the distribution $D$. 
We call a non-negative real-valued function $\mu : \mathbb{N} \rightarrow \mathbb{R}^+$ negligible if $\mu(n) = o(1/p(n))$, for every polynomial $p(n)$. Given an integer $m\in \N$ and modulus $q \geq 2$, we represent elements in $\Z_q^m$ as integers $\Z^m \cap (-\frac{q}{2},\frac{q}{2}]^m$.

\subsection{Quantum computation}

For a comprehensive overview of quantum computation, we refer to the introductory texts \cite{NielsenChuang11,Wilde13}. We denote a finite-dimensional complex Hilbert space by $\mathcal{H}$, and we use subscripts to distinguish between different systems (or registers). For example, we let $\mathcal{H}_{A}$ be the Hilbert space corresponding to a system $A$. 
The tensor product of two Hilbert spaces $\algo H_A$ and $\algo H_B$ is another Hilbert space denoted by $\algo H_{AB} = \algo H_A \otimes \algo H_B$.
The Euclidean norm of a vector $\ket{\psi} \in \algo H$ over the finite-dimensional complex Hilbert space $\mathcal{H}$ is denoted as $\| \psi \| = \sqrt{\braket{\psi|\psi}}$. Let $L(\algo H)$
denote the set of linear operators over $\algo H$. A quantum system over the $2$-dimensional Hilbert space $\mathcal{H} = \mathbb{C}^2$ is called a \emph{qubit}. For $n \in \mathbb{N}$, we refer to quantum registers over the Hilbert space $\mathcal{H} = \big(\mathbb{C}^2\big)^{\otimes n}$ as $n$-qubit states. More generally, we associate \emph{qudits} of dimension $d \geq 2$ with a $d$-dimensional Hilbert space $\mathcal{H} = \mathbb{C}^d$. We use the word \emph{quantum state} to refer to both pure states (unit vectors $\ket{\psi} \in \mathcal{H}$) and density matrices $\rho \in \mathcal{D}(\mathcal{H)}$, where we use the notation $\mathcal{D}(\mathcal{H)}$ to refer to the space of positive semidefinite matrices of unit trace acting on $\algo H$. 
For simplicity, we frequently consider \emph{subnormalized states}, i.e. states in the space of positive semidefinite operators over $\algo H$ with trace norm not exceeding $1$, denoted by $\algo S_{\leq}(\algo H)$.
The \emph{trace distance} of two density matrices $\rho,\sigma \in \mathcal{D}(\mathcal{H)}$ is given by
$$
\| \rho - \sigma \|_\tr = \frac{1}{2} \Tr\left[ \sqrt{ (\rho - \sigma)^\dag (\rho - \sigma)}\right].
$$
We frequently use the compact notation $\rho \approx_\eps \sigma$ which means that there exists some $\eps \in [0,1]$ such that $\| \rho - \sigma \|_\tr \leq \eps$.
The \emph{purified distance} is defined as $P(\rho,\sigma) = \sqrt{1 - F(\rho,\sigma)^2}$, where $F(\rho,\sigma) = \| \sqrt{\rho}\sqrt{\sigma}\|_1$ denotes the fidelity.
A \textit{classical-quantum} (CQ) state $\rho \in \mathcal{D}(\mathcal{H}_{XB})$ depends on a classical variable in system $X$ which is correlated with a quantum system $B$. If the classical system $X$ is distributed according to a probability distribution $P_{\algo X}$ over the set $\algo X$, then all possible joint states $\rho_{XB}$ can be expressed as
\begin{align*}
\rho_{XB} = \sum_{x \in \algo X} P_{\algo X}(x) \proj{x}_X \otimes \rho_B^x.
\end{align*}

\paragraph{Quantum channels and measurements.}

A quantum channel $\Phi:  L(\algo H_A) \rightarrow L(\algo H_B)$ is a linear map between linear operators over the Hilbert spaces $\algo H_A$ and $\algo H_B$. Oftentimes, we use the compact notation $\Phi_{A \rightarrow B}$ to denote a quantum channel between $L(\algo H_A)$ and $L(\algo H_B)$. We say that a channel $\Phi$ is \emph{completely positive} if, for a reference system $R$ of arbitrary size, the induced map $\id_R \otimes \Phi$ is positive, and we call it \emph{trace-preserving} if $\Tr[\Phi(X)] = \Tr[X]$, for all $X \in L(\algo H)$. A quantum channel that is both completely positive and trace-preserving is called a quantum $\CPTP$ channel. 
Let $\algo X$ be a set. A \emph{generalized measurement} on a system $A$ is a set of linear operators $\{\vec M_A^x\}_{x \in \algo X}$ such that
$$
\sum_{x \in \algo X} \left(\vec M_A^x \right)^\dag\left(\vec M_A^x \right) = \id_A.
$$
We can represent a measurement as a $\CPTP$ map $\algo M_{A \rightarrow X}$ that maps states on system $A$ to measurement outcomes in a register denoted by $X$. For example, let $\rho \in \algo D(\algo H_{AB})$ be a bipartite state. Then,
$$
\algo M_{A \rightarrow X}: \quad \rho_{AB} \quad \mapsto \quad \sum_{x \in \algo X} \proj{x}_X \otimes \tr_A \left[\vec M_A^x \rho_{AB}{\vec M_A^x}^\dag \right],
$$
yields a normalized classical-quantum state. A positive-operator valued measure ($\POVM$) on a quantum system $A$ is a set of Hermitian positive semidefinite operators $\{\vec M_A^x\}_{x \in \algo X}$ such that
$$
\sum_{x \in \algo X} \vec M_A^x = \id_A.
$$
Oftentimes, we identify a $\POVM$ $\{\vec M_A^x\}_{x \in \algo X}$ with an associated generalized measurement $\{\sqrt{\vec M_A^x}\}_{x \in \algo X}$. The \emph{overlap} $\mathsf{c}$ of two $\POVM$s $\{\vec M_A^x\}_{x \in
\algo X}$ and $\{\vec N_A^y\}_{y \in \algo X}$ acting on a quantum system $A$ is defined by
$$ \mathsf{c} =  \max_{\substack{x,y}}   \left\| \sqrt{\vec M_A^x} \sqrt{\vec N_A^y}\right\|_\infty^2.$$
We say that two measurements are \emph{mutually unbiased}, if the overlap satisfies $\mathsf{c} =  1/d$, where $d = \dim(\algo H_A)$ is the dimension of the associated Hilbert space.

\paragraph{Quantum algorithms.} By a polynomial-time \textit{quantum algorithm} (or $\QPT$ algorithm) we mean a polynomial-time uniform family of quantum circuits given by $\mathcal{C} = \bigcup_{n \in \N} C_n$, where each circuit $C \in \algo C$ is described by a sequence of unitary gates and measurements. Similarly, we also define (classical) probabilistic polynomial-time $(\PPT)$ algorithms. A quantum algorithm may, in general, receive (mixed) quantum states as inputs and produce (mixed) quantum states as outputs. Occasionally, we restrict $\QPT$ algorithms implicitly. For example, if we write $\Pr[\mathcal{A}(1^{\lambda}) = 1]$ for a $\QPT$ algorithm $\mathcal{A}$, it is implicit that $\mathcal{A}$ is a $\QPT$ algorithm that outputs a single classical bit.

We extend the notion of $\QPT$ algorithms to $\CPTP$ channels via the following definition.

\begin{definition}[Efficient $\CPTP$ maps]\label{def:efficient-CPTP} A family of $\CPTP$ maps $\{\Phi_\lambda: L(\algo H_{A_\lambda}) \rightarrow L(\algo H_{B_\lambda}) \}_{\lambda \in \N}$ is called efficient, if there exists a polynomial-time uniformly generated family of circuits $\{ C_\lambda\}_{\lambda \in \N}$ acting on the Hilbert space $\algo H_{A_\lambda} \otimes \algo H_{B_\lambda} \otimes \algo H_{C_\lambda}$ such that, for all $\lambda \in \N$ and for all $\rho \in \algo H_{A_\lambda}$,
$$
\Phi_\lambda(\rho_\lambda) = \Tr_{A_\lambda C_\lambda}[C_\lambda (\rho_\lambda \otimes \ketbra{0}{0}_{B_\lambda C_\lambda})].
$$

\end{definition}

\begin{definition}
[Indistinguishability of ensembles of random variables] Let $\lambda \in N$ be a parameter. We say that two ensembles of random variables $X = \{X_\lambda\}$ and $Y = \{Y_\lambda\}$ are computationally indistinguishable, denoted by $X \approx_c Y$, if for all $\QPT$ distinguishers $\algo D$ which output a single bit, it holds that
$$ \big| \Pr[\mathcal{D}(1^\lambda,X_\lambda)=1] - \Pr[\mathcal{D}(1^\lambda,Y_\lambda)=1] \big| \leq \negl(\lambda) \,.$$
\end{definition}

\begin{definition}[Indistinguishability of ensembles of quantum states, \cite{10.1145/1132516.1132560}]
\label{def: indistinguishability}
Let $p: \mathbb{N} \rightarrow \mathbb{N}$ be a polynomially bounded function,
and let $\rho_\lambda$ and $\sigma_\lambda$
be $p(\lambda)$-qubit quantum states. We say that $\{\rho_{\lambda}\}_{\lambda \in \mathbb{N}}$ and $\{\sigma_\lambda\}_{\lambda \in \mathbb{N}}$ are quantum computationally indistinguishable ensembles of quantum states, denoted by $\rho_{\lambda} \approx_c \sigma_{\lambda}\,,$
if, for any $\QPT$ distinguisher $\mathcal{D}$ with single-bit output, any polynomially bounded $q: \mathbb{N} \rightarrow \mathbb{N}$, any family of $q(\lambda)$-qubit auxiliary states $\{\nu_{\lambda}\}_{\lambda \in \mathbb{N}}$, and every $\lambda \in \mathbb{N}$,
$$ \big| \Pr[\mathcal{D}(1^\lambda,\rho_{\lambda} \otimes \nu_{\lambda})=1] - \Pr[\mathcal{D}(1^\lambda,\sigma_{\lambda} \otimes \nu_{\lambda})=1] \big| \leq \negl(\lambda) \,.$$
%We say that $\algo D$ is a $(T,\eps)$ distinguisher if it runs in time $T(\lambda)$ and succeeds with probability at most $\eps(\lambda)$.
\end{definition}

\begin{lemma}["Almost As Good As New" Lemma, \cite{aaronson2016complexity}]\label{lem:almost} Let $\rho \in \algo D(\algo H)$ be a density matrix over a Hilbert space $\algo H$. Let $U$ be an arbitrary unitary and let $(\boldsymbol{\Pi}_0,\boldsymbol{\Pi}_1 = \id - \boldsymbol{\Pi}_0)$ be projectors acting on $\algo H \otimes \algo H_\aux$. We interpret $(U,\boldsymbol{\Pi}_0,\boldsymbol{\Pi}_1)$ as a measurement performed by appending an ancillary system in the state $\ketbra{0}{0}_\aux$, applying the unitary $U$ and subsequently performing the two-outcome measurement $\{\boldsymbol{\Pi}_0,\boldsymbol{\Pi}_1\}$ on the larger system. Suppose that the outcome corresponding to $\boldsymbol{\Pi}_0$ occurs with probability $1-\eps$, for some $\eps \in [0,1]$. In other words, it holds that $\Tr[\boldsymbol{\Pi}_0(U \rho \otimes \ketbra{0}{0}_\aux U^\dag)] = 1- \eps$. Then,
$$
\| \widetilde{\rho} - \rho \|_\tr \leq \sqrt{\eps},
$$
where $\widetilde{\rho}$ is the state after performing the measurement and applying $U^\dag$, and after tracing out $\algo H_\aux$:
$$
\widetilde{\rho} = \Tr_\aux\left[U^\dag \left( 
\boldsymbol{\Pi}_0 U (\rho \otimes \ketbra{0}{0}_\aux)U^\dag \boldsymbol{\Pi}_0 + \boldsymbol{\Pi}_1 U (\rho \otimes  \ketbra{0}{0}_\aux)U^\dag \boldsymbol{\Pi}_1 
\right)U \right].
$$
\end{lemma}

We also use the following lemma on the closeness to ideal states:

\begin{lemma}[\cite{Unruh2013}, Lemma 10]\label{lem:closeness-ideal}
Let $\boldsymbol{\Pi}$ be an arbitrary projector and let $\ket{\psi}$ be a normalized pure state such that $\|\boldsymbol{\Pi}\ket{\psi}\|^2 = 1 - \eps$, for some $\eps \geq 0$. Then, there exists a (pure) ideal state,
$$
\ket{\bar{\psi}} = \frac{\boldsymbol{\Pi}\ket{\psi}}{\|\boldsymbol{\Pi}\ket{\psi}\|},
$$
with the property that
$$
\| \proj{\psi} - \proj{\bar{\psi}}\|_\tr \leq \sqrt{\eps} \quad \,\, \text{ and } \quad \,\, \ket{\bar{\psi}} \in \mathrm{im}(\boldsymbol{\Pi}).
$$ 
In other words, the state $\ket{\bar{\psi}}$ is within trace distance $\eps>0$ of the state $\ket{\psi}$ and lies in the image of $\boldsymbol{\Pi}$.
\end{lemma}

We also use the following elementary lemma.

\begin{lemma}[\cite{coladangelo2020quantum}, Lemma 23]\label{lem:TD_inequalities}
Let $\rho,\sigma \in \algo D(\algo H)$ be two states with the property that $\|\rho - \sigma \|_{\mathrm{tr}} \leq \eps$, for some $\eps \geq 0$. Let $\boldsymbol{\Pi}$ be an arbitrary matrix acting on $\algo H$ such that $0 \leq \boldsymbol{\Pi} \leq \id$. Then,
$$
|\mathrm{Tr}[\boldsymbol{\Pi} \rho] - \mathrm{Tr}[\boldsymbol{\Pi} \sigma]|\leq \eps.
$$
\end{lemma}

\subsection{Classical and quantum entropies}\label{sec:entropies}

\paragraph{Classical entropies.}

Let $X$ be a random variable with an arbitrary distribution $P_{\algo X}$ over an alphabet $\algo X$. The \emph{min-entropy} of $X$, denoted by $\hmin(X)$, is defined by the following quantity
    $$
    \hmin(X) = - \log \left(\max_{x \in \algo X} \underset{X \sim P_{\algo X}}{\Pr}[X = x] \right).
    $$
The \emph{conditional min-entropy} of $X$ conditioned on a correlated random variable $Y$ is defined by
$$
\hmin(X|Y) = - \log\left( \underset{y \leftarrow Y}{\mathbb{E}} \Big[\max_{x \in \algo X} \underset{X \sim P_{\algo X}}{\Pr}[X = x | Y = y] \Big] \right).$$

\begin{lemma}[Leftover Hash Lemma, \cite{Hastad88pseudo-randomgeneration}]\label{lem:LHL}
Let $n,m \in \N$ and $q\geq 2$ a prime. Let $P$ be a distribution over $\Z_q^m$ and suppose that $\hmin(X) \geq n \log q + 2 \log(1/\eps) + O(1)$ for $\eps > 0$, where $X$ denotes a random variable with distribution $P$.
Then, the following two distributions are within total variance distance $\eps$:
$$
(\vec A, \vec A \cdot \vec x \Mod{q}) \quad \approx_\eps \quad (\vec A,\vec u): \quad\quad \vec A \rand \Z_q^{n \times m}, \,\,\vec u \rand \Z_q^n.
$$

\end{lemma}

\paragraph{Quantum entropies.}

\begin{definition}[Quantum min-entropy]\label{def:min-entropy}
Let $A$ and $B$ be two quantum systems and let $\rho_{AB} \in \algo S_{\leq}(\algo H_{AB})$ be any bipartite state. The min-entropy of $A$ conditioned on $B$ of the state $\rho_{AB}$ is defined as
$$
\hmin(A \, | \, B)_\rho = \underset{\sigma \in S_{\leq}(\algo H_{B}) }{\max} \sup \left\{ \lambda \in \mathbb{R} \ : \, \rho_{AB} \leq 2^{-\lambda} \id_A \otimes \sigma_B  \right\}.
$$
\end{definition}

\begin{definition}[Smooth quantum min-entropy]\label{def:smooth-entropies} Let $A$ and $B$ be quantum systems and let $\rho_{AB} \in \algo S_{\leq}(\algo H_{AB})$. Let $\eps \geq 0$. We define the $\eps$-smooth quantum min-entropy of $A$ conditioned on $B$ of
$\rho_{AB}$ as
\begin{align*}
  \hmin^\eps(A \, | \, B)_\rho \,\,&= \underset{\substack{\tilde{\rho}_{AB}\\
  P(\tilde{\rho}_{AB}, \rho_{AB}) \leq \eps}}{\sup}  \hmin(A \, | \, B)_{\tilde{\rho}}.
\end{align*}
\end{definition}

The conditional min-entropy of a CQ state $\rho_{XB}$ captures the difficulty of guessing the content of a classical register $X$ given quantum side information $B$. This motivates the following definition.

\begin{definition}[Guessing probability]\label{def:guessing}
Let $\rho_{XB} \in \algo D(\algo H_{X} \otimes \algo H_B)$ be a CQ state, where $X$ is a classical register over an alphabet $\algo X$ and $B$ is a quantum system. Then, the guessing probability of $X$ given $B$ is defined as
$$
p_{\guess}(X|B)_\rho = \underset{\vec M_B^x}{\sup} \sum_{x \in \algo X} \Pr[X=x]_\rho \cdot \Tr\left[\vec M_B^x \rho_{B} {\vec M_B^x}^\dag\right].
$$
\end{definition}

The following operational meaning of min-entropy is due to Koenig, Renner and Schaffner~\cite{Konig_2009}.

\begin{theorem}[\cite{Konig_2009}, Theorem 1]\label{thm:guessing}
Let $\rho_{XB} \in \algo D(\algo H_{X} \otimes \algo H_B)$ be a CQ state, where $X$ is a classical register over an alphabet $\algo X$ and $B$ is a quantum system. Then, it holds that
$$
\hmin(X \, | \, B)_\rho = - \log \left(p_{\guess}(X|B)_\rho \right).
$$
\end{theorem}

\subsection{Fourier analysis} \label{sec:fourier}

Let $q \geq 2$ be an integer modulus and let $m \in \N$. The \emph{$q$-ary (discrete) Fourier transform} takes as input a function $f : \Z^m \rightarrow \mathbb{C}$ and produces a function $\hat{f} : \Z_q^m \rightarrow \mathbb{C}$ (the Fourier transform of $f$) defined by
$$
\hat{f}(\vec y) = \sum_{\vec x \in \Z^m} f(\vec x) \cdot e^{\frac{2\pi i}{q} \ip{\vec y,\vec x}}.
$$
For brevity, we oftentimes write $\omega_q = e^{ \frac{2 \pi i}{q}} \in \mathbb{C}$ to denote the primitive $q$-th root of unity.
The $m$-qudit \emph{$q$-ary quantum Fourier transform} over the ring $\Z_q^m$ is defined by the operation,
$$
\FT_q : \quad \ket{\vec x} \quad \mapsto \quad \sqrt{q^{-m}} \displaystyle\sum_{\vec y \in \Z_q^n} e^{\frac{2\pi i}{q} \ip{\vec y,\vec x}} \ket{\vec y}, \quad\quad \forall \vec x \in \Z_q^m.
$$
It is well known that the $q$-ary quantum Fourier transform can be efficiently performed on a quantum computer for any modulus $q \geq 2$~\cite{892139}. Note the quantum Fourier transform of a normalized quantum state
$$
\ket{\Psi} = \sum_{\vec x \in \Z^m} f(\vec x) \ket{\vec x} \quad \text{ with } \quad \sum_{\vec x \in \Z^m} |f(\vec x)|^2 =1,
$$
for a function $f : \Z^m \rightarrow \mathbb{C}$, results in the state (the Fourier transform of $\ket{\Psi}$) given by
\begin{align*}
\FT_q\ket{\Psi} &= 
\sqrt{q^{-m}} \displaystyle\sum_{\vec y \in \Z_q^n} \Big(\sum_{\vec x \in \Z^m} f(\vec x) \cdot
e^{\frac{2\pi i}{q} \ip{\vec y,\vec x}}\Big) \ket{\vec y}\\
&=
\sqrt{q^{-m}} \displaystyle\sum_{\vec y \in \Z_q^n} \hat{f}(\vec y) \ket{\vec y}.
\end{align*}
Notice that the Fourier transform of $\ket{\Psi}$ is \emph{unitary} if $\mathrm{supp}(f) \subseteq \Z^m \cap (-\frac{q}{2},\frac{q}{2}]^m$. 
We frequently make use of the following standard identity for Fourier characters.

\begin{lemma}[Orthogonality of Fourier characters]\label{lem:orth}
Let $q \geq 2$ be any integer modulus and let $\omega_q = e^{ \frac{2 \pi i}{q}} \in \mathbb{C}$ denote the primitive $q$-th root of unity. Then, for arbitrary $x,y \in \Z_q$:
$$
\sum_{v \in \Z_q} \omega_q^{v \cdot x} \omega_q^{-v \cdot y} = q \,\, \delta_{ x, y}.
$$
\end{lemma}

\subsection{Generalized Pauli operators}

\begin{definition}[Generalized Pauli operators]
Let $q\geq 2$ be an integer modulus and $\omega_q = e^{2 \pi i/q}$ be the primitive $q$-th root of unity. The generalized $q$-ary Pauli operators $\{\vec{X}_q^b\}_{b \in \Z_q}$ and $\{\vec{Z}_q^b\}_{b \in \Z_q}$ are given by
\begin{align*}
\vec{X}_q^b &= \sum_{a \in \Z_q} \ket{a+b \Mod{q}}\bra{a}, \quad\text{ and }\\
\vec{Z}_q^b &= \sum_{a \in \Z_q} \omega_q^{a \cdot b} \ket{a}\bra{a}.
\end{align*}
For $\vec b = ( b_1,\dots,b_m) \in \Z_q^m$, we use the notation $\vec{X}_q^{\vec b} = \vec{X}_q^{b_1} \otimes  \dots \otimes \vec{X}_q^{b_m}$ and $\vec{Z}_q^{\vec b} = \vec{Z}_q^{b_1} \otimes  \dots \otimes \vec{Z}_q^{b_m}$.
\end{definition}

\begin{lemma}\label{lem:XZ-conjugation} Let $q\geq 2$ be an integer modulus. Then, for all $b \in \Z_q$, it holds that
\begin{align*}
\vec{Z}_q^{b} &= \FT_q \, \vec{X}_q^{b} \,\FT_q^\dag \\
\vec{X}_q^{b} &= \FT_q^\dag \, \vec{Z}_q^{b} \,\FT_q.
\end{align*}
\end{lemma}
\begin{proof} It suffices to show the first identity only as the second identity follows by conjugation with $\FT_q$. Using the orthogonality of Fourier characters over $\Z_q$ (\expref{Lemma}{lem:orth}), we find that
\begin{align*}
\vec{Z}_q^{b} &=
 \sum_{x \in \Z_q} \omega_q^{x \cdot b} \ketbra{x}{x}\\
&= \sum_{x,y' \in \Z_q} \omega_q^{x \cdot b} \left(\frac{1}{q}\sum_{a \in \Z_q}  \omega_q^{x \cdot a} \omega_q^{-a \cdot y'} \right) \ketbra{x}{y'}\\ 
&=\frac{1}{q}
\sum_{x,y \in \Z_q} \sum_{x',y' \in \Z_q} \sum_{a \in \Z_q}  \omega_q^{x \cdot y} \omega_q^{- x' \cdot y'} \braket{y | a + b \Mod{q}} \cdot\braket{a|x'} \, \ketbra{x}{y'} \\
&= \frac{1}{q}\left(\sum_{x, y \in \Z_q} \omega_q^{x \cdot y} \ketbra{x}{y} \right) \sum_{a \in \Z_q} \ket{a+b \Mod{q}}\bra{a} \left(\sum_{x',y' \in \Z_q} \omega_q^{- x' \cdot y'} \ketbra{\vec x'}{\vec y'} \right)\\ 
&=\FT_q \, \vec{X}_q^{b} \,\FT_q^\dag.
\end{align*}

\end{proof}

\begin{definition}[Pauli-$\vec Z$ dephasing channel] Let $q\geq 2$ be an integer modulus and let $m \in \N$. Let $\vec p$ be a probability distribution over $\Z_q^m$. Then, the Pauli-$\vec Z$ dephasing channel with respect to $\vec p$ is defined as
$$
\algo Z_{\vec p}(\rho) = \sum_{\vec z \in \Z_q^m} p_{\vec z} \, \vec Z_q^{\vec z} \rho \vec Z_q^{-\vec z}, \quad \,\, \forall \rho \in L((\mathbb{C}^q)^{\otimes m}).
$$
We use $\algo Z$ to denote the uniform Pauli-$\vec Z$ channel for which $\vec p$ is the uniform distribution over $\Z_q^m$.
\end{definition}

The following lemma shows that the uniform Pauli-$\vec Z$ channel on input $\rho$ returns a diagonal state which consists of diagonal elements
of $\rho$ encoded in the standard basis.
\begin{lemma}\label{lem:random-Z}
Let $q\geq 2$ be a modulus and $m \in \N$. Then, the uniform Pauli-$\vec Z$ dephasing channel satsifies,
$$
\algo Z(\rho) = q^{-m} \sum_{\vec z \in \Z_q^m} \, \vec Z_q^{\vec z} \rho \vec Z_q^{-\vec z} = \sum_{\vec x \in \Z_q^m} \Tr[\ketbra{\vec x}{\vec x} \rho] \,\ketbra{\vec x}{\vec x}, \quad \,\, \forall \rho \in L((\mathbb{C}^q)^{\otimes m}).
$$

\end{lemma}
\begin{proof}
Suppose that the state $\rho$ has the following form in the standard basis,
$$
\rho =\sum_{\vec x,\vec y \in \Z_q^m} \alpha_{\vec x,\vec y} \ketbra{\vec x}{\vec y} \,\,\in L((\mathbb{C}^q)^{\otimes m}).
$$
Using the orthogonality of Fourier characters over $\Z_q$ (\expref{Lemma}{lem:orth}), we obtain
\begin{align*}
\algo Z(\rho) &= q^{-m} \sum_{\vec z \in \Z_q^m} \, \vec Z_q^{\vec z} \rho \vec Z_q^{-\vec z}\\
&= q^{-m} \sum_{\vec z \in \Z_q^m}\sum_{\vec x,\vec y \in \Z_q^m} \alpha_{\vec x,\vec y} \, \vec Z_q^{\vec z} \ketbra{\vec x}{\vec y}\vec Z_q^{-\vec z}\\
&=  \sum_{\vec x,\vec y \in \Z_q^m}  \alpha_{\vec x,\vec y} \, 
\left( q^{-m}
\sum_{\vec z \in \Z_q^m}
\omega_q^{\ip{\vec x,\vec z}} \omega_q^{-\ip{\vec y,\vec z}}
\right)
 \ketbra{\vec x}{\vec y}\\
 &=\sum_{\vec x \in \Z_q^m} \alpha_{\vec x,\vec x} \ketbra{\vec x}{\vec x}\\
 &=\sum_{\vec x \in \Z_q^m} \Tr[\ketbra{\vec x}{\vec x} \rho] \ketbra{\vec x}{\vec x}.
\end{align*}
\end{proof}

\subsection{Lattices and the Gaussian mass}\label{sec:lattices}

A $\emph{lattice}$ $\Lambda \subset \mathbb{R}^m$ is a discrete subgroup of $\mathbb{R}^m$. To avoid handling matters of precision, we will only consider integer lattices $\Lambda \subseteq \Z^m$ throughout this work.
The \emph{dual} of a lattice $\Lambda \subset \mathbb{R}^m$, denoted by $\Lambda^*$, is the lattice of all vectors $y \in \mathbb{R}^m$ that satisfy $\ip{\vec y,\vec x} \in \Z$, for all vectors $\vec x \in \Lambda$. In other words, we define
$$
\Lambda^* = \left\{ \vec y \in \mathbb{R}^m \, : \, \ip{\vec y,\vec x} \in \Z, \text{ for all }  \vec x \in \Lambda\right\}.
$$
Given a lattice $\Lambda \subset \mathbb{R}^m$ and a vector $\vec t \in \mathbb{R}^m$, we define the coset with respect to $\vec t $ as the lattice shift $\Lambda - \vec t = \{\vec x \in \mathbb{R}^m :\, \vec x + \vec t \in \Lambda\}$. Note that many different shifts $\vec t$ can define the same coset.

The \emph{Gaussian measure} $\rho_\sigma$ with parameter $\sigma > 0$ is defined as the function
\begin{align*}
\rho_\sigma(\vec x) = \exp(-\pi \|\vec x \|^2/ \sigma^2), \quad \,\, \forall \vec x \in \mathbb{R}^m.    
\end{align*}
Let $\Lambda \subset \mathbb{R}^m$ be a lattice and let $\vec t \in \mathbb{R}^m$ be a shift. We define the \emph{Gaussian mass} of $\Lambda - \vec t$ as the quantity
\begin{align*}
\rho_\sigma(\Lambda - \vec t) = \sum_{\vec y \in \Lambda}\rho_\sigma(\vec y- \vec t).
\end{align*}

The \emph{discrete Gaussian distribution} $D_{\Lambda - \vec t,\sigma}$ is the distribution over the coset $\Lambda - \vec t$ that assigns probability proportional to $e^{-\pi \|\vec x - \vec t \|^2/ \sigma^2}$ for lattice points $\vec x \in \Lambda$. In other words, we have
$$
D_{\Lambda - \vec t,\sigma} (\vec x)= \frac{\rho_\sigma(\vec x - \vec t)}{\rho_\sigma(\Lambda - \vec t)}, \quad \,\, \forall \vec x \in \Lambda.
$$

We make use of the following tail bound for the Gaussian mass of a lattice \cite[Lemma 1.5 (ii)]{Banaszczyk1993}.

\begin{lemma}\label{lem:gaussian-tails}
For any $m$-dimensional lattice $\Lambda$ and shift $\vec t \in \mathbb{R}^m$ and for all $\sigma>0$, $c \geq (2 \pi)^{-\frac{1}{2}}$ it holds that
$$
\rho_\sigma \left( (\Lambda - \vec t)\setminus \algo B^m(\vec 0, c \sqrt{m} \sigma) \right) \leq (2 \pi e c^2)^{\frac{m}{2}} e^{- \pi c^2m} \rho_\sigma(\Lambda),
$$
where $B^m(\vec 0, s) = \{ \vec x \in \mathbb{R}^m \, : \, \|\vec x\|_2 \leq s\}$ denotes the $m$-dimensional ball of radius $s > 0$.
\end{lemma}

\paragraph{$q$-ary lattices.} In this work, we mainly consider \emph{$q$-ary lattices} $\Lambda$ that that satisfy $q\Z^m \subseteq \Lambda \subseteq \Z^m$, for some integer modulus $q \geq 2$. Specifically, we consider lattices generated by a matrix $\vec A \in \Z_q^{n \times m}$ for some $n,m \in \N$. The first lattice consists of all vectors which are perpendicular to the rows of $\vec A$, namely
$$
\Lambda_q^\bot(\vec A) = \{\vec x \in \Z^m: \, \vec A \cdot \vec x = \vec 0 \Mod{q}\}.
$$
Note that $\Lambda_q^\bot(\vec A)$ contains $q\Z^m$; in particular, it contains the identity $\vec 0 \in \Z^m$. For any \emph{syndrome} $\vec y \in \Z_q^n$ in the column span of $\vec A$, we also consider the lattice coset $\Lambda_q^{\vec y}(\vec A)$ given by
$$
\Lambda_q^{\vec y}(\vec A) = \{\vec x \in \Z^m: \, \vec A \cdot \vec x = \vec y \Mod{q}\} = \Lambda_q^\bot(\vec A) + \vec u,
$$
where $\vec u \in \Z^m$ is an arbitrary integer solution to the equation $\vec A \vec u = \vec y \Mod{q}$.

The second lattice is the lattice generated by $\vec A^T$ and is defined by
$$
\Lambda_q(\vec A) = \{\vec y \in \Z^m: \, \vec y = \vec A^T \cdot \vec s \Mod{q}, \text{ for some } \vec s \in \Z^n\}.
$$
The $q$-ary lattices  $\Lambda_q(\vec A)$ and $\Lambda_q^{\bot}(\vec A)$ are dual to each other (up to scaling). Specifically, we have
$$
q \cdot \Lambda_q^\bot(\vec A)^* =\Lambda_q(\vec A) \quad \text{ and } \quad q \cdot \Lambda_q(\vec A)^* = \Lambda_q^{\bot}(\vec A).
$$
Whenever $\vec A \in \Z_q^{n \times m}$ is full-rank, i.e. the subset-sums of the columns of $\vec A$ generate $\Z_q^n$, then
$\det(\Lambda_q^\bot(\vec A)) = q^n$.
We use the following facts due to Gentry, Peikert and Vaikuntanathan~\cite{cryptoeprint:2007:432}.

\begin{lemma}[\cite{cryptoeprint:2007:432}, Lemma 5.1]\label{lem:full-rank}
Let $n\in \N$ and let $q \geq 2$ be a prime modulus with $m \geq 2 n \log q$. Then, for all but a $q^{-n}$ fraction of $\vec A \in \Z_q^{n \times m}$, the subset-sums of the columns of $\vec A$ generate $\Z_q^n$. In other words, a uniformly random matrix $\vec A \rand \Z_q^{n \times m}$ is full-rank with overwhelming probability.
\end{lemma}

\begin{lemma}[\cite{cryptoeprint:2007:432}, Corollary 5.4]\label{lem:Gaussian-LHL}
Let $n\in \N$ and $q \geq 2$ be a prime with $m \geq 2 n \log q$. Then, for all but a $2q^{-n}$ fraction of $\vec A \in \Z_q^{n \times m}$ and $\sigma = \omega(\sqrt{\log m})$, the distribution of the syndrome $\vec A \cdot \vec e = \vec u \Mod{q}$ is within negligible total variation distance of the uniform distribution over $\Z_q^n$, where $\vec e \sim D_{\Z^m,\sigma}$.
\end{lemma}

The following lemma is a consequence of \cite[Lemma 4.4]{1366257} and \cite[Lemma 5.3]{cryptoeprint:2007:432}.

\begin{lemma}\label{lem:tailboundII}
Let $n\in \N$ and let $q \geq 2$ be a prime modulus with $m \geq 2 n \log q$. Let $\vec A \in \Z_q^{n \times m}$ be a matrix whose columns generate $\Z_q^n$. Then, for any $\sigma = \omega(\sqrt{\log m})$ and for any syndrome $\vec y \in \Z_q^n$:
$$
\Pr_{\vec x \sim D_{\Lambda_q^{\vec y}(\vec A),\sigma}} \Big[ \|\vec x\| \geq \sqrt{m}\sigma \Big] \leq \negl(n).
$$
\end{lemma}

\begin{definition}[Periodic Gaussian]
Let $m \in \N$, let $q \geq 2$ be a modulus and let $\sigma >0$. The $q$-periodic Gaussian $\rho_{\sigma,q}$ function is the periodic continuation of the Gaussian measure $\rho_\sigma$, where 
$$
\rho_{\sigma,q}(\vec x) = \rho_\sigma (\vec x + q \Z^m), \quad \forall \vec x \in \mathbb{R}^m.
$$
\end{definition}

For any function $f: \Z^m \rightarrow \mathbb{C}$ and lattice $\Lambda \subseteq \Z^m$, the well-known \emph{Poisson summation formula} states that $f(\Lambda) = \det(\Lambda^*) \hat{f}(\Lambda^*)$.
We use the following variant of the formula which applies to $q$-ary lattices.

\begin{lemma}[Poisson summation formula for $q$-ary lattices]\label{lem:poisson}
Let $q$ be a prime modulus and let $\vec A \in \Z_q^{n \times m}$ be any matrix whose columns generate $\Z_q^n$. Let $\vec v,\vec w \in \Z_q^m$ and $\sigma >0$ be arbitrary. Then, it holds that
$$
\sum_{\vec x \in \Lambda_q^{\vec v}(\vec A)} \rho_{\sigma}(\vec x) \cdot e^{-\frac{2\pi i}{q}\ip{\vec w,\vec x}} =
\frac{\sigma^m}{q^n} \cdot \sum_{\vec y \in \Z_q^n} \rho_{q/\sigma,q}(\vec w + \vec y \vec A) \cdot e^{\frac{2\pi i}{q}\ip{\vec y,\vec v}}.
$$
\end{lemma}
\begin{proof}
Because $\vec A \in \Z_q^{n \times m}$ is full-rank, it holds that
$\det(\Lambda_q^\bot(\vec A)) = q^n$. Let $\Lambda_q^{\vec v}(\vec A)$ be the lattice coset given by $\Lambda_q^\bot(\vec A) + \vec u$,
for some $\vec u \in \Z^m$ with $\vec A \cdot\vec u = \vec v\Mod{q}$. By the Poisson summation formula,
\begin{align*}
\sum_{\vec x \in \Lambda_q^{\vec v}(\vec A)} \rho_{\sigma}(\vec x) \cdot e^{-\frac{2\pi i}{q}\ip{\vec w,\vec x}} &= \sum_{\vec x \in \Lambda_q^\bot(\vec A) + \vec u} \rho_{\sigma}(\vec x) \cdot e^{-\frac{2\pi i}{q}\ip{\vec w,\vec x}}\\
&= \frac{\sigma^m}{\det(\Lambda_q^\bot(\vec A))} \sum_{\vec y  \in \frac{1}{q} \Lambda_q(\vec A)} \rho_{1/\sigma} (\vec w + q \cdot \vec y) \cdot e^{2 \pi i \ip{\vec y,\vec u}} \\
&= \frac{\sigma^m}{q^n} \sum_{\vec y  \in \Lambda_q(\vec A)} \rho_{q/\sigma} (\vec w +  \vec y) \cdot e^{\frac{2 \pi i}{q} \ip{\vec y,\vec u}} \\
&= \frac{\sigma^m}{q^n} \sum_{\vec y  \in \Z_q^n} \rho_{q/\sigma} (\vec w + \vec y \vec A + q \cdot \Z^m) \cdot e^{\frac{2 \pi i}{q} \ip{\vec A^\intercal \vec y,\vec u}} \\
&= \frac{\sigma^m}{q^n} \cdot \sum_{\vec y \in \Z_q^n} \rho_{q/\sigma,q}(\vec w + \vec y \vec A) \cdot e^{\frac{2\pi i}{q}\ip{\vec y,\vec v}}.
\end{align*}
\end{proof}

For $\vec x \in \Z^m$, let $[\vec x]_q$ denote the unique representative $\bar{\vec x} \in \Z^m \cap (-\frac{q}{2},\frac{q}{2}]^m$ such that $\vec x \equiv \bar{\vec x} \Mod{q}$. The following lemma due to Brakerski~\cite{Brakerski18} says that, whenever $\sigma$ is much smaller than the modulus $q$, the periodic Gaussian $\rho_{\sigma,q}$ is close to the non-periodic (but truncated) Gaussian. 
\begin{lemma}[\cite{Brakerski18}, Lemma 2.6]\label{lem:periodic-vs-truncated-Gaussian}
Let $q\geq 2$, $\vec x \in \Z^m$ such that $\|[\vec x]_q \| < q/4$ and
$\sigma > 0$. Then,
$$
1 \,\leq \, \frac{\rho_{\sigma,q}(\vec x)}{\rho_\sigma([\vec x]_q)} \,\leq \, 1 + 2^{-(\frac{1}{2}(q/\sigma)^2 -m)}.
$$
\end{lemma}

A simple consequence of the tail bound in \expref{Lemma}{lem:gaussian-tails} is that the discrete Gaussian $D_{\Z^m,\sigma}$ distribution is essentially only supported on the finite set $\{\vec x \in \Z^m : \|\vec x\| \leq \sigma \sqrt{m}\}$, which suggests the use of \emph{truncation}.
Given a modulus $q \geq 2$ and $\sigma >0$, we define the \emph{truncated} discrete Gaussian distribution $D_{\Z_q^m,\sigma}$ over the finite set $\Z^m \cap (-\frac{q}{2},\frac{q}{2}]^m$ with support $\{\vec x \in \Z_q^m : \|\vec x\| \leq \sigma \sqrt{m}\}$ as the density
$$
D_{\Z_q^m,\sigma}(\vec x) = \frac{\rho_\sigma(\vec x)}{\displaystyle\sum_{\vec y \in \Z_q^m,\|\vec y\| \leq \sigma\sqrt{m} } \rho_\sigma(\vec y)}
$$
We define the analogous \emph{periodic} discrete Gaussian distribution $D_{\Z_q^m,\sigma,q}$ as
$$
D_{\Z_q^m,\sigma,q}(\vec x) = \frac{\rho_{\sigma,q}(\vec x)}{\displaystyle\sum_{\vec y \in \Z_q^m,\|\vec y\| \leq \sigma\sqrt{m} } \rho_{\sigma,q}(\vec y)}
$$

\begin{lemma}\label{lem:TD-periodic-truncated}
Let $m \in \N$, $q \geq 2$ a modulus and let $\sigma \in (0,q/\sqrt{8m})$. Consider the quantum states,
$$
\ket{\psi} = \sum_{\vec x \in \Z_q^m} \sqrt{D_{\Z_q^m,\sigma}(\vec x)} \ket{\vec x} \quad\quad \text{and} \quad\quad \ket{\phi} = \sum_{\vec x \in \Z_q^m} \sqrt{D_{\Z_q^m,\sigma,q}(\vec x)} \ket{\vec x}.
$$
Then, it holds that
$$
\left\| \ketbra{\psi}{\psi} - \ketbra{\phi}{\phi} \right\|_\tr \,\leq \, \sqrt{1 - \left(1 + 2^{-(\frac{1}{2}(q/\sigma)^2 -m)}\right)^{-1}}.
$$
\end{lemma}
\begin{proof} We first bound the Hellinger distance,
\begin{align}
H^2(D_{\Z_q^m,\sigma},D_{\Z_q^m,\sigma,q}) &= 1 - \sum_{\vec x \in \Z_q^m} \sqrt{D_{\Z_q^m,\sigma}(\vec x) \cdot D_{\Z_q^m,\sigma,q}(\vec x)}.
\end{align}
To this end, we define two normalization factors
\begin{align}
Z_{\sigma}  = \sum_{\vec y \in \Z_q^m,\|\vec y\| \leq \sqrt{m}\sigma } \rho_\sigma(\vec y) \quad\quad \text{ and } \quad\quad Z_{\sigma,q} = \sum_{\vec y \in \Z_q^m,\|\vec y\| \leq \sqrt{m}\sigma } \rho_{\sigma,q}(\vec y). 
\end{align}
From \expref{Lemma}{lem:periodic-vs-truncated-Gaussian}, it follows for any $\vec x \in \Z^m \cap (-\frac{q}{2},\frac{q}{2}]^m$ with $\| \vec x \| < q/4$ that
\begin{align}\label{eq:periodic-vs-truncated}
\rho_{\sigma,q}^2(\vec x) \cdot \left(1 + 2^{-(\frac{1}{2}(q/\sigma)^2 -m)}\right)^{-1} \leq \,\, \rho_\sigma(\vec x) \cdot \rho_{\sigma,q}(\vec x).
\end{align}
Recall also that the truncated discrete Gaussian is supported on the finite set
$$
\supp(D_{\Z_q^m,\sigma}) = \{\vec x \in \Z_q^m : \|\vec x\| \leq \sqrt{m} \sigma\}.
$$
Plugging in Eq.~\eqref{eq:periodic-vs-truncated}, we can bound the Hellinger distance as follows:
\begin{align*}
H^2(D_{\Z_q^m,\sigma},D_{\Z_q^m,\sigma,q}) &= 1 - \sum_{\vec x \in \Z_q^m} \sqrt{D_{\Z_q^m,\sigma}(\vec x) \cdot D_{\Z_q^m,\sigma,q}(\vec x)}\\
&= 1 - \sqrt{Z_\sigma^{-1} \cdot Z_{\sigma,q}^{-1}} \sum_{\vec x \in \Z_q^m,\|\vec x\| \leq \sqrt{m}\sigma } \sqrt{\rho_{\sigma}(\vec x) \cdot \rho_{\sigma,q}(\vec x)}\\
&\leq 1 - \sqrt{\frac{Z_\sigma^{-1} \cdot Z_{\sigma,q}^{-1}}{ 1 + 2^{-(\frac{1}{2}(q/\sigma)^2 -m)}}} \sum_{\vec x \in \Z_q^m,\|\vec x\| \leq \sqrt{m}\sigma } \rho_{\sigma,q}(\vec x)\\
&\leq  1 -  \left(1 + 2^{-(\frac{1}{2}(q/\sigma)^2 -m)}\right)^{-1/2}.
\end{align*}
Therefore, it holds that
\begin{align*}
\left\| \ketbra{\psi}{\psi} - \ketbra{\phi}{\phi} \right\|_\tr \,&\leq \, \sqrt{1 - (1  -H^2(D_{\Z_q^m,\sigma},D_{\Z_q^m,\sigma,q}))^2}\\
&\leq  \sqrt{1 - \left(1 + 2^{-(\frac{1}{2}(q/\sigma)^2 -m)}\right)^{-1}}.
\end{align*}
\end{proof}

The following result allows us to bound the total variation distance between a truncated discrete Gaussian $D_{\Z_q^m,\sigma}$ and its perturbation by a fixed vector $\vec e_0 \in \Z^m$.

\begin{lemma}[\cite{brakerski2021cryptographic}, Lemma 2.4]\label{lem:shifted-gaussian}
Let $q \geq 2$ be a modulus, $m \in \N$ and $\sigma > 0$. Then, for any $\vec e_0 \in \Z^m$,
$$
\| D_{\Z_q^m,\sigma} - (D_{\Z_q^m,\sigma} + \vec e_0)\|_{\mathsf{TV}} \leq 2 \cdot \Big(1 - e^{\frac{-2 \pi \sqrt{m} \| \vec e_0 \|}{\sigma}} \Big).
$$
\end{lemma}

\subsection{Cryptography}

In this section, we review several definitions in cryptography.

\paragraph{Public-key encryption.}

\begin{definition}[Public-key encryption] A public-key encryption $(\PKE)$ scheme $\Sigma = (\KeyGen,\Enc,\Dec)$ with plaintext space $\algo M$ is a triple of $\QPT$ algorithms consisting of a key generation algorithm $\KeyGen$, an encryption algorithm $\Enc$, and a decryption algorithm $\Dec$.
\begin{description}
\item $\KeyGen(1^\lambda) \rightarrow (\pk,\sk):$ takes as input the parameter $1^\lambda$ and outputs a public key $\pk$ and secret key $\sk$.
\item $\Enc(\pk,m) \rightarrow \ct:$ takes as input the public key $\pk$ and a plaintext $m \in \algo M$, and outputs a ciphertext $\ct$.
\item $\Dec(\sk,\ct) \rightarrow m'\, \mathbf{or}\,\bot:$ takes as input the secret key $\sk$ and ciphertext $\ct$, and outputs $m'\in \algo M$ or $\bot$.
\end{description}
\end{definition}

\begin{definition}[Correctness of $\PKE$] For any $\lambda \in \mathbb{N}$, and for any $m \in \algo M$:
$$
\Pr \left[ \Dec(\sk,\ct) \neq m \, \bigg| \, \substack{
(\pk,\sk) \leftarrow \KeyGen(1^\lambda)\\
\ct \leftarrow \Enc(\pk,m)
}\right] \, \leq \, \negl(\lambda).
$$
\end{definition}

\begin{definition}[\INDCPA security]\label{def:ind-cpa} Let $\Sigma = (\KeyGen,\Enc,\Dec)$ be a $\PKE$ scheme and let $\algo A$ be a $\QPT$ adversary. We define the security experiment $\Exp^{\mathsf{ind\mbox{-}cpa}}_{\Sigma,\algo A,\lambda}(b)$ between $\algo A$ and a challenger as follows:
\begin{enumerate}
    \item The challenger generates a pair $(\pk,\sk) \from \KeyGen(1^\lambda)$, and sends $\pk$ to $\algo A$.
    \item $\algo A$ sends a plaintext pair $(m_0,m_1) \in \algo M \times \algo M$ to the challenger.
    \item The challenger computes $\ct_b \leftarrow \Enc(\pk,m_b)$, and sends $\ct_b$ to $\algo A$.
    \item $\algo A$ outputs a bit $b' \in \bit$, which is also the output of the experiment.
\end{enumerate}
We say that the scheme $\Sigma$ is $\INDCPA$-secure if, for any $\QPT$ adversary $\algo A$, it holds that
$$
\Adv_{\Sigma,\algo A}(\lambda) := |\Pr[\Exp^{\mathsf{ind\mbox{-}cpa}}_{\Sigma,\algo A,\lambda}(0)=1] - \Pr[\Exp^{\mathsf{ind\mbox{-}cpa}}_{\Sigma,\algo A,\lambda}(1) = 1] |
 \leq \negl(\lambda).$$
\end{definition}

\subsection{The Short Integer Solution problem}

The (inhomogenous) $\SIS$ problem was introduced by Ajtai~\cite{DBLP:conf/stoc/Ajtai96} in his seminal work on average-case lattice problems. The problem is defined as follows. 

\begin{definition}[Inhomogenous SIS problem,\cite{DBLP:conf/stoc/Ajtai96}]\label{def:ISIS} Let $n,m \in \N$ be integers, let $q\geq 2$ be a modulus and let $\beta >0$ be a parameter. The Inhomogenous Short Integer Solution problem $(\ISIS)$ problem is to find a short solution $\vec x \in \Z^m$ with $\|\vec x\|_2 \leq \beta$ such that $\vec A \cdot \vec x = \vec y \Mod{q}$ given as input a tuple $(\vec A \rand \Z_q^{n \times m},\vec y \rand \Z_q^n)$.
%We say that an algorithm solves the $\ISIS_{n,m,q,\beta}$ problem if it runs in (classical or quantum) time $\poly(n \log q)$ and finds a solution with probability at least $1/\poly(n \log q)$.
The Short Integer Solution $(\SIS)$ problem is a homogenous variant of $\ISIS$ with input $(\vec A \rand \Z_q^{n \times m},\vec 0 \in\Z_q^n)$.
\end{definition}

Micciancio and Regev~\cite{DBLP:journals/siamcomp/MicciancioR07} showed that the average-case $\SIS$ problem is as hard as approximating worst-case lattice problems to within small factors. Gentry, Peikert and Vaikuntanathan~\cite{cryptoeprint:2007:432} subsequently gave an improved reduction showing that, for any $m=\poly(n)$, $\beta=\poly(n)$, and prime $q \geq \beta \cdot \omega(\sqrt{n \log q})$, the average-case problems $\SIS_{n,m,q,\beta}$ and $\ISIS_{n,m,q,\beta}$ are as hard as approximating the shortest independent vector problem $(\mathsf{SIVP})$ problem in the
worst case to within a factor $\gamma = \beta \cdot \tilde{O}(\sqrt{n})$.

\subsection{The Learning with Errors problem}

The \emph{Learning with Errors} problem was introduced by Regev~\cite{Regev05} and serves as the primary basis of hardness of post-quantum cryptosystems. The problem is defined as follows. 

\begin{definition}[``Search'' $\LWE$,\cite{Regev05}]\label{def:search-LWE} Let $n,m \in \N$ be integers, let $q\geq 2$ be a modulus and let $\alpha \in (0,1)$ be a parameter. The Learning with Errors $(\LWE)$ problem is to find a secret vector $\vec s$ given as input a sample $(\vec A, \vec s\vec A + \vec e \Mod{q})$ from the distribution $\LWE_{n,q,\alpha q}^m$, where $\vec A \rand \Z_q^{n \times m}$ and $\vec s \rand \Z_q^n$ are uniformly random, and where $\vec e \sim D_{\Z^m, \alpha q}$ is sampled from the discrete Gaussian distribution.
%We say that an algorithm solves the ("search") $\LWE_{n,q,\alpha q}^m$ problem if it runs in (classical or quantum) time $\poly(n \log q)$ and finds $\vec s$ with probability at least $1/\poly(n \log q)$.
\end{definition}

\begin{definition}[``Decisional'' $\LWE$,\cite{Regev05}]\label{def:decisional-lwe} Let $n,m \in \N$ be integers, let $q\geq 2$ be a modulus and let $\alpha \in (0,1)$ be a parameter. The ``decision'' Learning with Errors $(\mathsf{DLWE})$ problem is to distinguish between
$$
(\vec A \rand \Z_q^{n \times m},\vec s\vec A+ \vec e \Mod{q}) \quad \text{ and } \quad (\vec A \rand \Z_q^{n \times m},\vec u \rand \Z_q^m),\,\,
$$
where $\vec s \rand  \Z_q^n$ is uniformly random and where $\vec e \sim D_{\Z^m,\alpha q}$ is a discrete Gaussian noise vector.
%We say that an algorithm solves the $\mathsf{DLWE}_{n,q,\alpha q}^m$ problem if it runs in (classical or quantum) time $\poly(n \log q)$ and succeeds with probability at least $\frac{1}{2} + 1/\poly(n \log q)$.
\end{definition}

As shown in \cite{Regev05}, the $\LWE_{n,q,\alpha q}^m$ problem with parameter $\alpha q \geq 2 \sqrt{n}$ is at least as hard as approximating the shortest independent vector problem $(\mathsf{SIVP})$ to within a factor of $\gamma = \widetilde{O}(n / \alpha)$ in worst case lattices of dimension $n$. In this work we assume the subexponential hardness of $\LWE_{n,q,\alpha q}^m$ which relies on the worst case hardness of approximating short vector problems in
lattices to within a subexponential factor.

\section{Primal and Dual Gaussian States}

Our Dual-Regev-type encryption schemes with certified deletion in \expref{Section}{sec:Dual-Regev-PKE} and \expref{Section}{sec:Dual-Regev-FHE} rely on two types of Gaussian superpositions, which we call \emph{primal} and \emph{dual} Gaussian states. The former (i.e., primal) state corresponds to a quantum superposition of $\LWE$ samples with respect to a matrix $\vec A \in \Z_q^{n \times m}$, and (up to a phase) can be thought of as a superposition of Gaussian balls around random lattice vectors in $\Lambda_q(\vec A)$. The latter (i.e., dual) state corresponds to a Gaussian superposition over a particular coset,
$$\Lambda_q^{\vec y}(\vec A) = \{ \vec x \in \Z^m  :  \vec A \cdot \vec x = \vec y \Mod{q}\},$$
of the $q$-ary lattice $\Lambda_q^\bot(\vec A) = \{\vec x \in \Z^m: \, \vec A \cdot\vec x = \vec 0 \Mod{q}\}$ defined in \expref{Section}{sec:lattices}.

Our terminology regarding which state is primal and which state is dual is completely arbitrary. In fact, the $q$-ary lattices $\Lambda_q(\vec A)$ and $\Lambda_q^{\bot}(\vec A)$ are both dual to each other (up to scaling), and satisfy
$$
q \cdot \Lambda_q^\bot(\vec A)^* =\Lambda_q(\vec A) \quad \text{ and } \quad q \cdot \Lambda_q(\vec A)^* = \Lambda_q^{\bot}(\vec A).
$$
We choose to refer to the quantum superposition of $\LWE$ samples as the \emph{primal} Gaussian state because it corresponds directly to the ciphertexts of our encryption scheme, whereas the \emph{dual} Fourier mode is only used in order to prove deletion. 
%Gaussian superpositions first appeared in Regev's quantum reduction from worst-case lattice problems to $\LWE$, and have also been used by Stehlé et al.~\cite{cryptoeprint:2009/285} who gave a quantum reduction from $\SIS$ to $\LWE$. Roberts~\cite{Roberts19} and Kitagawa et al.~\cite{kitagawa2021secure} used similar (dual) Gaussian states to construct quantum money and secure software leasing schemes. 
%Various other forms of superpositions of $\LWE$ samples have been considered by Grilo, Kerenidis and Zijlstra~\cite{Grilo_2019} in the context of quantum learning theory and by Alagic, Jeffery, Ozols and Poremba~\cite{cryptography4010010}, as well as by Chen, Liu and Zhandry~\cite{chen2021quantum}, in the context of quantum cryptanalysis of $\LWE$-based cryptosystems.
We define primal and dual Gaussian states as follows.

\begin{definition}[Gaussian states]\label{def:Gaussian-states}
Let $m \in \N$, $q \geq 2$ be an integer modulus and $\sigma >0$. Then,
\begin{itemize}
\item (primal Gaussian state:) for all $\vec A \in \Z_q^{n \times m}$ and $\vec y \in \Z_q^m$, we let
$$
\ket{\psi_{\vec A,\vec y}} \,\,= \sum_{\vec s \in \Z_q^n} \sum_{\vec e \in \Z_q^m} \rho_{q/\sigma}(\vec e) \, \omega_q^{-\ip{\vec s,\vec y}} \ket{\vec s \vec A + \vec e \Mod{q}};
$$

    \item (dual Gaussian state:) for all $\vec A \in \Z_q^{n \times m}$ and $\vec y \in \Z_q^m$, we let
$$    
\ket{\hat \psi_{\vec A,\vec y}} \,\,= \sum_{\substack{\vec x \in \Z_q^m\\ \vec A \vec x = \vec y \Mod{q}}}\rho_{\sigma}(\vec x) \ket{\vec x}. \quad\quad\quad\quad\quad\quad\quad\quad\quad
$$    
\end{itemize}
For simplicity, we oftentimes drop the subscript on $\vec A$ and write $\ket{\psi_{\vec y}}$ and $\ket{\hat \psi_{\vec y}}$, respectively.
\end{definition}

\subsection{Duality lemma}

The following lemma states that, up to negligible trace distance, the primal and dual Gaussian states in \expref{Definition}{def:Gaussian-states} are related via the $q$-ary quantum Fourier transform.

\begin{lemma}[Duality lemma]\label{lem:switching}
Let $m \in \N$, $q \geq 2$ be a prime modulus and let $\sigma \in (\sqrt{8m},q/\sqrt{8m})$.
Let $\vec A \in \Z_q^{n \times m}$ be a matrix whose columns generate $\Z_q^n$ and let $\vec y \in \Z_q^n$ be arbitrary. Then, up to negligible trace distance, the primal and dual Gaussian states are related via the quantum Fourier transform:
\begin{align*}
\FT_q \ket{ \psi_{\vec y}} \quad &\approx_\eps \quad \ket{\hat\psi_{\vec y}} = \sum_{\substack{\vec x \in \Z_q^m\\ \vec A \vec x = \vec y \Mod{q}}}\rho_{\sigma}(\vec x) \ket{\vec x};\\
\vspace{2mm}
\FT_q^\dag \ket{\hat\psi_{\vec y}} \quad &\approx_\eps \quad \ket{ \psi_{\vec y}} = \sum_{\vec s \in \Z_q^n} \sum_{\vec e \in \Z_q^m} \rho_{q/\sigma}(\vec e) \, \omega_q^{-\ip{\vec s,\vec y}} \ket{\vec s \vec A + \vec e \Mod{q}},
\end{align*}
where $\eps: \N \rightarrow \mathbb{R}^+$ is a negligible function in the parameter $m \in \N$.
\end{lemma}
\begin{proof}
Let $\vec y \in \Z_q^n$ be an arbitrary vector and recall that the dual Gaussian coset $\ket{\hat\psi_{\vec y}}$ is given by
\begin{align}
\ket{\hat\psi_{\vec y}} \,\,= \sum_{\substack{\vec x \in \Z_q^m\\ \vec A \vec x = \vec y \Mod{q}}}\rho_{\sigma}(\vec x) \ket{\vec x}.
\label{eq:primal-state}
\end{align}
We denote by $\Lambda_q^{\vec y}(\vec A) = \{ \vec x \in \Z^m  :  \vec A \vec x = \vec y \Mod{q}\}$ be the associated coset of the lattice $\Lambda_q^\bot(\vec A)$. Consider now the Gaussian superposition over the entire lattice coset $\Lambda_q^{\vec y}(\vec A)$ formally defined by
\begin{align}\label{eq:full-coset-state}
\ket{\hat\phi_{\vec y}} \,\,=
\sum_{\vec x \in \Lambda_q^{\vec y}(\vec A)}\rho_{\sigma}(\vec x) \ket{\vec x}.
\end{align}
Since $\sigma < q/\sqrt{8m}$, it follows from the tail bound in \expref{Lemma}{lem:tailboundII} that the state in \eqref{eq:primal-state} is within negligible trace distance of the state in Eq.~\eqref{eq:full-coset-state}.
Applying the (inverse) quantum Fourier transform, we get
\begin{align}
\ket{\phi_{\vec y}} \, \overset{\text{def}}{=} \, \FT_q^\dag \ket{\hat\phi_{\vec y}}
= \sum_{\vec z \in \Z_q^m} \Big(\sum_{\vec x \in \Lambda_q^{\vec y}(\vec A)} \rho_\sigma(\vec x) \cdot
\omega_q^{-\ip{\vec x,\vec z}}\Big) \ket{\vec z}.
\end{align}
From the Poisson summation formula (\expref{Lemma}{lem:poisson}) and a subsequent change of variables, it follows that
\begin{align}
\ket{\phi_{\vec y}}
&= \sum_{\vec z \in \Z_q^m} \Big(\sum_{\vec s \in \Z_q^n} \rho_{q/\sigma,q}(\vec z + \vec s \vec A) \cdot
\omega_q^{\ip{\vec s,\vec y}} \Big)\ket{\vec z} \nonumber\\
&= \sum_{\vec s \in \Z_q^n} \sum_{\vec e \in Z_q^m} \rho_{q/\sigma,q}(\vec e) \cdot
\omega_q^{-\ip{\vec s,\vec y}} \ket{\vec s \vec A + \vec e \Mod{q}}.
\end{align}
Because $\sigma > \sqrt{8m}$ it follows from \expref{Lemma}{lem:TD-periodic-truncated} that there exists
$$
\kappa(m) =  \sqrt{1 - \left(1 + 2^{-3m}\right)^{-1}} \, \geq 0
$$
such that 
\begin{align}
\ket{\phi_{\vec y}} \approx_\kappa \sum_{\vec s \in \Z_q^n} \sum_{\vec e \in Z_q^m} \rho_{q/\sigma}(\vec e) \cdot
\omega_q^{-\ip{\vec s,\vec y}} \ket{\vec s \vec A + \vec e \Mod{q}}.
\end{align}
Putting everything together, it follows from the triangle inequality that
$$
\FT_q^\dag \ket{\hat\psi_{\vec y}} \quad \approx_\eps \quad \ket{\psi_{\vec y}} = \sum_{\vec s \in \Z_q^n} \sum_{\vec e \in \Z_q^m} \rho_{q/\sigma}(\vec e) \, \omega_q^{-\ip{\vec s,\vec y}} \ket{\vec s \vec A + \vec e},
$$
where $\eps(m) = \negl(m) + \kappa(m)$. Using that
$\sqrt{1-1
/(1+x)} \leq \sqrt{x}$ for all $x > 0$, we have
\begin{align*}
\eps(m) &=   \negl(m) + \sqrt{1 - \left(1 + 2^{-3m}\right)^{-1}}\\
&\leq\negl(m) + 2^{-\frac{3m}{2}}.
\end{align*}
Thus, we have that $\eps(m) \leq \negl(m)$.
This proves the claim.
\end{proof}

\begin{corollary}\label{cor:switching}
Let $m \in \N$, $q \geq 2$ be a prime and $\sigma \in (\sqrt{8m},q/\sqrt{8m})$.
Let $\vec A \in \Z_q^{n \times m}$ be a matrix whose columns generate $\Z_q^n$ and let $\vec y \in \Z_q^n$ be arbitrary. Then, there exists a negligible function
$\eps(m)$ such that
\begin{align*}
\FT_q \vec X_q^{\vec v} \ket{ \psi_{\vec y}} \quad &\approx_\eps \quad
\vec Z_q^{\vec v} \ket{\hat\psi_{\vec y}},\quad\quad \forall \vec v \in \Z_q^m.
\end{align*}
\end{corollary}
\begin{proof} From \expref{Lemma}{lem:XZ-conjugation} it follows that $\FT_q \vec X_q^{\vec v} = \vec Z_q^{\vec v} \FT_q$, for all $\vec v \in \Z_q^m$. Moreover, \expref{Lemma}{lem:switching} implies that $\FT_q \ket{\psi_{\vec y}}$ is within negligible trace distance of $\ket{\hat\psi_{\vec y}}$. This proves the claim.
\end{proof}

\subsection{Efficient state preparation}

In this section, we give two algorithms that prepare the \emph{primal} and \emph{dual} Gaussian states from \expref{Definition}{def:Gaussian-states}.
We remark that Gaussian superpositions over $\Z_q^m$ with parameter $\sigma = \Omega(\sqrt{m})$ can be efficiently implemented using standard quantum state preparation techniques, for example using \emph{rejection sampling} and the \emph{Grover-Rudolph algorithm}. We refer to~\cite{Grover2002CreatingST,Regev05,Brakerski18,brakerski2021cryptographic}) for a reference. 

Our first algorithm (see \expref{Algorithm}{alg:GenDual} in \expref{Figure}{fig:GenDual}) prepares the dual Gaussian state from \expref{Definition}{def:Gaussian-states} with respect to an input matrix $\vec  A \in \Z_q^{n \times m}$ and parameter $\sigma = \Omega(\sqrt{m})$, and is defined as follows.

\begin{figure}[h]
    \centering
    \begin{algorithm}[H]\label{alg:GenDual}
\DontPrintSemicolon
\SetAlgoLined
\KwIn{Matrix $\vec A \in \Z_q^{n \times m}$ and parameter $\sigma = \Omega(\sqrt{m})$.}
\KwOut{Gaussian state $\ket{\hat\psi_{\vec y}}$ and $\vec y \in \Z_q^n$.}

Prepare a Gaussian superposition in system $X$ with parameter $\sigma > 0$:
    $$
 \ket{\hat\psi}_{XY} =    \sum_{\vec x \in \Z_q^m} \rho_\sigma(\vec x) \ket{\vec x}_X \otimes \ket{\vec 0}_Y.
    $$\;
Apply the unitary $U_{\vec A}: \ket{\vec x}\ket{\vec 0} \rightarrow \ket{\vec x} \ket{\vec A \cdot \vec x \Mod{q}}$ on systems $X$ and $Y$:
$$
 \ket{\hat \psi}_{XY} =   \sum_{\vec x \in \Z_q^m} \rho_\sigma(\vec x) \ket{\vec x}_X \otimes \ket{\vec A \cdot \vec x \Mod{q}}_Y.
  $$\;    
Measure system $Y$ in the computational basis, resulting in the state
    $$
    \ket{\hat\psi_{\vec y}}_{XY} = \sum_{\substack{\vec x \in \Z_q^m:\\ \vec A \vec x= \vec y}} \rho_\sigma(\vec x) \ket{\vec x}_X \otimes \ket{\vec y}_Y.
    $$\;
Output the state $\ket{\hat\psi_{\vec y}}$ in system $X$ and the outcome $\vec y \in \Z_q^n$ in system $Y$.
 \caption{$\mathsf{GenDual}(\vec A,\sigma)$}
\end{algorithm}

    \caption{Quantum algorithm which takes as input a matrix $\vec A \in \Z_q^{n \times m}$ and a width parameter $\sigma = \Omega(\sqrt{m})$, and outputs the dual Gaussian state in \expref{Definition}{def:Gaussian-states}.}
    \label{fig:GenDual}
\end{figure}

Our second algorithm (see \expref{Algorithm}{alg:GenPrimal} in \expref{Figure}{fig:GenPrimal}) prepares the primal Gaussian state with respect to an input matrix $\vec  A \in \Z_q^{n \times m}$ and parameter $\sigma = \Omega(\sqrt{m})$. Here, in order for \expref{Lemma}{lem:switching} to apply, it is crucial that the columns of $\vec A$ generate $\Z_q^n$. Fortunately, it follows from \expref{Lemma}{lem:full-rank} that a uniformly random matrix $\vec  A \rand \Z_q^{n \times m}$ satisfies this property with overwhelming probability.

\begin{figure}[h]
    \centering
\begin{algorithm}[H]\label{alg:GenPrimal}
\DontPrintSemicolon
\SetAlgoLined
\KwIn{Matrix $\vec A \in \Z_q^{n \times m}$ whose columns  generate $\Z_q^n$, and a parameter $\sigma = \Omega(\sqrt{m})$.}
\KwOut{Gaussian state $\ket{\psi_{\vec y}}$ and $\vec y \in \Z_q^n$.}

Run $\mathsf{GenDual}(\vec A,\sigma)$, resulting in the state
    $$
    \ket{\hat\psi_{\vec y}}_{XY} = \sum_{\substack{\vec x \in \Z_q^m:\\ \vec A \vec x= \vec y}} \rho_\sigma(\vec x) \ket{\vec x}_X \otimes \ket{\vec y}_Y.
    $$\;
Apply the quantum Fourier transform $\FT_q$ to system $X$.\;  

Output the state in system $X$, denoted by $\ket{\psi_{\vec y}}$, and the outcome $\vec y \in \Z_q^n$ in system $Y$.
 \caption{$\mathsf{GenPrimal}(\vec A,\sigma)$}
\end{algorithm}
  \caption{Quantum algorithm which takes as input a matrix $\vec A \in \Z_q^{n \times m}$ and a parameter $\sigma = \Omega(\sqrt{m})$, and outputs the primal Gaussian state in \expref{Definition}{def:Gaussian-states}.}
    \label{fig:GenPrimal}
\end{figure}

\subsection{Invariance under Pauli-Z dephasing}

In this section, we prove a surprising property about the dual Gaussian state from \expref{Definition}{def:Gaussian-states}. We prove \expref{Theorem}{thm:invariance-Pauli-Z}, which says that the Pauli-$\vec Z$ dephasing channel with respect to the $\LWE$ distribution leaves the dual Gaussian state approximately invariant.

\begin{theorem}\label{thm:invariance-Pauli-Z}
Let $n,m\in \N$ be integers and let $q\geq 2$ be a prime modulus, each parameterized by the security parameter $\lambda \in \N$. Let  $\sigma \in (\sqrt{8m},q/\sqrt{8m})$ be a function of $\lambda$. Let $\vec y \in \Z_q^n$ be any vector and $\vec A \in \Z_q^{n \times m}$ be any matrix whose columns generate $\Z_q^n$, and let $\ket{\hat \psi_{\vec y}}$ be the
dual Gaussian state,
$$
\ket{\hat \psi_{\vec y}} \,\,= \sum_{\substack{\vec x \in \Z_q^m\\ \vec A \vec x = \vec y \Mod{q}}}\rho_{\sigma}(\vec x) \ket{\vec x}.
$$
Let $\algo Z_{\LWE_{n,q,\alpha q}^m}$ be the Pauli-$\vec Z$ dephasing channel with respect to the $\LWE_{n,q,\alpha q}^m$ distribution for $\vec A \in \Z_q^{n \times m}$ and a noise ratio $\alpha \in (0,1)$ with relative noise magnitude $1/\alpha= \sigma \cdot 2^{o(n)}$, i.e.
$$
\algo Z_{\LWE_{n,q,\alpha q}^m}(\rho) = \sum_{\vec s_0 \in \Z_q^n} \sum_{\vec e_0 \in \Z_q^m} q^{-n}  D_{\Z_q^m,\alpha q}(\vec e_0) \,
\vec Z_q^{s_0\cdot\vec A + \vec e_0} \,\rho\, \vec Z_q^{-(s_0 \cdot \vec A + \vec e_0)}, \quad\quad \forall \rho \in L((\mathbb{C}^q)^{\otimes m}).
$$
Then, there exists a negligible function $\eps(\lambda)$ such that
$$
\algo Z_{\LWE_{n,q,\alpha q}^m}(\ketbra{\hat\psi_{\vec y}}{\hat\psi_{\vec y}}) \,\,\approx_\eps \,\, \ketbra{\hat\psi_{\vec y}}{\hat\psi_{\vec y}}.
$$
In other words, the Pauli-$\vec Z$ dephasing channel with respect to the $\LWE$ distribution leaves the dual Gaussian state approximately invariant.
\end{theorem}

\begin{proof}
Let $\vec y \in \Z_q^n$ be an arbitrary vector and recall that the dual Gaussian state $\ket{\hat\psi_{\vec y}}$ is given by
\begin{align}\label{eq:Gaussian-coset-1}
\ket{\hat \psi_{\vec y}} \,\,= \sum_{\substack{\vec x \in \Z_q^m\\ \vec A \vec x = \vec y \Mod{q}}}\rho_{\sigma}(\vec x) \ket{\vec x}.
\end{align}
Consider a sample $\vec b=\vec s_0 \cdot\vec A + \vec e_0 \Mod{q}) \sim\LWE_{n,q,\alpha q}^m$ with $\vec s \rand  \Z_q^n$ and $\vec e_0 \sim D_{\Z_q^m,\alpha q}$. Because $\sigma \in (\sqrt{8m},q/\sqrt{8m})$ and $1/\alpha= \sigma \cdot 2^{o(n)}$, there exist negligible $\eta(\lambda)$ and $\kappa(\lambda)$ such that
\begin{align*}
\vec Z_q^{s_0\cdot\vec A  + \vec e_0} \ket{\hat \psi_{\vec y}} 
\,\,\,&= \,\,\,\, \FT_q \,
\vec X_q^{s_0\cdot\vec A  + \vec e_0} \,\FT_q^\dag \ket{\hat \psi_{\vec y}} & (\text{\expref{Lemma}{lem:XZ-conjugation}})\\ 
&\approx_{\eta} \,\, \FT_q \,
\vec X_q^{s_0\cdot\vec A  + \vec e_0} \, \ket{ \psi_{\vec y}} & (\text{\expref{Lemma}{lem:switching}})\\ 
&\approx_{\kappa} \,\,\omega_{q}^{\ip{\vec s_0,\vec y}} \FT_q  \ket{ \psi_{\vec y}}& (\text{\expref{Lemma}{lem:shifted-gaussian}})\\ 
&\approx_{\eta} \,\, \omega_{q}^{\ip{\vec s_0,\vec y}}\ket{\hat \psi_{\vec y}}. & (\text{\expref{Lemma}{lem:switching}})
\end{align*}
Here, $\ket{ \psi_{\vec y}}$ is the primal Gaussian state given by
$$
\ket{ \psi_{\vec y}} = \sum_{\vec s \in \Z_q^n} \sum_{\vec e \in \Z_q^m} \rho_{q/\sigma}(\vec e) \, \omega_q^{-\ip{\vec s,\vec y}} \ket{\vec s \vec A + \vec e \Mod{q}}.
$$
In other words, $\ket{\hat \psi_{\vec y}}$ in Eq.~\eqref{eq:Gaussian-coset-1} is an approximate eigenvector of the generalized Pauli operator $\vec Z_q^{s_0\cdot\vec A  + \vec e_0}$ with respect to the same matrix $\vec A \in \Z_q^{n \times m}$. Note that we can simply discard
$\omega_{q}^{\ip{\vec s_0,\vec y}} \in \mathbb{C}$ because it serves as a global phase. Hence, there exists a negligible function $\eps(\lambda)$ such that
\begin{align*}
\algo Z_{\LWE_{n,q,\alpha q}^m}(\ketbra{\hat\psi_{\vec y}}{\hat\psi_{\vec y}}) &= \sum_{\vec s_0 \in \Z_q^n} \sum_{\vec e_0 \in \Z_q^m} q^{-n}  D_{\Z_q^m,\alpha q}(\vec e_0) \,
\vec Z_q^{s_0\cdot\vec A  + \vec e_0}  \ketbra{\hat\psi_{\vec y}}{\hat\psi_{\vec y}} \vec Z_q^{-(s_0\cdot\vec A  + \vec e_0)}   \\
&\approx_\eps \Bigg(\sum_{\vec s_0 \in \Z_q^n} q^{-n} \Bigg) \cdot \left(\sum_{\vec e_0 \in \Z_q^m} D_{\Z_q^m,\alpha q}(\vec e_0)\right) \, \ketbra{\hat\psi_{\vec y}}{\hat\psi_{\vec y}}\\
&=  \ketbra{\hat\psi_{\vec y}}{\hat\psi_{\vec y}}.
\end{align*}
\end{proof}

\section{Uncertainty Relation for Fourier Basis Projections}\label{sec:uncertainty}

In this section, we prove an entropic uncertainty relation with respect to so-called Fourier basis projections. Informally, we say that a projector $\widehat{\boldsymbol{\Pi}}$ is a \emph{Fourier basis projection}, if $\widehat{\boldsymbol{\Pi}}$ corresponds to a projector (onto a subset of $\Z_q^m$) which is conjugated by the $q$-ary Fourier transform $\FT_q$.
Notice that the deletion procedures of our encryption schemes with certified deletion in \expref{Section}{sec:Dual-Regev-PKE} and \expref{Section}{sec:Dual-Regev-FHE} require a Fourier basis projection onto a small set of solutions to the (inhomogenous) short integer solution $(\ISIS)$ problem. Another example can be found in the work of Aaronson and Christiano~\cite{https://doi.org/10.48550/arxiv.1203.4740} who used Hadamard basis projections (a special case of the $q$-ary Fourier transform) onto small hidden subspaces to verify quantum money states.

Our uncertainty relation captures the following intuitive property: any system which passes a Fourier basis projection onto a small subset of $\Z_q^m$ (say, with high probability) must necessarily be \emph{unentangled} with any auxiliary system. We formalize this statement using the (smooth) quantum min-entropy (\expref{Definition}{def:smooth-entropies}).

\subsection{Fourier basis projections}

\begin{definition}[Fourier basis projection]\label{def:fourier-basis-projection} Let $m \in \N$ and let $q \geq 2$ be an integer modulus. Let $\algo S \subseteq \Z_q^m$ be an arbitrary set and let $\boldsymbol{\Pi}_{\algo S}$ be the associated projector onto $\algo S$, where
$$
\boldsymbol{\Pi}_{\algo S}= \sum_{\vec x \in \algo S} \proj{\vec x}.
$$
Then, we define the associated Fourier basis projection onto $\algo S$ as the projector
$$ \widehat{\boldsymbol{\Pi}}_{\algo S}=\FT_q^\dag \boldsymbol{\Pi}_{\algo S} \FT_q.
$$
\end{definition}

\subsection{Uncertainty relation}

In this section, our main result is the following.

\begin{theorem}[Uncertainty relation for Fourier basis projections]\label{thm:uncertainty} Let $m \in \N$, $q \geq 2$ be a modulus, $\{\ket{\psi^{\vec x}}\}_{\vec x \in \Z_q^m}$ be any family of normalized auxiliary states, and let $\ket{\psi}_{AB}$ be any state of the form
$$
\ket{\psi}_{AB} = \sum_{\vec x \in \Z_q^m} \alpha_{\vec x} \ket{\vec x}_A \otimes \ket{\psi^{\vec x}}_B  \quad\,\, \text{ s.t. } \quad \,\, \sum_{\vec x \in \Z_q^m} |\alpha_{\vec x}|^2 = 1.
$$
Let $\algo S \subseteq \Z_q^m$ be an arbitrary set and define the following projectors onto system $A$,
$$
\boldsymbol{\Pi}_{\algo S} = \sum_{\vec x \in \algo S} \proj{\vec x} \quad\,\, \text{ and } \quad \,\, \widehat{\boldsymbol{\Pi}}_{\algo S}=\FT_q^\dag \boldsymbol{\Pi}_{\algo S} \FT_q.
$$
Suppose that $\|(\widehat{\boldsymbol{\Pi}}_{\algo S} \otimes \id_B)\ket{\psi}_{AB}\|^2 = 1-\eps$, for some $\eps \geq 0$. Then, it holds that
$$
\hmin^{\sqrt{\eps}}(X|B)_\rho \,\geq\, m \cdot \log q - 2\cdot \log|\algo S|.
$$
Here, $\rho_{XB}$ results from a computational basis measurement of system $A$ of the state $\proj{\psi}_{AB}$, i.e.
$$
\rho_{XB} = \sum_{\vec x \in \Z_q^m} \proj{\vec x}_X \otimes \tr_A \left[(\proj{\vec x}_A \otimes \id_B) \proj{\psi}_{AB} \right].
$$
\begin{proof}
Suppose that $\ket{\psi}_{AB}$ satisfies $\|(\widehat{\boldsymbol{\Pi}}_{\algo S} \otimes \id_B)\ket{\psi}_{AB}\|^2 = 1-\eps$, for some $\eps \geq 0$.
From \expref{Lemma}{lem:closeness-ideal}, it follows that there exists an ideal pure state,
$$
\ket{\bar{\psi}}_{AB} = \frac{(\widehat{\boldsymbol{\Pi}}_{\algo S} \otimes \id_B)\ket{\psi}_{AB}}{\|(\widehat{\boldsymbol{\Pi}}_{\algo S} \otimes \id_B)\ket{\psi}_{AB}\|} = \sum_{\vec x \in \Z_q^m}  \bar{\alpha}_{\vec x} \ket{\vec x}_A \otimes \ket{\psi^{\vec x}}_B  \quad\,\, \text{ s.t. } \quad \,\, \sum_{\vec x \in \Z_q^m} |\bar{\alpha}_{\vec x}|^2 = 1,
$$
with the property that
$$
\| \proj{\psi} - \proj{\bar{\psi}}\|_\tr \leq \sqrt{\eps} \quad \,\, \text{ and } \quad \,\, \ket{\bar{\psi}} \in \mathrm{im}(\widehat{\boldsymbol{\Pi}}_{\algo S} \otimes \id_B).
$$ 
Because $\ket{\bar{\psi}}_{AB}$ lies in the image of the projector $\widehat{\boldsymbol{\Pi}}_{\algo S} \otimes \id_B$, we have
$$
\ket{\bar{\psi}}_{AB} = (\widehat{\boldsymbol{\Pi}}_{\algo S} \otimes \id_B) \ket{\bar{\psi}}_{AB} = q^{-m}\sum_{\vec x,\vec x' \in \Z_q^m} \sum_{\vec s \in \algo S} \bar{\alpha}_{\vec x'} \cdot  \omega_q^{\ip{\vec x,\vec s}}\omega_q^{-\ip{\vec x',\vec s}} \ket{\vec x}_A \otimes \ket{\psi^{\vec x'}}_B.
$$
Let us now analyze the ideal state $\bar{\rho}_{XB}$ which results from a computational basis measurement of system $A$ of the state $\proj{\bar\psi}_{AB}$. In other words, we consider the CQ state given by
$$
\bar{\rho}_{XB} = \sum_{\vec x \in \Z_q^m} \proj{\vec x}_X \otimes \tr_A \left[(\proj{\vec x}_A \otimes \id_B) \proj{\bar{\psi}}_{AB} \right].
$$
By the definition of the guessing probability in \expref{Definition}{def:guessing}, we have
\begin{align*}
p_{\guess}(X|B)_{\bar{\rho}} &=\underset{\boldsymbol{M}_B^x}{\sup} \sum_{x \in \Z_q^m} \Big\|(\proj{\vec x}_A \otimes \boldsymbol{M}_B^{\vec x}) \ket{\bar{\psi}}_{AB}\Big\|^2\\
&= \underset{\boldsymbol{M}_B^x}{\sup} \sum_{x \in \Z_q^m} q^{-2m}
\left\|
\sum_{\vec x' \in \Z_q^m} \sum_{\vec s \in \algo S} \bar{\alpha}_{\vec x'} \cdot  \omega_q^{\ip{\vec x,\vec s}}\omega_q^{-\ip{\vec x',\vec s}} \ket{\vec x}_A \otimes \boldsymbol{M}_B^{\vec x}\ket{\psi^{\vec x'}}_B
\right\|^2\\
&= \underset{\boldsymbol{M}_B^x}{\sup} \sum_{x \in \Z_q^m} 
q^{-2m}\left\|
\sum_{\vec x' \in \Z_q^m} \bar{\alpha}_{\vec x'} \cdot \Big( \sum_{\vec s \in \algo S}\omega_q^{\ip{\vec x,\vec s}}\omega_q^{-\ip{\vec x',\vec s}}\Big) \,\ket{\vec x}_A \otimes \boldsymbol{M}_B^{\vec x}\ket{\psi^{\vec x'}}_B
\right\|^2.
\end{align*}
Using the Cauchy-Schwarz-inequality, we find that for any $\vec x \in \Z_q^m$:
\begin{align}
&\left\|\sum_{\vec x' \in \Z_q^m} \bar{\alpha}_{\vec x'} \cdot \Big( \sum_{\vec s \in \algo S}\omega_q^{\ip{\vec x,\vec s}}\omega_q^{-\ip{\vec x',\vec s}}\Big)  \ket{\vec x}_A \otimes \boldsymbol{M}_B^{\vec x}\ket{\psi^{\vec x'}}_B
\right\| \nonumber\\
&\leq \sqrt{\sum_{\vec x' \in \Z_q^m} \Big|\bar{\alpha}_{\vec x'} \cdot \Big( \sum_{\vec s \in \algo S}\omega_q^{\ip{\vec x,\vec s}}\omega_q^{-\ip{\vec x',\vec s}}\Big)\Big|^2} \cdot \sqrt{ \sum_{\vec x' \in \Z_q^m}\Big\|\ket{\vec x}_A \otimes \boldsymbol{M}_B^{\vec x}\ket{\psi^{\vec x'}}_B
\Big\|^2}\nonumber\\
&\leq \sqrt{ |\algo S|^2\sum_{\vec x' \in \Z_q^m}\big|\bar{\alpha}_{\vec x'}\big|^2} \cdot \sqrt{ \sum_{\vec x' \in \Z_q^m}\Big\|\ket{\vec x}_A \otimes \boldsymbol{M}_B^{\vec x}\ket{\psi^{\vec x'}}_B
\Big\|^2}\nonumber\\
&= |\algo S|\cdot \sqrt{ \sum_{\vec x' \in \Z_q^m}\Big\|\boldsymbol{M}_B^{\vec x}\ket{\psi^{\vec x'}}_B
\Big\|^2}.\label{eq:CS-step}
\end{align}
Using the inequality in~\eqref{eq:CS-step}, we can now bound the guessing probability as follows:
\begin{align*}
p_{\guess}(X|B)_{\bar{\rho}} \,\,&\leq \,\,  \frac{|\algo S|^2}{q^{2m}} \cdot  \underset{\boldsymbol{M}_B^{\vec x}}{\sup} \sum_{\vec x \in \Z_q^m} 
\sum_{\vec x' \in \Z_q^m}\Big\|\boldsymbol{M}_B^{\vec x}\ket{\psi^{\vec x'}}_B
\Big\|^2\\
&=  \frac{|\algo S|^2}{q^{2m}} \cdot \sum_{\vec x' \in \Z_q^m} \underset{\boldsymbol{M}_B^{\vec x}}{\sup} \sum_{\vec x \in \Z_q^m} \Big\|\boldsymbol{M}_B^{\vec x}\ket{\psi^{\vec x'}}_B
\Big\|^2\\
&= \frac{|\algo S|^2}{q^{m}}. & (\text{since } \sum_{\vec x} \boldsymbol{M}_B^{\vec x} = \id)
\end{align*}
Because the \emph{purified distance} is bounded above by the \emph{trace distance}, it follows that
$$
P(\rho_{XB}, \bar{\rho}_{XB}) \, \leq \, \|\rho_{XB} - \bar{\rho}_{XB}\|_\tr \,\leq\, \| \proj{\psi} - \proj{\bar{\psi}}\|_\tr \,\leq\, \sqrt{\eps}.
$$
Therefore, by the definition of (smooth) min-entropy (see \expref{Definition}{def:smooth-entropies}), we have
\begin{align}
\hmin(X \, | \, B)_{\bar{\rho}} \,\,\, \leq \underset{\substack{\sigma_{XB}\\
  P(\sigma_{XB}, \rho_{XB}) \leq \sqrt{\eps}}}{\sup}  \hmin(X \, | \, B)_{\sigma} \,\,\,= \,\, \hmin^{\sqrt{\eps}}(X \, | \, B)_\rho. \label{eq:min-entropy-ineq}
\end{align}
Putting everything together, it follows from \eqref{eq:min-entropy-ineq} and \expref{Theorem}{thm:guessing} that
\begin{align*}
\hmin^{\sqrt{\eps}}(X \, | \, B)_\rho \, &\geq \,\hmin(X \, | \, B)_{\bar{\rho}}\\
&=- \log \big(p_{\guess}(X|B)_{\bar\rho}\big)\\
\,&\geq\, m \cdot \log q - 2\cdot \log|\algo S|.\quad\quad
\end{align*}
This proves the claim.
\end{proof}
\end{theorem}

\section{Gaussian-Collapsing Hash Functions}

Unruh~\cite{cryptoeprint:2015/361} introduced the notion of collapsing hash functions in his seminal work on computationally binding quantum commitments. This property is captured by the following definition.

\begin{definition}[Collapsing hash function, \cite{cryptoeprint:2015/361}] Let $\lambda \in \N$ be the security parameter. A hash function family $\algo H = \{H_\lambda\}_{\lambda \in \N}$ is called collapsing if, for every $\QPT$ adversary $\algo A$,
$$
|
\Pr[ \mathsf{CollapseExp}_{\algo H,\algo A,\lambda}(0)=1] - \Pr[ \mathsf{CollapseExp}_{\algo H,\algo A,\lambda}(1)=1]
| \leq \negl(\lambda).
$$
Here, the experiment $\mathsf{CollapseExp}_{\algo H,\algo A,\lambda}(b)$ is defined as follows:
\begin{enumerate}
    \item The challenger samples a random hash function $h \rand H_\lambda$, and sends a description of $h$ to $\algo A$.
    \item $\algo A$ responds with a (classical) string $y \in \bit^{n(\lambda)}$ and an $m(\lambda)$-qubit quantum state in system $X$.
    \item The challenger coherently computes $h$ (into an auxiliary system $Y$) given the state in system $X$, and then performs a two-outcome measurement on $Y$ indicating whether the output of $h$ equals $y$. If $h$ does not equal $y$ the challenger aborts and
outputs $\bot$.

\item If $b=0$, the challenger does nothing. Else, if $b=1$, the challenge measures the $m(\lambda)$-qubit system $X$ in the computational basis. Finally, the challenger returns the state in system $X$ to $\algo A$.

\item $\algo A$ returns a bit $b'$, which we define as the output of the experiment.
\end{enumerate}
\end{definition}

Motivated by the properties of the dual Gaussian state from \expref{Definition}{def:Gaussian-states}, we consider a special class of hash functions which are \emph{collapsing} with respect to Gaussian superpositions. Informally, we say that a hash function $h$ is \emph{Gaussian-collapsing} if it is computationally difficult to distinguish between a Gaussian superposition of pre-images and a single (measured) Gaussian pre-image (of $h$). We formalize this below.

\begin{definition}[Gaussian-collapsing hash function]\label{def:gaussian-collapsing} Let $\lambda \in \N$ be the security parameter, $m(\lambda),n(\lambda) \in \N$ and let $q(\lambda) \geq 2$ be a modulus. Let $\sigma > 0$. A hash function family $\algo H = \{H_\lambda\}_{\lambda \in \N}$ with domain $\algo X = \Z_q^{m}$ and range $\algo Y=\Z_q^{n}$ is called $\sigma$-Gaussian-collapsing if, for every $\QPT$ adversary $\algo A$,
$$
|
\Pr[ \mathsf{GaussCollapseExp}_{\algo H,\algo A,\lambda}(0)=1] - \Pr[ \mathsf{GaussCollapseExp}_{\algo H,\algo A,\lambda}(1)=1]
| \leq \negl(\lambda).
$$
Here, the experiment $\mathsf{GaussCollapseExp}_{\algo H,\algo A,\lambda}(b)$ is defined as follows:
\begin{enumerate}
    \item The challenger samples a random hash function $h \rand H_\lambda$
    and prepares the quantum state
    $$
    \ket{\hat\psi}_{XY} = \sum_{\vec x \in \Z_q^m} \rho_\sigma(\vec x) \ket{\vec x}_X \otimes \ket{h(\vec x)}_Y.
    $$

    \item The challenger measures system $Y$ in the computational basis, resulting in the state
    $$
    \ket{\hat\psi_{\vec y}}_{XY} = \sum_{\substack{\vec x \in \Z_q^m:\\ h(\vec x)= \vec y}} \rho_\sigma(\vec x) \ket{\vec x}_X \otimes \ket{\vec y}_Y.
    $$
\item If $b=0$, the challenger does nothing. Else, if $b=1$, the challenger measures system $X$ of the quantum state $\ket{\hat\psi_{\vec y}}$ in the computational basis. Finally, the challenger sends the outcome state in systems $X$ to $\algo A$, together with the string $\vec y \in \Z_q^n$ and a classical description of the hash function $h$.

\item $\algo A$ returns a bit $b'$, which we define as the output of the experiment.
\end{enumerate}
\end{definition}

The following follows immediately from the definition of Gaussian-collapsing hash functions, and the fact that the dual Gaussian state can be efficiently prepared using \expref{Algorithm}{alg:GenDual}.

\begin{claim}\label{claim:collapsing}
Let $\algo H = \{H_\lambda\}_{\lambda \in \N}$ be a hash function family with domain $\algo X = \Z_q^{m}$ and range $\algo Y=\Z_q^{n}$, where $m(\lambda),n(\lambda) \in \N$. If $\algo H$ is collapsing, then $\algo H$ is also $\sigma$-Gaussian-collapsing, for any $\sigma =\Omega(\sqrt{m})$.
\end{claim}

\subsection{Ajtai's hash function}

Liu and Zhandry~\cite{cryptoeprint:2019/262} implicitly showed that the \emph{Ajtai} hash function $h_{\vec A}(\vec x) = \vec A \vec x \Mod{q}$ is collapsing -- and thus \emph{Gaussian-collapsing} -- via the notion of \emph{lossy functions} and by assuming the superpolynomial hardness of (decisional) $\LWE$. In this section, we give a simple and direct proof that the Ajtai hash function is Gaussian-collapsing assuming (decisional) $\LWE$, which might be of independent interest.

\begin{theorem}\label{thm:GaussCollapse}
Let $n\in \N$ and $q\geq 2$ be a prime modulus with $m \geq 2n \log q$, each parameterized by $\lambda \in \N$. Let  $\sigma \in (\sqrt{8m},q/\sqrt{8m})$ be a function of $\lambda$. Then, the Ajtai hash function family $\algo H = \{H_\lambda\}_{\lambda \in \N}$ with
$$
H_\lambda = \left\{ h_{\vec A}: \Z_q^m \rightarrow \Z_q^n \, \text{ s.t. } \, h_{\vec A}(\vec x) = \vec A \cdot \vec x \Mod{q}; \, \vec A \in \Z_q^{n \times m} \right\} 
$$
is $\sigma$-Gaussian-collapsing assuming the quantum hardness of the decisional $\LWE_{n,q,\alpha q}^m$ problem, for any parameter $\alpha \in (0,1)$ with relative noise magnitude $1/\alpha= \sigma \cdot 2^{o(n)}$.
\end{theorem}

\begin{proof} Let $\algo A$ denote the $\QPT$ adversary in the experiment $\mathsf{GaussCollapseExp}_{\algo H,\algo A,\lambda}(b)$ with $b \in \bit$.
To prove the claim, we give a reduction from the decisional $\LWE_{n,q,\alpha q}^m$ assumption. We are given as input a sample $(\vec A,\vec b)$ with $\vec A \rand \Z_q^{m \times n}$, where $\vec b=\vec s_0 \cdot \vec A  + \vec e_0 \Mod{q})$ is either a sample from the $\LWE$ distribution with $\vec s_0 \rand  \Z_q^n$ and $\vec e_0 \sim D_{\Z^m,\alpha q}$, or where $\vec b$ is a uniformly random string $\vec u \rand \Z_q^m$. 

Consider the distinguisher $\algo D$ that acts as follows on input $1^\lambda$ and $(\vec A,\vec b)$:
\begin{enumerate}
    \item $\algo D$ prepares a bipartite quantum state on systems $X$ and $Y$ with
    $$
    \ket{\hat\psi}_{XY} = \sum_{\vec x \in \Z_q^m} \rho_\sigma(\vec x) \ket{\vec x}_X \otimes \ket{\vec A \cdot \vec x \Mod{q}}_Y.
    $$
    \item $\algo D$ measures system $Y$ in the computational basis, resulting in the state
    $$
    \ket{\hat\psi_{\vec y}}_{XY} = \sum_{\substack{\vec x \in \Z_q^m:\\ \vec A \vec x= \vec y}} \rho_\sigma(\vec x) \ket{\vec x}_X \otimes \ket{\vec y}_Y.
    $$

    \item $\algo D$ applies the generalized Pauli-$\vec Z$ operator $\vec Z_q^{\vec b}$ on system $X$, resulting in the state
    
    $$
    (\vec Z_q^{\vec b} \ot \id_Y)\ket{\hat\psi_{\vec y}}_{XY} = \sum_{\substack{\vec x \in \Z_q^m:\\ \vec A \vec x= \vec y}} \rho_\sigma(\vec x) \left(\vec Z_q^{\vec b}\ket{\vec x}_X \right) \otimes \ket{\vec y}_Y.
    $$
    
    \item $\algo D$ runs the adversary $\algo A$ on input system $X$ and classical descriptions of $\vec A \in \Z_q^{n \times m}$ and $\vec y \in \Z_q^n$.
    
    \item $\algo D$ outputs whatever bit $b' \in \bit$ the adversary $\algo A$ outputs.
\end{enumerate}
Suppose that, for every $\lambda \in \N$, there exists a polynomial $p(\lambda)$ such that
\begin{align*}
|
\Pr[ \mathsf{GaussCollapseExp}_{\algo H,\algo A,\lambda}(0)=1] - \Pr[ \mathsf{GaussCollapseExp}_{\algo H,\algo A,\lambda}(1)=1]
| \geq \frac{1}{p(\lambda)}.
\end{align*}
We now show that this implies that $\algo D$ succeeds at the decisional $\LWE_{n,q,\alpha q}^m$ experiment with advantage at least $1/p(\lambda) - \negl(\lambda)$. We distinguish between the following two cases.

If $(\vec A,\vec b)$ is a sample from the $\LWE$ distribution with $\vec b= \vec s_0 \cdot \vec A  + \vec e_0 \Mod{q})$, then the adversary $\algo A$ receives as input the following quantum state in system $X$:
\begin{align*}
\algo Z_{\LWE_{n,q,\alpha q}^m}(\ketbra{\hat\psi_{\vec y}}{\hat\psi_{\vec y}}_X) &= \sum_{\vec s_0 \in \Z_q^n} \sum_{\vec e_0 \in \Z^m} q^{-n}  D_{\Z^m,\alpha q}(\vec e_0) \,\,
\vec Z_q^{s_0\cdot\vec A  + \vec e_0}  \ketbra{\hat\psi_{\vec y}}{\hat\psi_{\vec y}}_X \vec Z_q^{-(s_0\cdot\vec A  + \vec e_0)}.
\end{align*}
From \expref{Theorem}{thm:invariance-Pauli-Z} it follows that there exists a negligible function $\eps(\lambda)$ such that
\begin{align*}
\algo Z_{\LWE_{n,q,\alpha q}^m}(\ketbra{\hat\psi_{\vec y}}{\hat\psi_{\vec y}}_X) \,\,\approx_\eps \,\, \ketbra{\hat\psi_{\vec y}}{\hat\psi_{\vec y}}_X.
\end{align*}
In other words, $\algo A$ receives as input a state in system $X$ which is within negligible trace distance of the dual Gaussian state $\ket{{\hat\psi_{\vec y}}}$, which corresponds precisely to the input in $\mathsf{GaussCollapseExp}_{\algo H,\algo A,\lambda}(0)$.

If $(\vec A,\vec b)$ is a uniformly random sample, where $\vec b$ is a random string $\vec u \rand \Z_q^m$, then the adversary $\algo A$ receives as input the following quantum state in system $X$:
$$
\algo Z(\ketbra{\hat\psi_{\vec y}}{\hat\psi_{\vec y}}_X) = q^{-m} \sum_{\vec u \in \Z_q^m} \vec Z_q^{\vec u} \, \ketbra{\hat\psi_{\vec y}}{\hat\psi_{\vec y}}_X \, \vec Z_q^{-\vec u}.
$$
Because $\algo Z$ corresponds to the uniform Pauli-$\vec Z$ dephasing channel, it follows from \expref{Lemma}{lem:random-Z} that
$$
\algo Z(\ketbra{\hat\psi_{\vec y}}{\hat\psi_{\vec y}}_X) =  \sum_{\vec x \in \Z_q^m}
\big| \ip{\vec x | \hat\psi_{\vec y}}\big| ^2
\,\ketbra{\vec x}{\vec x}_X.
$$
In other words, $\algo A$ receives as input a mixed state which is the result of a computational basis measurement of the Gaussian state $\ket{{\hat\psi_{\vec y}}}$. Note that this corresponds precisely to the input in $\mathsf{GaussCollapseExp}_{\algo H,\algo A,\lambda}(1)$.

By assumption, the adversary $\algo A$ succeeds with advantage at least $1/p(\lambda)$. Therefore, the distinguisher $\algo D$ 
succeeds at the decisional $\LWE_{n,q,\alpha q}^m$ experiment with probability at least $1/p(\lambda) - \negl(\lambda)$.
\end{proof}

\begin{theorem}\label{thm:pseudorandom-SLWE}
Let $n\in \N$ and $q\geq 2$ be a prime modulus with $m \geq 2n \log q$, each parameterized by the security parameter $\lambda \in \N$. Let  $\sigma \in (\sqrt{8m},q/\sqrt{8m})$ be a function of $\lambda$ and $\vec  A \rand \Z_q^{n \times m}$ be a matrix.

Then,
the following states are computationally indistinguishable assuming the quantum hardness of decisional $\LWE_{n,q,\alpha q}^m$, for any parameter $\alpha \in (0,1)$ with relative noise magnitude $1/\alpha= \sigma \cdot 2^{o(n)}$:
\begin{itemize}
\item For any $(\ket{\hat\psi_{\vec y}},\vec y) \leftarrow \mathsf{GenDual}(\vec A,\sigma)$ in \expref{Algorithm}{alg:GenDual}:
$$
\ket{\hat\psi_{\vec y}}=\sum_{\substack{\vec x \in \Z_q^m\\ \vec A \vec x = \vec y \Mod{q}}}\rho_{\sigma}(\vec x) \,\ket{\vec x}\quad \approx_c \quad\,\, \ket{\vec x_0} \,\,: \,\, \quad\,\, \, \vec x_0 \sim D_{\Lambda_q^{\vec y}(\vec A),\frac{\sigma}{\sqrt{2}}}.\,\,\,
$$

\item For any $(\ket{\psi_{\vec y}},\vec y) \leftarrow \mathsf{GenPrimal}(\vec A,\sigma)$ in \expref{Algorithm}{alg:GenPrimal}:
$$
\ket{\psi_{\vec y}}=\sum_{\vec s \in \Z_q^n} \sum_{\vec e \in \Z_q^m} \rho_{\frac{q}{\sigma}}(\vec e) \, \omega_q^{-\ip{\vec s,\vec y}} \ket{\vec s \vec A + \vec e} \,\,\, \approx_c \,\,\,\,\, \sum_{\vec u \in \Z_q^m} \omega_q^{-\ip{\vec u,\vec x_0}}\ket{\vec u} \,\,: \,\, \quad \vec x_0 \sim D_{\Lambda_q^{\vec y}(\vec A),\frac{\sigma}{\sqrt{2}}}.
$$
\end{itemize}
Moreover, the distribution of $\vec y \in \Z_q^n$ is negligibly close in total variation distance to the uniform distribution over $\Z_q^n$. Here, $\Lambda_q^{\vec y}(\vec A) = \{ \vec x \in \Z^m  :  \vec A \vec x = \vec y \Mod{q}\}$ denotes a coset of the lattice $\Lambda_q^\bot(\vec A)$.
\end{theorem}

\begin{proof}
Let $\vec  A \rand \Z_q^{n \times m}$ be a random matrix. From \expref{Lemma}{lem:full-rank} it follows that the columns of $\vec A$ generate $\Z_q^n$ with overwhelming probability. Let us also recall the following simple facts about the discrete Gaussian.
According to \expref{Lemma}{lem:Gaussian-LHL}, the distribution of the syndrome $\vec A \cdot \vec x = \vec y \Mod{q}$ is statistically close to the uniform distribution over $\Z_q^n$, whenever $\vec x \sim D_{\Z^m,\sigma}$ and $\sigma = \omega(\sqrt{\log m})$. Moreover, the conditional distribution of $\vec x \sim D_{\Z^m,\sigma}$ given the syndrome $\vec y \in \Z_q^n$ is a discrete Gaussian distribution $D_{\Lambda_q^{\vec y}(\vec A),\sigma}$.

Let us now show the first statement. Recall that in \expref{Theorem}{thm:GaussCollapse} we show that the Ajtai hash function $h_{\vec A}(\vec x) = \vec A \cdot \vec x \Mod{q}$ is $\sigma$-Gaussian-collapsing assuming the decisional $\LWE_{n,q,\alpha q}^m$ assumption and a noise ratio $1/\alpha= \sigma \cdot 2^{o(n)}$.
Therefore, for $\vec y \in \Z_q^n$, the (normalized variant of the) dual Gaussian state,
$$
\ket{\hat \psi_{\vec y}} = \sum_{\substack{\vec x \in \Z_q^m:\\ \vec A \vec x = \vec y \Mod{q}}}\rho_{\sigma}(\vec x) \,\ket{\vec x}
$$
is computationally indistinguishable from the (normalized) classical mixture,
$$
\sum_{\vec x \in \Z_q^m}
\big| \ip{\vec x | \hat\psi_{\vec y}}\big| ^2
\ketbra{\vec x}{\vec x} = \left(\sum_{\substack{\vec z \in \Z_q^m:\\ \vec A \vec z = \vec y \Mod{q}}} \rho_{\sigma/\sqrt{2}}(\vec z) \right)^{-1}\sum_{\substack{\vec x \in \Z_q^m:\\ \vec A \vec x = \vec y \Mod{q}}} \rho_{\sigma/\sqrt{2}}(\vec x) \,\, \ketbra{\vec x}{\vec x},
$$
which is the result of a computational basis measurement of $\ket{\hat \psi_{\vec y}}$.\footnote{Here, the additional factor $1/\sqrt{2}$ arises from the normalization of the dual Gaussian state $\ket{\hat \psi_{\vec y}}$.}
Since $\sigma \in (\sqrt{8m},q/\sqrt{8m})$, the tail bound in \expref{Lemma}{lem:tailboundII} implies that the above mixture is statistically close to the discrete Gaussian $D_{\Lambda_q^{\vec y}(\vec A),\frac{\sigma}{\sqrt{2}}}$.

The second statement follows immediately by applying the (inverse) Fourier transform to both of the states above. Note that in \expref{Lemma}{lem:switching} we showed that the primal Gaussian state $$\ket{\psi_{\vec y}}=
\sum_{\vec s \in \Z_q^n} \sum_{\vec e \in \Z_q^m} \rho_{\frac{q}{\sigma}}(\vec e) \, \omega_q^{-\ip{\vec s,\vec y}} \ket{\vec s \vec A + \vec e}$$ 
is within negligible trace distance of $\FT_q^\dag \ket{\hat \psi_{\vec y}}$. This proves the claim.

\end{proof}

\subsection{Strong Gaussian-collapsing conjecture}

Our quantum encryption schemes with certified deletion in \expref{Section}{sec:Dual-Regev-PKE} and \expref{Section}{sec:Dual-Regev-FHE} rely on the assumption that Ajtai's hash function satisfies a strong Gaussian-collapsing property in the presence of leakage. We formalize the property as the following simple and falsifiable conjecture.

\begin{conjecture}[Strong Gaussian-Collapsing Conjecture]\label{conj:SGC}\ \\
Let $\lambda \in \N$ be the security parameter, $n(\lambda) \in \N$, $q(\lambda) \geq 2$ be a modulus and $m \geq 2n \log q$ be an integer. Let $\sigma = \Omega(\sqrt{m})$ be a parameter and let
$\algo H = \{H_\lambda\}_{\lambda \in \N}$ be the Ajtai hash function family with
$$
H_\lambda = \left\{ h_{\vec A}: \Z_q^m \rightarrow \Z_q^n \, \text{ s.t. } \, h_{\vec A}(\vec x) = \vec A \cdot \vec x \Mod{q}; \, \vec A \in \Z_q^{n \times m} \right\}.
$$
The Strong Gaussian-Collapsing Conjecture $(\mathsf{SGC}_{n,m,q,\sigma})$ states that, for every $\QPT$ adversary $\algo A$,
$$
|
\Pr[ \mathsf{StrongGaussCollapseExp}_{\algo H,\algo A,\lambda}(0)=1] - \Pr[ \mathsf{StrongGaussCollapseExp}_{\algo H,\algo A,\lambda}(1)=1]
| \leq \negl(\lambda).
$$
Here, the experiment $\mathsf{StrongGaussCollapseExp}_{\algo H,\algo A,\lambda}(b)$ is defined as follows:
\begin{enumerate}
    \item The challenger samples $ \bar{\vec A} \rand \Z_q^{n \times (m-1)}$
    and prepares the quantum state
    $$
    \ket{\hat\psi}_{XY} = \sum_{\vec x \in \Z_q^m} \rho_\sigma(\vec x) \ket{\vec x}_X \otimes \ket{\vec A \cdot \vec x \Mod{q}}_Y,
    $$
    where $\vec A = [\bar{\vec A} | \bar{\vec A} \cdot \bar{\vec x} \Mod{q}] \in \Z_q^{n \times m}$ is a matrix with $\bar{\vec x} \rand \bit^{m-1}$.

    \item The challenger measures system $Y$ in the computational basis, resulting in the state
    $$
    \ket{\hat\psi_{\vec y}}_{XY} = \sum_{\substack{\vec x \in \Z_q^m:\\ \vec A \vec x= \vec y \Mod{q}}} \rho_\sigma(\vec x) \ket{\vec x}_X \otimes \ket{\vec y}_Y.
    $$
\item If $b=0$, the challenger does nothing. Else, if $b=1$, the challenger measures system $X$ of the quantum state $\ket{\hat\psi_{\vec y}}$ in the computational basis. Finally, the challenger sends the outcome state in systems $X$ to $\algo A$, together with the matrix $\vec A \in \Z_q^{n \times m}$ and the string $\vec y \in \Z_q^n$.

\item $\algo A$ sends a classical witness $\vec w \in \Z_q^m$ to the challenger.

\item The challenger checks whether $\vec A \cdot \vec w = \vec y \Mod{q}$ and $\| \vec w \| \leq \sqrt{m} \sigma/\sqrt{2}$. If $\vec w$ passes both checks, the challenger sends $\vec t = (\bar{\vec x},-1) \in \Z_q^m$ to $\algo A$ with $\vec A \cdot \vec t = \vec 0 \Mod{q}$. Else, the challenger aborts.
 
\item $\algo A$ returns a bit $b'$, which we define as the output of the experiment.
\end{enumerate}
\end{conjecture}

\begin{remark}
We also consider an $N$-fold variant of $\mathsf{SGC}_{n,m,q,\sigma}$, which we denote by $\mathsf{SGC}_{n,m,q,\sigma}^N$, where the challenger prepares $N$ independent states $\ket{\hat\psi_{\vec y_1}} \otimes \dots \otimes \ket{\hat\psi_{\vec y_N}}$ in Steps $2$--$3$, for outcomes $\vec y_1,\dots,\vec y_N \in \Z_q^n$. A simple hybrid argument shows that $\mathsf{SGC}_{n,m,q,\sigma}^N$ is implied by $\mathsf{SGC}_{n,m,q,\sigma}$, for any $N = \poly(\lambda)$.
\end{remark}

\paragraph{Towards a proof of the strong-Gaussian-collapsing conjecture.}

Unfortunately, we currently do not know how to prove \expref{Conjecture}{conj:SGC} from standard assumptions, such as $\LWE$ or $\ISIS$. The difficulty emerges when we attempt to reduce the security to the $\LWE$ (or $\ISIS$) problem with respect to the same matrix $\vec A\in \Z_q^{n \times m}$. In order to simulate the experiment $\mathsf{StrongGaussCollapseExp}_{\algo H,\algo A,\lambda}$ with respect to an adversary $\algo A$, we have to eventually forward a short trapdoor vector $\vec t \in \Z^{m}$ in order to simulate the second phase of the experiment once $\algo A$ has produced a valid witness. Notice, however, that the reduction has no way of obtaining a short vector $\vec t$ in the kernel of $\vec A$ as it is trying to break the underlying $\LWE$ (or $\ISIS$) problem with respect to $\vec A$ in the first place. Therefore, any successful security proof must necessarily exploit the fact that there is \emph{interaction} between the challenger and the adversary $\algo A$, and that a short trapdoor vector $\vec t$ is only revealed \emph{after} $\algo A$ has already produced a valid short pre-image of $\vec y \in \Z_q^n$.

When trying to distinguish between the state $\ket{\hat\psi_{\vec y}}$ and a single Gaussian pre-image $\ket{\vec x_0}$ with the property that $\vec A \cdot \vec x_0 = \vec y \Mod{q}$, it is useful to work with the Fourier basis. Without loss of generality, we can assume that $\algo A$ instead receives one of the following states during in Step $3$; namely
$$
\sum_{\vec s \in \Z_q^n} \sum_{\vec e \in \Z_q^m} \rho_{\frac{q}{\sigma}}(\vec e) \, \omega_q^{-\ip{\vec s,\vec y}} \ket{\vec s \vec A + \vec e}_X \quad\,\, \text{ or } \quad\,\, \sum_{\vec u \in \Z_q^m} \omega_q^{-\ip{\vec u,\vec x_0}}\ket{\vec u}_X.
$$

One natural approach is prepare an auxiliary system, say $B$, which could later help the adversary determine whether $X$ corresponds to a superposition of $\LWE$ samples or a superposition of uniform samples once the trapdoor $\vec t$ is revealed (ideally, without disturbing $X$ so as to allow for a Fourier basis measurement).
Because finding a valid witness $\vec w$ to the $\ISIS$ problem specified by $(\vec A,\vec y)$ now amounts to a Fourier basis projection (as in \expref{Definition}{def:fourier-basis-projection}), the entropic uncertainty relation in \expref{Theorem}{thm:uncertainty} immediately rules out large class of attacks, including the \emph{shift-by-$\LWE$-sample attack} we described in \expref{Section}{sec:overview}. There, the idea is to reversibly shift system $X$ by a fresh \LWE sample into an auxiliary system $B$. 
If system $X$ corresponds to a superposition of $\LWE$ samples, we obtain a separate \LWE sample which is \emph{re-randomized}, whereas, if $X$ is a superposition of uniform samples, the outcome remains random. Hence, if the aforementioned
procedure succeeded without disturbing system $X$, we could potentially find a valid witness $\vec w$ and simultaneously distinguish the auxiliary system $B$ with access to the trapdoor $\vec t$.
As we observed before, however, such an attack must necessarily entangle the two systems $X$ and $B$ in a way that prevents it from finding a solution to the $\ISIS$ problem specified by $(\vec A,\vec y)$. Intuitively, if the state in system $X$ yields a short-pre image $\vec w$ \emph{with high probability} via a Fourier basis measurement, then system $X$ cannot be entangled with any auxiliary systems.
Because the set $\algo S$ of valid short pre-images (i.e. the set of solution to the $\ISIS$ problem specified by $\vec A$ and $\vec y$) is much smaller than the size of $\Z_q^m$ (in particular, if $\sigma \sqrt{m} \ll q$), \expref{Theorem}{thm:uncertainty} tells us that the min-entropy of system $X$ (once it is measured in the computational basis) given system $B$ must necessarily be large. 
We remark that this statement holds \emph{information-theoretically}, and does not rely on the hardness of $\LWE$. This suggests that, even if the trapdoor $\vec t$ is later revealed, system $B$ cannot contain any relevant information about whether system $X$ initially corresponded to a superposition of $\LWE$ samples, or to a superposition of uniform samples.
While this argument is not sufficient to prove \expref{Conjecture}{conj:SGC}, it captures the inherent difficulty in extracting information encoded in two mutually unbiased bases, i.e. the computational basis and the Fourier basis.

\section{Public-Key Encryption with Certified Deletion}

In this section, we formalize the notion of public-key encryption with certified deletion.

\subsection{Definition}

In this work, we consider public-key encryption schemes with certified deletion for which verification of a
deletion certificate is \emph{public}; meaning anyone with access to the verification key can verify that deletion has taken place.
We first introduce the following definition.

\begin{definition}[Public-key encryption with certified deletion]\label{def:PKE-CD} A public-key encryption scheme with certified deletion ($\PKECD$) $\Sigma = (\KeyGen,\Enc,\Dec,\Del,\Vrfy)$ with plaintext space $\algo M$ consists of a tuple of $\QPT$ algorithms, a key generation algorithm $\KeyGen$, an encryption algorithm $\Enc$, and a decryption algorithm $\Dec$, a deletion algorithm $\Del$, and a verification algorithm $\Vrfy$.
\begin{description}
\item $\KeyGen(1^\lambda) \rightarrow (\pk,\sk):$ takes as input the parameter $1^\lambda$ and outputs a public key $\pk$ and secret key $\sk$.
\item $\Enc(\pk,m) \rightarrow (\vk,\ct):$ takes as input the public key $\pk$ and a plaintext $m \in \algo M$, and
outputs a classical (public) verification key $\vk$ together with a quantum ciphertext $\ct$.
\item $\Dec(\sk,\ct) \rightarrow m'\, \mathbf{or}\,\bot:$ takes as input the secret key $\sk$ and ciphertext $\ct$, and outputs $m'\in \algo M$ or $\bot$.
\item $\Del(\ct) \rightarrow \pi:$ takes as input a ciphertext $\ct$ and outputs a classical certificate $\pi$.
\item $\Vrfy(\vk,\pi) \rightarrow \top \, \mathbf{or} \, \bot:$ takes as input the verification key $\vk$ and certificate $\pi$, and outputs $\top$ or $\bot$.
\end{description}
\end{definition}

\begin{definition}[Correctness of $\PKECD$] We require two separate kinds of correctness properties, one for decryption and one for verification.
\begin{description}
\item (Decryption correctness:)
For any $\lambda \in \mathbb{N}$, and for any $m \in \algo M$:
$$
\Pr \left[ \Dec(\sk,\ct) \neq m \, \bigg| \, \substack{
(\pk,\sk) \leftarrow \KeyGen(1^\lambda)\\
\ct \leftarrow \Enc(\pk,m)
}\right] \, \leq \, \negl(\lambda).
$$
\item (Verification correctness:)
For any $\lambda \in \mathbb{N}$, and for any $m \in \algo M$:
$$
\Pr \left[ \Vrfy(\vk,\pi) = \bot \, \bigg| \, \substack{
(\pk,\sk) \leftarrow \KeyGen(1^\lambda)\\
(\vk,\ct) \leftarrow \Enc(\pk,m)\\
\pi \leftarrow \Del(\ct)
}\right] \, \leq \, \negl(\lambda).
$$
\end{description}
\end{definition}

The notion of $\INDCPACD$ security for public-key encryption was first introduced by Hiroka, Morimae, Nishimaki and Yamakawa~\cite{hiroka2021quantum}.

\subsection{Certified deletion security}

In terms of security, we adopt the following definition.

\begin{definition}[Certified deletion security for $\PKE$]\label{def:CD-security} Let $\Sigma = (\KeyGen,\Enc,\Dec,\Del,\Vrfy)$ be a $\PKECD$ scheme and let $\algo A$ be a $\QPT$ adversary (in terms of the security parameter $\lambda \in \N$). We define the security experiment $\Exp^{\mathsf{pk\mbox{-}cert\mbox{-}del}}_{\Sigma,\algo A,\lambda}(b)$ between $\algo A$ and a challenger as follows:
\begin{enumerate}
    \item The challenger generates a pair $(\pk,\sk) \from \KeyGen(1^\lambda)$, and sends $\pk$ to $\algo A$.
    \item $\algo A$ sends a plaintext pair $(m_0,m_1) \in \algo M \times \algo M$ to the challenger.
    \item The challenger computes $(\vk,\ct_b) \leftarrow \Enc(\pk,m_b)$, and sends $\ct_b$ to $\algo A$.
    \item At some point in time, $\algo A$ sends the certificate $\pi$ to the challenger.
    \item The challenger computes $\Vrfy(\vk, \pi)$ and sends $\sk$ to $\algo A$, if the output is $\top$, and sends $\bot$ otherwise.
    \item $\algo A$ outputs a guess $b' \in \bit$, which is also the output of the experiment.
\end{enumerate}
We say that the scheme $\Sigma$ is $\INDCPACD$-secure if, for any $\QPT$ adversary $\algo A$, it holds that
$$
\Adv_{\Sigma,\algo A}^{\mathsf{pk}\mbox{-}\mathsf{cert}\mbox{-}\mathsf{del}}(\lambda) := |\Pr[\Exp^{\mathsf{pk}\mbox{-}\mathsf{cert}\mbox{-}\mathsf{del}}_{\Sigma,\algo A,\lambda}(0)=1] - \Pr[\Exp^{\mathsf{pk}\mbox{-}\mathsf{cert}\mbox{-}\mathsf{del}}_{\Sigma,\algo A,\lambda}(1) = 1] |
 \leq \negl(\lambda).$$
\end{definition}

\section{Dual-Regev Public-Key Encryption with Certified Deletion}\label{sec:Dual-Regev-PKE}

In this section, we consider the Dual-Regev $\PKE$ scheme due to Gentry, Peikert and Vaikuntanathan \cite{cryptoeprint:2007:432}. Unlike Regev's original $\PKE$ scheme in~\cite{Regev05}, the Dual-Regev $\PKE$ scheme has the useful property that the ciphertext takes the form of a regular sample from the $\LWE$ distribution together with an additive shift which depends on the plaintext.

\subsection{Construction}

\paragraph{Parameters.} Let $\lambda \in \N$ be the security parameter. We choose the following set of parameters for our Dual-Regev $\PKE$ scheme with certified deletion (each parameterized by $\lambda$).
\begin{itemize}
\item an integer $n \in \N$.
    \item a prime modulus $q \geq 2$.
    \item an integer $m \geq  2n \log q$.
    \item a noise ratio $\alpha \in (0,1)$ such that $ \sqrt{8(m+1)} \leq \frac{1}{\alpha} \leq \frac{q}{\sqrt{8(m+1)}}$.
\end{itemize}

\begin{construction}[Dual-Regev $\PKE$ with Certified Deletion]\label{cons:dual-regev-cd}
Let $\lambda \in \N$.
The Dual-Regev $\PKE$ scheme $\DualPKECD = (\KeyGen,\Enc,\Dec,\Del,\Vrfy)$ with certified deletion is defined as follows:
\begin{description}
\item $\KeyGen(1^\lambda) \rightarrow (\pk,\sk):$ sample a random matrix $\bar{\vec A} \rand \Z_q^{n\times m}$ and a vector $\bar{\vec x} \rand \bit^{m}$
and choose $\vec A = [\bar{\vec A} | \bar{\vec A} \cdot \bar{\vec x} \Mod{q}]$.
Output $(\pk,\sk)$, where $\pk=\vec A \in \Z_q^{n \times (m+1)}$ and $\sk = (-\bar{\vec x}, 1) \in \Z_q^{m+1}$.
\item $\Enc(\pk,x) \rightarrow (\vk,\ket{\ct})$: parse $\vec A \leftarrow \pk$ and run $(\ket{\psi_{\vec y}},\vec y) \leftarrow \mathsf{GenPrimal}(\vec A,1/\alpha)$ in \expref{Algorithm}{alg:GenPrimal}, where $\vec y \in \Z_q^n$. To encrypt a single bit $b \in \bit$, output the pair
$$
\left(\vk \leftarrow (\vec A \in \Z_q^{n\times (m+1)}, \vec y  \in \Z_q^n), \quad \ket{\ct} \leftarrow \vec X_q^{(0,\dots,0, b \cdot  \lfloor\frac{q}{2} \rfloor)} \ket{\psi_{\vec y}} \right),
$$
where $\vk$ is the public verification key and $\ket{\ct}$ is an $(m+1)$-qudit quantum ciphertext.

\item $\Dec(\sk,\ket{\ct}) \rightarrow \bit:$ to decrypt, measure the ciphertext $\ket{\ct}$ in the computational basis with outcome $\vec c \in \Z_q^m$. Compute $\vec c^T \cdot \sk \in \Z_q$ and output $0$, if itis closer to $0$ than to $\lfloor\frac{q}{2}\rfloor$, and output $1$, otherwise.

\item $\Del(\ket{\ct}) \rightarrow \pi:$ Measure $\ket{\ct}$ in the Fourier basis and output the measurement outcome $\pi \in \Z_q^{m+1}$.
\item $\Vrfy(\vk,\pi) \rightarrow \{\top,\bot\}:$ to verify a deletion certificate $\pi \in \Z_q^{m+1}$, parse $(\vec A,\vec y) \leftarrow \vk$ and output $\top$, if $\vec A \cdot \pi = \vec y \Mod{q}$ and $\| \pi \| \leq \sqrt{m+1}/\sqrt{2}\alpha$, and output $\bot$, otherwise.
\end{description}
\end{construction}

\paragraph{Proof of correctness.}
Let us now establish the correctness properties of $\DualPKECD$ in \expref{Construction}{cons:dual-regev-cd}.

\begin{lemma}[Correctness of decryption]\label{lem:correctness-decryption-dual-regev-pke-cd}
Let $n\in \N$ and $q\geq 2$ be a prime modulus with $m \geq 2n \log q$, each parameterized by the security parameter $\lambda \in \N$. Let $\alpha$ be a ratio with $ \sqrt{8(m+1)} \leq \frac{1}{\alpha} \leq \frac{q}{\sqrt{8(m+1)}}$.
Then, for $b \in \bit$, the scheme $\DualPKECD = (\KeyGen,\Enc,\Dec,\Del,\Vrfy)$ in \expref{Construction}{cons:dual-regev-cd} satisfies:
    $$
    \Pr \left[\Dec(\sk,\ket{\ct}) = b \, \bigg| \, \substack{
    (\pk,\sk) \leftarrow \KeyGen(1^\lambda)\\
    (\vk,\ket{\ct}) \leftarrow \Enc(\pk,b)
    }\right] \, \geq \, 1 - \negl(\lambda).
    $$
\end{lemma}
\begin{proof}
By the Leftover Hash Lemma (\expref{Lemma}{lem:LHL}), the distribution of $\vec A = [\bar{\vec A} | \bar{\vec A} \cdot \bar{\vec x} \Mod{q}]$ is within negligible total variation distance of the uniform distribution over $\Z_q^{n \times (m+1)}$. Moreover, from \expref{Lemma}{lem:full-rank} it follows that the columns of $\vec A$ generate $\Z_q^n$ with overwhelming probability. Since the noise ratio $\alpha \in (0,1)$ satisfies $ \sqrt{8(m+1)} \leq \frac{1}{\alpha} \leq \frac{q}{\sqrt{8(m+1)}}$, it then follows from \expref{Corollary}{cor:switching} that the ciphertext $\ket{\ct}$ is within negligible trace distance of the state
$$
\sum_{\vec s \in \Z_q^n} \sum_{\vec e \in \Z_q^{m+1}} \rho_{\alpha q}(\vec e) \, \omega_q^{-\ip{\vec s,\vec y}} \ket{\vec s \vec A + \vec e + (0,\dots,0, b \cdot  \lfloor\frac{q}{2}\rfloor)} 
$$
A measurement in computational basis yields an outcome $\vec c$ such that
$$
\vec c = \vec s_0 \vec A + \vec e_0 + (0,\dots,0, b \cdot  \lfloor\frac{q}{2}\rfloor) \in \Z_q^{m+1},
$$
where $\vec s_0 \rand \Z_q^n$ is random and where $\vec e_0 \sim D_{\Z_q^{m+1},\frac{\alpha q}{\sqrt{2}}}$ is a sample from the (truncated) discrete Gaussian such that $\| \vec e_0\| \leq \alpha q\sqrt{\frac{m+1}{2}} < \lfloor\frac{q}{4}\rfloor$. Since $\Dec(\sk,\ket{\ct})$ computes $\vec c^T \cdot \sk \in \Z \cap (-\frac{q}{2},\frac{q}{2}]$ and outputs $0$, if it
is closer to $0$ than to $\lfloor\frac{q}{2}\rfloor$ over , and $1$ otherwise, it succeeds with overwhelming probability.

\end{proof}

Let us now prove the following property.

\begin{lemma}[Correctness of verification]\label{lem:correctness-verification-dual-regev-pke-cd}
Let $n\in \N$ and $q\geq 2$ be a prime modulus with $m \geq 2n \log q$, each parameterized by the security parameter $\lambda \in \N$. Let $\alpha$ be a ratio with $ \sqrt{8(m+1)} \leq \frac{1}{\alpha} \leq \frac{q}{\sqrt{8(m+1)}}$.
Then, for $b \in \bit$, the scheme $\DualPKECD = (\KeyGen,\Enc,\Dec,\Del,\Vrfy)$ in \expref{Construction}{cons:dual-regev-cd} satisfies:
    $$
    \Pr \left[\Verify(\vk, \pi) =  \top \, \bigg| \, \substack{
    (\pk,\sk) \leftarrow \KeyGen(1^\lambda)\\
    (\vk,\ket{\ct}) \leftarrow \Enc(\pk,b)\\
    \pi \leftarrow \Del(\ket{\ct})
    }\right] \, \geq \, 1 - \negl(\lambda).
    $$
\end{lemma}

\begin{proof} By the Leftover Hash Lemma (\expref{Lemma}{lem:LHL}), the distribution of $\vec A = [\bar{\vec A} | \bar{\vec A} \cdot \bar{\vec x} \Mod{q}]$ is within negligible total variation distance of the uniform distribution over $\Z_q^{n \times (m+1)}$. From \expref{Lemma}{lem:full-rank} it follows that the columns of $\vec A$ generate $\Z_q^n$ with overwhelming probability. Since $\alpha \in (0,1)$ is a ratio parameter with $\sqrt{8(m+1)} \leq \frac{1}{\alpha} \leq \frac{q}{\sqrt{8(m+1)}}$, \expref{Corollary}{cor:switching} implies that the Fourier transform of the ciphertext $\ket{\ct}$ is within negligible trace distance of the state
$$
\ket{\widehat{\ct}}=\sum_{\substack{\vec x \in \Z_q^{m+1}:\\ \vec A \vec x = \vec y \Mod{q}}}\rho_{1/\alpha}(\vec x) \, \omega_q^{\ip{\vec x,(0,\dots,0, b \cdot  \lfloor\frac{q}{2} \rfloor)}} \,\ket{\vec x}.
$$
From \expref{Lemma}{lem:tailboundII}, it follows that the distribution of computational basis measurement outcomes is within negligible total variation distance of $\pi \sim D_{\Lambda_q^{\vec y}(\vec A),\frac{1}{\sqrt{2}\alpha}}$ with $\| \pi\| \leq \sqrt{m+1}/\sqrt{2}\alpha$.
This proves the claim.
\end{proof} 

\subsection{Proof of security}\label{sec:DualRegev-PKE-CD}

Let us now analyze the security of our Dual-Regev $\PKE$ scheme with certified deletion in \expref{Construction}{cons:dual-regev-cd}.

\paragraph{$\INDCPA$ security of $\DualPKECD$.}
We first prove that our public-key encryption scheme $\DualPKECD$ in \expref{Construction}{cons:dual-regev-cd} satisfies the notion $\INDCPA$ security according to \expref{Definition}{def:ind-cpa}. The proof follows from \expref{Theorem}{thm:pseudorandom-SLWE} and assumes the hardness of (decisional) $\LWE$ (\expref{Definition}{def:decisional-lwe}). We add it for completeness.

\begin{theorem}\label{thm:Dual-Regev-PKE-CPA} Let $n\in \N$ and $q\geq 2$ be a prime modulus with $m \geq 2n \log q$, each parameterized by the security parameter $\lambda \in \N$. Let $\alpha \in (0,1)$ be a noise ratio parameter with $\sqrt{8(m+1)} \leq \frac{1}{\alpha} \leq \frac{q}{\sqrt{8(m+1)}}$.
Then, the scheme $\DualPKECD$ in \expref{Construction}{cons:dual-regev-cd} is $\INDCPA$-secure assuming the quantum hardness of the decisional $\LWE_{n,q,\beta q}^m$ problem, for any $\beta \in (0,1)$ with $\alpha/\beta= \lambda^{\omega(1)}$.
\end{theorem}
\begin{proof}
Let $\Sigma = \DualPKECD$. We need to show that, for any $\QPT$ adversary $\algo A$, it holds that
$$
\Adv_{\Sigma,\algo A}(\lambda) := |\Pr[\Exp^{\mathsf{ind\mbox{-}cpa}}_{\Sigma,\algo A,\lambda}(0)=1] - \Pr[\Exp^{\mathsf{ind\mbox{-}cpa}}_{\Sigma,\algo A,\lambda}(1) = 1] |
 \leq \negl(\lambda).$$
Consider the experiment $\Exp^{\mathsf{ind\mbox{-}cpa}}_{\Sigma,\algo A,\lambda}(b)$ between the adversary $\algo A$ and a challenger taking place as follows:
\begin{enumerate}
    \item The challenger generates a pair $(\pk,\sk) \from \KeyGen(1^\lambda)$, and sends $\pk$ to $\algo A$.
    \item $\algo A$ sends a distinct plaintext pair $(m_0,m_1) \in \bit \times \bit$ to the challenger.
    \item The challenger computes $(\vk,\ct_b) \leftarrow \Enc(\pk,m_b)$, and sends $\ket{\ct_b}$ to $\algo A$.
    \item $\algo A$ outputs a guess $b' \in \bit$, which is also the output of the experiment.
\end{enumerate}
Recall that the procedure $\Enc(\pk,m_b)$ outputs a pair $(\vk,\ket{\ct_b})$, where $(\vec A \in \Z_q^{n\times (m+1)}, \vec y  \in \Z_q^n) \leftarrow \vk$ is the verification key and where the ciphertext $\ket{\ct_b}$ is within negligible trace distance of
\begin{align}\label{eq:ct-DualRegevPKECD-security}
\sum_{\vec s \in \Z_q^n} \sum_{\vec e \in \Z_q^{m+1}} \rho_{\alpha q}(\vec e) \, \omega_q^{-\ip{\vec s,\vec y}} \ket{\vec s \vec A + \vec e + (0,\dots,0, m_b \cdot  \lfloor q/2 \rfloor) \Mod{q}} 
\end{align}
Let $\beta \in (0,1)$ be such that $\alpha/\beta= \lambda^{\omega(1)}$. From \expref{Theorem}{thm:pseudorandom-SLWE} it follows that, under the (decisional) $\LWE_{n,q,\beta q}^{m}$ assumption, the quantum ciphertext $\ket{\ct_b}$ is computationally indistinguishable from the state
\begin{align}\label{eq:random-state}
\sum_{\vec u \in \Z_q^{m+1}} \omega_q^{-\ip{\vec u,\vec x_0}}\ket{\vec u}, \quad \,\, \vec x_0 \sim D_{\Lambda_q^{\vec y}(\vec A),\frac{1}{\sqrt{2}\alpha}}.
\end{align}
Because the state in Eq.~\eqref{eq:random-state} is completely independent of $b \in \bit$, it follows that
$$
\Adv_{\Sigma,\algo A}(\lambda) := |\Pr[\Exp^{\mathsf{ind\mbox{-}cpa}}_{\Sigma,\algo A,\lambda}(0)=1] - \Pr[\Exp^{\mathsf{ind\mbox{-}cpa}}_{\Sigma,\algo A,\lambda}(1) = 1] |
 \leq \negl(\lambda).$$
This proves the claim.
\end{proof}

\paragraph{$\INDCPACD$ security of \DualPKECD.} In this section, we prove that our public-key encryption scheme $\DualPKECD$ in \expref{Construction}{cons:dual-regev-cd} satisfies the notion of \emph{certified deletion security} assuming the \emph{Strong Gaussian-Collapsing (SGC) Conjecture} (see \expref{Conjecture}{conj:SGC}). This is a strengthening of the Gaussian-collapsing property which we proved under the (decisional) \LWE assumption (see \expref{Theorem}{thm:GaussCollapse}).

\begin{theorem}\label{thm:Dual-Regev-PKE-CD} Let $n\in \N$ and $q\geq 2$ be a prime modulus with $m \geq 2n \log q$, each parameterized by $\lambda \in \N$. Let   Let $\alpha$ be a ratio with $\sqrt{8(m+1)} \leq \frac{1}{\alpha} \leq \frac{q}{\sqrt{8(m+1)}}$. Then, the scheme $\DualPKECD$ in \expref{Construction}{cons:dual-regev-cd} is $\INDCPACD$-secure assuming the Strong Gaussian-Collapsing property $\mathsf{SGC}_{n,m+1,q,\frac{1}{\alpha}}$ from \expref{Conjecture}{conj:SGC}.
\end{theorem}

\begin{proof}
Let $\Sigma = \DualPKECD$. We need to show that, for any $\QPT$ adversary $\algo A$, it holds that
$$
\Adv_{\Sigma,\algo A}^{\mathsf{pk\mbox-}\mathsf{cert}\mbox{-}\mathsf{del}}(\lambda) := |\Pr[\Exp^{\mathsf{pk\mbox{-}cert\mbox{-}del}}_{\Sigma,\algo A,\lambda}(0)=1] - \Pr[\Exp^{\mathsf{pk\mbox{-}cert\mbox{-}del}}_{\Sigma,\algo A,\lambda}(1) = 1] |
 \leq \negl(\lambda).$$
 
We consider the following sequence of hybrids:

\begin{description}
\item $\mathbf{H_0:}$ This is the experiment $\Exp^{\mathsf{pk\mbox{-}cert\mbox{-}del}}_{\Sigma,\algo A,\lambda}(0)$ between $\algo A$ and a challenger:
\begin{enumerate}
\item The challenger samples a random matrix $\bar{\vec A} \rand \Z_q^{n\times m}$ and a vector $\bar{\vec x} \rand \bit^{m}$ and chooses $\vec A = [\bar{\vec A} | \bar{\vec A} \cdot \bar{\vec x} \Mod{q}]$.
The challenger chooses the secret key $\sk \leftarrow (-\bar{\vec x}, 1) \in \Z_q^{m+1}$ and the public key $\pk \leftarrow \vec A \in \Z_q^{n \times (m+1)}$.

\item $\algo A$ sends a distinct plaintext pair $(m_0,m_1) \in \bit \times \bit$ to the challenger. (Note: Without loss of generality, we can just assume that $m_0 = 0$ and $m_1=1$).
    
\item The challenger runs $(\ket{\psi_{\vec y}},\vec y) \leftarrow \mathsf{GenPrimal}(\vec A,1/\alpha)$ in \expref{Algorithm}{alg:GenPrimal}, and outputs
$$
\left(\vk \leftarrow (\vec A \in \Z_q^{n\times (m+1)}, \vec y  \in \Z_q^n), \quad \ket{\ct_0}  \leftarrow \ket{\psi_{\vec y}} \right).
$$

\item At some point in time, $\algo A$ returns a certificate $\pi$ to the challenger.

\item The challenger verifies $\pi$ and outputs $\top$, if $\vec A \cdot \pi = \vec y \Mod{q}$ and $\| \pi \| \leq \sqrt{m+1}/\sqrt{2}\alpha$, and output $\bot$, otherwise. If $\pi$ passes the test with outcome $\top$, the challenger sends $\sk$ to $\algo A$.
    
\item $\algo A$ outputs a guess $b' \in \bit$, which is also the output of the experiment.
\end{enumerate}

\item $\mathbf{H_1:}$ This is same experiment as in $\mathbf{H_0}$, except that (in Step 3) the challenger prepares the ciphertext in the Fourier basis rather than the standard basis. In other words, $\algo A$ receives the pair
$$
\left(\vk \leftarrow (\vec A \in \Z_q^{n\times (m+1)}, \vec y  \in \Z_q^n), \quad \ket{\ct_0}  \leftarrow \FT_q \ket{\psi_{\vec y}} \right).
$$

\item $\mathbf{H_2:}$ This is the experiment $\mathsf{StrongGaussCollapseExp}_{\algo H,\algo D,\lambda}(0)$ in \expref{Conjecture}{conj:SGC}:
\begin{enumerate}
\item The challenger samples a random matrix $\bar{\vec A} \rand \Z_q^{n\times m}$ and a vector $\bar{\vec x} \rand \bit^{m}$ and chooses $\vec A = [\bar{\vec A} | \bar{\vec A} \cdot \bar{\vec x} \Mod{q}]$ and $\vec t = (-\bar{\vec x}, 1) \in \Z_q^{m+1}$.
    
\item The challenger runs $(\ket{\hat\psi_{\vec y}},\vec y) \leftarrow \mathsf{GenDual}(\vec A,\sigma)$ in \expref{Algorithm}{alg:GenDual}, where $\vec y \in \Z_q^n$, and sends the triplet
$(\ket{\hat\psi_{\vec y}},\vec A,\vec y)$ to the adversary $\algo A$.

\item At some point in time, $\algo A$ returns a certificate $\pi$ to the challenger.

\item The challenger verifies $\pi$ and outputs $\top$, if $\vec A \cdot \pi = \vec y \Mod{q}$ and $\| \pi \| \leq \sqrt{m+1}/\sqrt{2}\alpha$, and output $\bot$, otherwise. If $\pi$ passes the test with outcome $\top$, the challenger sends $\vec t$ to $\algo A$.
    
\item $\algo A$ outputs a guess $b' \in \bit$, which is also the output of the experiment.
\end{enumerate}

\item $\mathbf{H_3:}$ This is the experiment $\mathsf{StrongGaussCollapseExp}_{\algo H,\algo D,\lambda}(1)$ in \expref{Conjecture}{conj:SGC}; it is the same as $\mathbf{H_2}$, except that the state $\ket{\hat\psi_{\vec y}}$ (in Step 2) is measured in the computational basis before it is sent to $\algo A$. 

\item $\mathbf{H_4:}$ This is same experiment as $\mathbf{H_3}$, except that (in Step 2) the challenger additionally applies the Pauli operator $\vec Z_q^{(0,\dots,0,\lfloor\frac{q}{2} \rfloor)}$ to the state $\ket{\hat\psi_{\vec y}}$ before it is measured in the computational basis.

\item $\mathbf{H_5:}$ This is same experiment as $\mathbf{H_4}$, except that (in Step 2) $\algo A$ receives the triplet 
$$
(\vec Z_q^{(0,\dots,0,\lfloor\frac{q}{2} \rfloor)}\ket{\hat\psi_{\vec y}},\quad\vec A \in \Z_q^{n \times (m+1)},\quad\vec y \in \Z_q^n).
$$

\item $\mathbf{H_6:}$ This is same experiment as $\mathbf{H_5}$, except that (in Step 2) the challenger prepares the quantum state $\vec Z_q^{(0,\dots,0,\lfloor\frac{q}{2} \rfloor)}\ket{\hat\psi_{\vec y}}$ in the (inverse) Fourier basis instead. In other words, $\algo A$ receives the triplet 
$$
(\FT_q^\dag\vec Z_q^{(0,\dots,0,\lfloor\frac{q}{2} \rfloor)}\ket{\hat\psi_{\vec y}},\quad\vec A \in \Z_q^{n \times (m+1)},\quad\vec y \in \Z_q^n).
$$

\item $\mathbf{H_7:}$ This is the experiment $\Exp^{\mathsf{pk\mbox{-}cert\mbox{-}del}}_{\Sigma,\algo A,\lambda}(1)$.
\end{description}

We now show that the hybrids are indistinguishable.

%% H0 vs H1
\begin{claim}
$$
 \Pr[\Exp^{\mathsf{pk\mbox{-}cert\mbox{-}del}}_{\Sigma,\algo A,\lambda}(0)=1] = \Pr[\mathbf{H_1} = 1].$$
\end{claim}
\begin{proof}
Without loss of generality, we can assume that the challenger applies the inverse Fourier transform before sending the ciphertext to $\algo A$. Therefore, the success probabilities are identical in $\mathbf{H_0}$ and $\mathbf{H_1}$. 
\end{proof}

%% H2 vs H1
\begin{claim}
$$
\Pr[\mathbf{H_1} = 1] = \Pr[\mathbf{H_2} = 1].$$
\end{claim}
\begin{proof}
Because the challenger in $\mathbf{H_1}$ always sends the ciphertext $\ket{\ct_0}$ corresponding to $m_0=0$ to the adversary $\algo A$, the two hybrids $\mathbf{H_1}$ and $\mathbf{H_2}$ are identical.
\end{proof}

%% H3 vs H2
\begin{claim} Under the Strong Gaussian-Collapsing property $\mathsf{SGC}_{n,m+1,q,\frac{1}{\alpha}}$, it holds that
$$
 | \Pr[\mathbf{H_2} = 1] - \Pr[\mathbf{H_3} = 1] |
 \leq \negl(\lambda).$$
\end{claim}
\begin{proof}
This follows directly from \expref{Conjecture}{conj:SGC}.
\end{proof}

%% H4 vs H3
\begin{claim}
$$
\Pr[\mathbf{H_3} = 1] = \Pr[\mathbf{H_4} = 1].$$
\end{claim}
\begin{proof}
Because the challenger measures the state $\ket{\hat\psi_{\vec y}}$ in Step 2 in the computational basis, applying the phase operator $\vec Z_q^{(0,\dots,0,\lfloor\frac{q}{2} \rfloor)}$ before the measurement does not affect the measurement outcome.
\end{proof}

%% H4 vs H5
\begin{claim} Under the Strong Gaussian-Collapsing property $\mathsf{SGC}_{n,m+1,q,\frac{1}{\alpha}}$, it holds that
$$
 | \Pr[\mathbf{H_4} = 1] - \Pr[\mathbf{H_5} = 1] |
 \leq \negl(\lambda).$$
\end{claim}
\begin{proof}
This follows from \expref{Conjecture}{conj:SGC} since, without loss of generality, we can assume that the challenger applies the phase operator $\vec Z_q^{(0,\dots,0,\lfloor\frac{q}{2} \rfloor)}$ before sending the state $\ket{\hat\psi_{\vec y}}$ to $\algo A$.
\end{proof}

%% H5 vs H6
\begin{claim}
$$
\Pr[\mathbf{H_5} = 1] = \Pr[\mathbf{H_6} = 1].
$$
\end{claim}

\begin{proof}
Without loss of generality, we can assume that the challenger applies the Fourier transform to $\vec Z_q^{(0,\dots,0,\lfloor\frac{q}{2} \rfloor)}\ket{\hat\psi_{\vec y}}$ before sending it to $\algo A$. Therefore, the success probabilities are identical in $\mathbf{H_5}$ and $\mathbf{H_6}$. 
\end{proof}

%% H6 vs H7
\begin{claim}
$$
|\Pr[\mathbf{H_6} = 1] -\Pr[\Exp^{\mathsf{pk\mbox{-}cert\mbox{-}del}}_{\Sigma,\algo A,\lambda}(1)=1]| \leq \negl(\lambda).$$
\end{claim}
\begin{proof}
From \expref{Lemma}{lem:XZ-conjugation}, we have $\FT_q \vec X_q^{\vec v} = \vec Z_q^{\vec v} \FT_q$, for all $\vec v \in \Z_q^m$. Hence, in $\mathbf{H_6}$, we can instead assume that the challenger runs $(\ket{\psi_{\vec y}},\vec y) \leftarrow \mathsf{GenPrimal}(\vec A,1/\alpha)$ in \expref{Algorithm}{alg:GenPrimal} and sends the following to $\algo A$:
$$
\left(\vk \leftarrow (\vec A \in \Z_q^{n\times (m+1)}, \vec y  \in \Z_q^n), \quad \ket{\ct_1}  \leftarrow \vec X_q^{(0,\dots,0,\lfloor\frac{q}{2} \rfloor)}\ket{\psi_{\vec y}} \right).
$$
From \expref{Corollary}{cor:switching}, we have that $\FT_q^\dag \vec Z_q^{\vec v} \ket{\hat\psi_{\vec y}}$ and $\vec X_q^{\vec v} \ket{{\psi_{\vec y}}}$ are within negligible trace distance, for all $\vec v \in \Z_q^m$. Because the challenger in $\mathbf{H_7}$ always sends the ciphertext $\ket{\ct_1}$ corresponding to $m_1=1$ to the adversary $\algo A$, it follows that the distinguishing advantage between $\mathbf{H_6}$ and $\mathbf{H_7}=\Exp^{\mathsf{pk\mbox{-}cert\mbox{-}del}}_{\Sigma,\algo A,\lambda}(1)$ is negligible.
\end{proof} 
Because the hybrids $\mathbf{H_0}$ and $\mathbf{H_7}$ are indistinguishable, this implies that
$$
\Adv_{\Sigma,\algo A}^{\mathsf{pk\mbox-}\mathsf{cert}\mbox{-}\mathsf{del}}(\lambda)\leq \negl(\lambda).$$
\end{proof}

Next, we show how to extend our Dual-Regev $\PKE$ scheme with certified deletion in \expref{Construction}{cons:dual-regev-cd} to a fully homomorphic encryption scheme of the same type.

\section{Fully Homomorphic Encryption with Certified Deletion}

In this section, we formalize the notion of homomorphic encryption with certified deletion which enables an untrusted quantum server to compute on encrypted data and, if requested, to simultaneously prove data deletion to a client. We also provide several notions of certified deletion security.

\subsection{Definition}

We begin with the following definition.

\begin{definition}[Homomorphic encryption with certified deletion]\label{def:INDCPACD} A homomorphic encryption scheme with certified deletion is a tuple $\HECD = (\KeyGen,\Enc,\Dec,\Eval,\Del,\Vrfy)$ of $\QPT$ algorithms (in the security parameter $\lambda \in \N$), a key generation algorithm $\KeyGen$, an encryption algorithm $\Enc$, a decryption algorithm $\Dec$, an evaluation algorithm $\Eval$, a deletion algorithm $\Del$, and a verification algorithm $\Vrfy$.
\begin{description}
\item $\KeyGen(1^\lambda) \rightarrow (\pk,\sk):$ takes as input $1^\lambda$ and outputs a public key $\pk$ and secret key $\sk$.
\item $\Enc(\pk,x) \rightarrow (\vk,\ct):$ takes as input the public key $\pk$ and a plaintext $x \in \bit$, and
outputs a classical verification key $\vk$ together with a quantum ciphertext $\ct$.
\item $\Dec(\sk,\ct) \rightarrow x'\, \mathbf{or}\,\bot:$ takes as input a key $\sk$ and ciphertext $\ct$, and outputs $x'\in \bit$ or $\bot$.
\item $\Eval(C,\ct,\pk) \rightarrow \widetilde{\ct}$: takes as input a key $\pk$ and applies a circuit $C: \bit^\ell \rightarrow \bit$ to a product of quantum ciphertexts $\ct = \ct_1 \otimes \dots \otimes \ct_\ell$ resulting in a state $\widetilde{\ct}$.
\item $\Del(\ct) \rightarrow \pi:$ takes as input a ciphertext $\ct$ and outputs a classical certificate $\pi$.
\item $\Vrfy(\vk,\pi) \rightarrow \top \, \mathbf{or} \, \bot:$ takes as input a key $\vk$ and certificate $\pi$, and outputs $\top$ or $\bot$.
\end{description}
\end{definition}

We remark that we frequently overload the functionality of the encryption and decryption procedures by allowing both procedures to take multi-bit messages as input, and to generate or decrypt a sequence of quantum ciphertexts bit-by-bit.

\begin{definition}[Compactness and full homomorphism]\label{def:compact} A homomorphic encryption scheme with certified deletion $\HECD = (\KeyGen,\Enc,\Dec,\Eval,\Del,\Vrfy)$ is fully homomorphic if, for any efficienty (in $\lambda \in \N$) computable circuit $C: \bit^\ell \rightarrow \bit$ and any set of inputs $x = (x_1,\dots,x_\ell) \in \bit^\ell$, it holds that
$$
\Pr\left[\Dec(\sk,\widetilde{\ct}) \neq C(x_1,\dots,x_\ell)
\, \bigg| \, \substack{
(\pk,\sk) \leftarrow \KeyGen(1^\lambda)\\
(\vk,\ct) \leftarrow \Enc(\pk,x)\\
\widetilde{\ct}\from \Eval(C,\ct, \pk)
}\right] \leq \negl(\lambda).
$$
We say that a fully homomorphic encryption scheme with certified deletion $(\FHECD)$ is compact if its decryption circuit is independent of
the circuit $C$. The scheme is leveled fully homomorphic if it takes $1^L$ as an additional input for its key generation procedure and can only evaluate depth $L$ Boolean circuits.

\end{definition}

\begin{definition}[Correctness of verification]\label{def:correctness-verification-HE} A homomorphic encryption scheme with certified deletion $\HECD = (\KeyGen,\Enc,\Dec,\Eval,\Del,\Vrfy)$ has correctness of verification if the following property holds for any integer $\lambda \in \mathbb{N}$ and any set of inputs $x = (x_1,\dots,x_\ell) \in \bit^\ell$
$$
\Pr \left[\Vrfy(\vk,\pi) = \bot \, \bigg| \, \substack{
(\pk,\sk) \leftarrow \KeyGen(1^\lambda)\\
(\vk,\ct) \leftarrow \Enc(\pk,x)\\
\pi \from \Del(\ct)
}\right] \, \leq \, \negl(\lambda).
$$
\end{definition}

Recall that a fully homomorphic encryption scheme with certified deletion enables
an untrusted quantum server to compute on encrypted data and to also prove data
deletion to a client. In this context, it is desirable for the client to be able to \emph{extract} (i.e., to decrypt) the outcome of the computation without irreversibly affecting the ability of the server to later prove deletion. We use the following definition.

\begin{definition}[Extractable $\FHE$ scheme with certified deletion]\label{def:extractable-FHE}
A fully homomorphic encryption scheme with certified deletion $\Sigma = (\KeyGen,\Enc,\Dec,\Eval,\Extract,\Del,\Vrfy)$ is called extractable, if
\begin{itemize}
    \item $\Eval(C,\ct_1,\dots,\ct_\ell, \pk)$ additionally outputs a circuit transcript $t_C$ besides $\widetilde{\ct}$;
    \item $\Extract\ip{\algo S(\rho,t_C),\algo R(\sk)}$ is an interactive protocol between a sender $\algo S$ (which takes as input a state $\rho$ and a circuit transcript $t_C$) and a receiver $\algo R$ (which takes as input a key $\sk$) with the property that, once the protocol is complete, $\algo S$ obtains a state $\widetilde{\rho}$ and $\algo R$ obtains a bit $y \in \bit$;
\end{itemize}
such that for any efficiently computable circuit $C: \bit^\ell \rightarrow \bit$ of depth $L$ and any input $x\in \bit^\ell$:
    $$
\Pr\left[ y \neq C(x_1,\dots,x_\ell)
\, \bigg| \, \substack{
(\pk,\sk) \leftarrow \KeyGen(1^\lambda,1^L)\\
(\vk,\ct) \leftarrow \Enc(\pk,x)\\
(\widetilde{\ct},t_C)\from \Eval(C,\ct, \pk)\\
(\widetilde{\rho},y)\from \Extract\ip{\algo S(\widetilde{\ct},t_C),\algo R(\sk)}
}\right] \leq \negl(\lambda), \quad \text{and}
$$
$$
\Pr \left[\Vrfy(\vk,\pi) = \bot \, \bigg| \, \substack{
(\pk,\sk) \leftarrow \KeyGen(1^\lambda,1^L)\\
(\vk,\ct) \leftarrow \Enc(\pk,x)\\
(\widetilde{\ct},t_C)\from \Eval(C,\ct, \pk)\\
(\widetilde{\rho},y)\from \Extract\ip{\algo S(\widetilde{\ct},t_C),\algo R(\sk)}\\
\pi \from \Del(\widetilde{\rho})
}\right] \, \leq \, \negl(\lambda).\quad\quad
$$
\end{definition}

\begin{remark}[Compactness of an extractable $\FHE$ scheme] Our notion of an extractable $\FHE$ scheme with certified deletion in \expref{Definition}{def:extractable-FHE} requires the evaluator to keep a transcript of the circuit that is being applied, which at first sight seems to violate the usual notion of compactness in \expref{Definition}{def:compact}. However, the action of the decryptor during the interactive protocol $\Extract$ is still independent of the circuit that is being applied, and so it is possible to recover an analogous form of compactness as before.

\end{remark}

\subsection{Certified deletion security}

Our notion of certified deletion security for homomorphic encryption ($\HE$) schemes is similar to the notion of $\INDCPACD$ security for public-key encryption schemes in \expref{Definition}{def:CD-security}.

\begin{definition}[Certified deletion security for $\HE$]\label{def:CD-security-HE} Let $\Sigma = (\KeyGen,\Enc,\Dec,\Eval,\Del,\Vrfy)$ be a homomorphic encryption scheme with certified deletion and let $\algo A$ be a $\QPT$ adversary. We define the security experiment $\Exp^{\mathsf{he\mbox{-}cert\mbox{-}del}}_{\Sigma,\algo A,\lambda}(b)$ between $\algo A$ and a challenger as follows:
\begin{enumerate}
 \item The challenger generates a pair $(\pk,\sk) \from \KeyGen(1^\lambda)$, and sends $\pk$ to $\algo A$.
    \item $\algo A$ sends a distinct plaintext pair $(m_0,m_1) \in \bit^\ell \times \bit^\ell$ to the challenger.
    \item The challenger computes $(\vk,\ct_b) \leftarrow \Enc(\pk,m_b)$, and sends $\ket{\ct_b}$ to $\algo A$.
     \item At some point in time, $\algo A$ sends a certificate $\pi$ to the challenger.
    \item The challenger computes $\Vrfy(\vk, \pi)$ and sends $\sk$ to $\algo A$, if the output is $1$, and $0$ otherwise.
    \item $\algo A$ outputs a guess $b' \in \bit$, which is also the output of the experiment.
\end{enumerate}
We say that the scheme $\Sigma$ is $\INDCPACD$-secure if, for any $\QPT$ adversary $\algo A$, that
$$
\Adv_{\Sigma,\algo A}^{\mathsf{he}\mbox{-}\mathsf{cert}\mbox{-}\mathsf{del}}(\lambda) := |\Pr[\Exp^{\mathsf{pk}\mbox{-}\mathsf{cert}\mbox{-}\mathsf{del}}_{\Sigma,\algo A,\lambda}(0)=1] - \Pr[\Exp^{\mathsf{he}\mbox{-}\mathsf{cert}\mbox{-}\mathsf{del}}_{\Sigma,\algo A,\lambda}(1) = 1] |
 \leq \negl(\lambda).$$
\end{definition}

\section{Dual-Regev Fully Homomorphic Encryption with Certified Deletion}\label{sec:Dual-Regev-FHE}

In this section, we describe the main result of this work. We introduce a protocol that allows an untrusted quantum server to perform homomorphic operations on encrypted data, and to simultaneously prove data deletion to a client. Our $\FHE$ scheme with certified deletion supports the evaluation of polynomial-sized Boolean circuits composed entirely of $\NAND$ gates (see \expref{Figure}{fig:NAND}) -- an assumption we can make without loss of generality, since the $\NAND$ operation is universal for classical computation.
Note that, for $a,b \in \bit$, the logical \texttt{NOT-AND} $(\NAND)$ operation is defined by
$$
\NAND(a,b) = \overline{a \land b} = 1 - a\cdot b.
$$
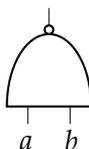
\begin{figure}[h!]
\centering
    \begin{circuitikz}
        \draw
        (0,0) node[rotate=90, nand port] (nand1) {};
        
        \node at (-0.3,-1.635) {$a$};
        \node at (0.3,-1.6) {$b$};
        
       % \node[circle,draw] (c) at (2,0){$x_1$};
    \end{circuitikz}
    \caption{$\NAND$ gate.}\label{fig:NAND}
\end{figure}
Recall also that a Boolean circuit with input $x \in \bit^n$ is a directed acyclic graph $G=(V,E)$
    in which each node in $V$ is either an input node (corresponding to an input bit $x_i$), an \texttt{AND} ($\land$) gate, an \texttt{OR} ($\lor$) gate, or a \texttt{NOT} ($\neg$) gate. We can naturally identify a Boolean circuit with a function $f: 
    \bit^n \rightarrow \bit$ which it computes. Due to the universality of the $\NAND$ operation, we can represent every Boolean circuit (and the function it computes) with an equivalent circuit consisting entirely of $\NAND$ gates. In \expref{Figure}{fig:Boolean-circuit}, we give an example of a Boolean circuit composed of three $\NAND$ gates that takes as input a string $x \in \bit^4$.

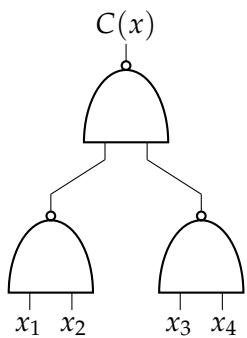
\begin{figure}[h!]
\centering
    \begin{circuitikz}
        \draw
        (0,0) node[rotate=90, nand port] (nand1) {};
        \node at (-0.3,-1.6) {$x_1$};
        \node at (0.3,-1.6) {$x_2$};
        
        \draw
        (2,0) node[rotate=90, nand port] (nand2) {};
        \node at (1.7,-1.6) {$x_3$};
        \node at (2.3,-1.6) {$x_4$};
    
        \draw
        (1,2) node[rotate=90, nand port] (nand3) {};
        
        \draw (nand1.out) -- (nand3.in 1);
        \draw (nand2.out) -- (nand3.in 2);
        
        \node at (1,2.4) {$C(x)$};
        
    \end{circuitikz}
    \caption{
A Boolean circuit $C$ made up of three $\NAND$ gates which takes as input a binary string of the form $x \in \bit^4$. The top-most $\NAND$ gate is the designated output node with outcome $C(x) \in \bit$. } \label{fig:Boolean-circuit}
\end{figure}

\subsection{Construction} \label{sec:FHE_construction}

In this section, we describe our fully homomorphic encryption scheme with certified deletion.
In order to define our construction, we require a so-called \emph{flattening} operation first introduced by Gentry, Sahai and Waters~\cite{GSW2013} in the context of homomorphic encryption and is also featured in the Dual-Regev $\FHE$ scheme of Mahadev~\cite{mahadev2018classical}. Let $n \in \N$, $q \geq 2$ be a prime modulus and $m \geq 2 n \log q$. We define a linear operator $\vec G \in \Z_q^{(m+1) \times N}$ called the \emph{gadget matrix}, where $N = (n+1) \cdot \lceil \log q \rceil$. The operator $\vec G$ converts a binary representation of a vector back to its original vector representation over the ring $\Z_q$. More precisely, for any binary vector $\vec a = (a_{1,0},\hdots, a_{1,l-1},\hdots,a_{m+1,0},\hdots,a_{m+1,l-1})$ of length $N$ with $\ell = \lceil \log q \rceil$, the matrix $\vec G$ produces a vector in $\Z_q^{m+1}$ as follows:

\begin{align}\label{eq:matrix-G}
\vec G(\vec a) = \left(
\sum_{j=0}^{\lceil \log q \rceil-1} 2^j \cdot a_{1,j} \,\,\,, \hdots, \sum_{j=0}^{\lceil \log q \rceil-1} 2^j \cdot a_{m+1,j}
\right).
\end{align}
We also define the associated (non-linear) inverse operation $\vec G^{-1}$ which converts a vector $\vec a \in \Z_q^{m+1}$ to its binary representation in $\bit^N$. In other words, we have that $\vec G^{-1} \cdot \vec G = \id$ acts as the identity operation.

Our (leveled) $\FHE$ scheme with certified deletion is based on the (leveled) Dual-Regev $\FHE$ scheme introduced by Mahadev~\cite{mahadev2018classical} which is a variant of the $\LWE$-based $\FHE$ scheme proposed by Gentry, Sahai and Waters~\cite{GSW2013}. We base our choice of parameters on the aforementioned two works.

Let us first recall the Dual-Regev $\FHE$ scheme below.

\begin{construction}[Dual-Regev leveled $\FHE$]\label{cons:Dual-Regev-FHE}
Let $\lambda \in \N$ be the security parameter.
The Dual-Regev leveled $\FHE$ scheme $\DualFHE = (\KeyGen,\Enc,\Dec,\Eval)$ consists of the following $\PPT$ algorithms:
\begin{description}
\item $\KeyGen(1^\lambda) \rightarrow (\pk,\sk):$ sample a uniformly random matrix $\bar{\vec A} \rand \Z_q^{n\times m}$ and vector $\bar{\vec x} \rand \bit^{m}$
and let $\vec A = [\bar{\vec A} | \bar{\vec A} \cdot \bar{\vec x} \Mod{q}]^T$.
Output $(\pk,\sk)$, where $\pk=\vec A \in \Z_q^{(m+1) \times n}$ and $\sk = (-\bar{\vec x}, 1) \in \Z_q^{m+1}$.
\item $\Enc(\pk,x):$ to encrypt $x \in \bit$, parse $\vec A \in \Z_q^{(m+1) \times n} \leftarrow \pk$, sample $\vec S \rand \Z_q^{n \times N}$ and $\vec E \sim D_{\Z^{(m+1)\times N}, \,\alpha q}$ and
output $\ct= \vec A \cdot \vec S + \vec E + x \cdot \vec G \Mod{q} \in \Z_q^{(m+1)\times N}$, where $\vec G$ is the gadget matrix in Eq.~\eqref{eq:matrix-G}.

\item $\Eval( C,\ct):$ apply the circuit $C$ composed of $\NAND$ gates on a ciphertext tuple $\ct$ as follows:
\begin{itemize}
    \item parse the ciphertext tuple as $(\ct_1,\dots,\ct_\ell) \leftarrow \ct$.
    \item repeat for every $\NAND$ gate in $C$: to apply a $\NAND$ gate on a ciphertext pair $(\ct_i,\ct_j)$, parse matrices $\vec C_i \leftarrow \ct_i$ and $\vec C_j \leftarrow \ct_j$ with $\vec C_i,\vec C_j \in \Z_q^{(m+1)\times N}$ and generate 
    $$
    \vec C_{ij} = \vec G - \vec C_i \cdot \vec G^{-1}(\vec C_j) \Mod{q}.
    $$
    Let $\ct_{ij} \leftarrow \vec C_{ij}$ denote the outcome ciphertext.
\end{itemize}

\item $\Dec(\sk,\ct):$ parse $\vec C \in \Z_q^{(m+1)\times N} \leftarrow \ct$ and compute $c = \sk^T \cdot \vec c_N \in \Z \cap (-\frac{q}{2},\frac{q}{2}]$, where $\vec c_N \in \Z_q^{m+1}$ is the $N$-th column of $\vec C$, and then output $0$, if $c$
is closer to $0$ than to $\lfloor\frac{q}{2}\rfloor$,
and output $1$, otherwise.
\end{description}
\end{construction}

The Dual-Regev $\FHE$ scheme supports the homomorphic evaluation of a $\NAND$ gate in the following sense. If $\ct_0$ and $\ct_1$ are ciphertexts that encrypt two bits $x_0$ and $x_1$, respectively, then the resulting outcome $\ct =\vec G - \ct_0 \cdot \vec G^{-1}(\ct_1) \Mod{q}$ is an encryption of $\NAND(x_0,x_1) = 1 - x_0 \cdot x_1$, where $\vec G$ is the gadget matrix that converts a binary representation of a vector back to its original representation over the ring $\Z_q$. Moreover, the new ciphertext $\ct$ maintains the form of an $\LWE$ sample with respect to the same public key $\pk$, albeit for a new $\LWE$ secret and a new (non-necessarily Gaussian) noise term of bounded magnitude. This property is crucial, as knowledge of the secret key $\sk$ (i.e., a short trapdoor vector) still allows for the decryption of the ciphertext $\ct$ once a $\NAND$ gate has been applied.

The following result is implicit in the work of Mahadev~\cite[Theorem 5.1]{mahadev2018classical}.

\begin{theorem}[\cite{mahadev2018classical}]
Let $\lambda \in \N$ be the security parameter. Let $n \in \N$, let $q\geq 2$ be a prime modulus and $m \geq 2 n \log q$. Let $N = (n+1) \cdot \lceil \log q \rceil$ be an integer and let $L$ be an upper bound on the depth of the polynomial-sized Boolean circuit which is to be evaluated. Let $\alpha \in (0,1)$ be a ratio such that 
$$
 2 \sqrt{n} \leq \alpha q \leq \frac{q}{4(m+1)\cdot  N\cdot (N+1)^L}.
$$
Then, the scheme in \expref{Construction}{cons:Dual-Regev-FHE} is an $\INDCPA$-secure leveled fully homomorphic encryption scheme under the $\LWE_{n,q,\alpha q}^{(m+1)\times N}$ assumption.
\end{theorem}
Note that the Dual-Regev $\FHE$ scheme is \emph{leveled} in the sense that an apriori upper bound $L$ on the $\NAND$-depth of the circuit is required to set the parameters appropriately. We remark that a proper (non-leveled) $\FHE$ scheme can be obtained under an additional circular security assumption~\cite{BrakerskiVaikuntanathan2011}.

The leveled Dual-Regev $\FHE$ scheme inherits a crucial property from its public-key counterpart. Namely, in contrast to the $\FHE$ scheme in~\cite{GSW2013}, the ciphertext takes the form of a regular sample from the $\LWE$ distribution together with an additive shift $x \cdot \vec G$ that depends on the plaintext $x \in \bit$. In particular, if a Boolean circuit $C$ of polynomial $\NAND$-depth $L$ is applied to the ciphertext corresponding to a plaintext $x \in \bit^\ell$ in \expref{Construction}{cons:Dual-Regev-FHE}, then the resulting final ciphertext is of the form $\vec A \cdot \vec S + \vec E + C(x) \vec G$,
where $\vec S \in \Z_q^{n \times N}$, $\vec E \in \Z_q^{(m+1)\times N}$ and $\| \vec E \|_\infty \leq \alpha q \sqrt{(m+1)N} \cdot(N+1)^{L}$ (see~\cite{GSW2013} for details). 
Choosing $1/\alpha$ to be sub-exponential in $N$ as in~\cite{GSW2013}, we can therefore allow for homomorphic computations of arbitrary polynomial-sized Boolean circuits of $\NAND$-depth at most $L$.
It is easy to see that the decryption procedure of the leveled Dual-Regev $\FHE$ scheme
is successful as long as the cumulative error $\vec E$ satisfies the condition $\| \vec E\|_\infty \leq \frac{q}{4\sqrt{(m+1)N}}$.

This property is essential as it allows us to extend Dual-Regev $\PKE$ scheme with certified deletion towards a leveled $\FHE$ scheme, which we denote by $\FHECD$. Using Gaussian coset states, we can again encode Dual-Regev ciphertexts for the purpose of certified deletion while simultaneously preserving their cryptographic functionality. 

\paragraph{Dual-Regev leveled $\FHE$ with certified deletion.}

Let us now describe our (leveled) $\FHE$ scheme with certified deletion. We base our choice of parameters on the Dual-Regev $\FHE$ scheme of Mahadev~\cite{mahadev2018classical} which is a variant of the scheme due to Gentry, Sahai and Waters~\cite{GSW2013}.

\paragraph{Parameters.} Let $\lambda \in \N$ be the security parameter and let $n \in \N$. Let $L$ be an upper bound on the depth of the polynomial-sized Boolean circuit which is to be evaluated. We choose the following set of parameters for the Dual-Regev leveled $\FHE$ scheme (each parameterized by the security parameter $\lambda$).
\begin{itemize}
    \item a prime modulus $q \geq 2$.
    \item an integer $m \geq 2n \log q$.
    \item an integer $N = (n+1) \cdot \lceil \log q \rceil$.
     
     \item a noise ratio $\alpha\in (0,1)$ such that
$$
\sqrt{8(m+1)N}\leq \alpha q \leq \frac{q}{\sqrt{8}(m+1)\cdot  N\cdot (N+1)^L}.
$$
\end{itemize}

\begin{construction}[Dual-Regev leveled $\FHE$ scheme with certified deletion]\label{cons:FHE-cd}
Let $\lambda \in \N$ be a parameter and $\DualFHE = (\KeyGen,\Enc,\Dec,\Eval)$ be the scheme in \expref{Construction}{cons:Dual-Regev-FHE}.
The Dual-Regev (leveled) $\FHE$ scheme $\DualFHECD = (\KeyGen,\Enc,\Dec,\Eval,\Del,\Vrfy)$ with certified deletion is defined by:
\begin{description}
\item $\KeyGen(1^\lambda) \rightarrow (\pk,\sk):$ generate $(\pk,\sk) \leftarrow \DualFHE.\KeyGen(1^\lambda)$ and output $(\pk,\sk)$.
\item $\Enc(\pk,x) \rightarrow (\vk,\ket{\ct}):$ to encrypt a bit $x\in \bit$, parse $\vec A \in \Z_q^{(m+1) \times n} \leftarrow \pk$ and, for $i \in [N]$, run $(\ket{\psi_{\vec y_i}},\vec y_i) \leftarrow \mathsf{GenPrimal}(\vec A^T,1/\alpha)$ in \expref{Algorithm}{alg:GenPrimal}, where $\vec y_i \in \Z_q^n$, and output the pair
$$
\left(\vk \leftarrow (\vec A \in \Z_q^{(m+1) \times n},(\vec y_1|\dots|\vec y_N)  \in \Z_q^{n \times N}), \quad \ket{\ct} \leftarrow \vec X_q^{x \cdot \vec g_1 } \ket{\psi_{\vec y_1}} \otimes \dots \otimes \vec X_q^{x \cdot \vec g_N } \ket{\psi_{\vec y_N}} \right),
$$
where $(\vec g_1,\dots,\vec g_N)$ are the columns of the gadget matrix $\vec G \in \Z_q^{(m+1)\times N}$ in Eq.~\eqref{eq:matrix-G}.

\item $\Eval(C,\ket{\ct}) \rightarrow (\ket{\widetilde{\ct}},t_C)$: apply the Boolean circuit $C$ composed of $\NAND$ gates to the ciphertext $\ket{\ct}$ in system $C_{\mathsf{in}} = C_1 \cdots C_\ell$ as follows: For every gate $\NAND_{ij}$ in the circuit $C$ between a ciphertext pair in systems $C_i$ and $C_j$, repeat the following two steps:
\begin{itemize}
    \item apply $U_{\NAND}$ from \expref{Definition}{def:homomorphic-NAND-gate} to systems $C_i C_j$ of the ciphertext $\ct$ by appending an auxiliary system $C_{ij}$. This results in a new ciphertext state $\ct$ which contains the additional system $C_{ij}$.
    \item add the gate $\NAND_{ij}$ to the circuit transcript $t_C$. 
\end{itemize}
Output $(\ket{\widetilde{\ct}},t_C)$, where $\ket{\widetilde{\ct}}$ is the final post-evaluation state in systems $C_{\mathsf{in}} C_\aux C_{\mathsf{out}}$ and
\begin{itemize}
    \item $C_{\mathsf{in}} = C_1 \cdots C_\ell$ denotes the initial ciphertext systems of $\ket{\ct_1} \otimes \dots \otimes \ket{\ct_\ell}$.
    \item $C_\aux$ denotes all intermediate auxiliary ciphertext systems.
    \item $C_{\mathsf{out}}$ denotes the final ciphertext system corresponding to the output of the circuit $C$.
\end{itemize}

\item $\Dec(\sk,\ket{\ct}) \rightarrow \bit^\mu \, \mathbf{or} \, \bot:$ measure the ciphertext $\ket{\ct}$ in the computational basis to obtain an outcome $\vec C$ and output $x' \leftarrow \DualFHE.\Dec(\sk,\vec C)$.

\item $\Del(\ket{\ct}) \rightarrow \pi:$ measure $\ket{\ct}$ in the Fourier basis with outcomes $\pi = (\pi_1|\dots|\pi_N) \in \Z_q^{(m+1)\times N}$.

\item $\Extract\ip{\algo S(\ket{\widetilde{\ct}},t_C),\algo R(\sk)} \rightarrow (\rho,y)$ 
this is the following interactive protocol between a sender $\algo S$ with input $\ket{\widetilde{\ct}}$ in systems $C_{\mathsf{in}} C_\aux C_{\mathsf{out}}$ and transcript $t_C$, and a receiver $\algo R$ with input $\sk$:
\begin{itemize}
    \item $\algo S$ and $\algo R$ run the rewinding protocol $\Pi = \ip{\algo S(\ket{\widetilde{\ct}},t_C),\algo R(\sk)}$ in \expref{Protocol}{prot:rewinding}. 
    
    \item Once $\Pi$ is complete, $\algo S$ obtains a state $\rho$ in system $C_{\mathsf{in}}$ and $\algo R$ obtains a bit $y \in \bit$.
\end{itemize}

\item $\Vrfy(\vk,\pk,\pi) \rightarrow \bit:$ to verify the deletion certificate $\pi = (\pi_1|\dots|\pi_N) \in \Z_q^{(m+1)\times N}$, parse $(\vec A \in \Z_q^{(m+1) \times n},(\vec y_1|\dots|\vec y_N)  \in \Z_q^{n \times N}) \leftarrow \vk$ and output $\top$, if both $\vec A^T \cdot \pi_i = \vec y_i \Mod{q}$ and $\| \pi_i \| \leq \sqrt{m+1}/\sqrt{2}\alpha$ for every $i \in [N]$, and output $\bot$, otherwise.
\end{description}
\end{construction}

\begin{protocol}[Rewinding Protocol]\label{prot:rewinding} Let $\DualFHE = (\KeyGen,\Enc,\Dec,\Eval)$ be the Dual-Regev $\FHE$ scheme in \expref{Construction}{cons:Dual-Regev-FHE}. Consider the following interactive protocol $\Pi = \ip{\algo S(\rho,t_C),\algo R(\sk)}$ between a sender $\algo S$ which takes as input state $\rho$ in systems $C_{\mathsf{in}} C_\aux C_{\mathsf{out}}$ and a transcript  $t_C$ of a Boolean circuit $C$, as well as a receiver $\algo R$ which takes as input a secret key $\sk$.
  \begin{enumerate}
        \item $\algo S$ sends system $C_{\mathsf{out}}$ of the state $\rho$ associated with the encrypted output of the circuit $C$ to $\algo R$.
        \item $\algo R$ runs $U_{\DualFHE.\Dec_{\sk}}$ (with the key $\sk$ hard coded) to reversibly decrypt system $C_{\mathsf{out}}$, where
        $$
        U_{\DualFHE.\Dec_{\sk}}: \quad \ket{\vec C}_{C_{\mathsf{out}}} \otimes \ket{0}_{M} \rightarrow \ket{\vec C}_{C_{\mathsf{out}}} \otimes \ket{\DualFHE.\Dec_{\sk}(\vec C)}_{M},
        $$
        for any matrix $\vec C \in \Z_q^{(m+1)\times N}$. $\algo R$ then measures system $M$ to obtain a bit $y \in \bit$ (the supposed output of the Boolean circuit $C$). Afterwards, $\algo R$ applies $ U_{\DualFHE.\Dec_{\sk}}^\dag$, discards the ancillary system $M$, and sends back the post-measurement system $\widetilde{C_{\mathsf{out}}}$ of the resulting ciphertext $\widetilde{\rho}$ to $\algo S$.
        \item $\algo S$ repeats the following two steps in order to uncompute the systems $C_\aux \widetilde{C_{\mathsf{out}}}$ from the state $\widetilde{\rho}$: For every gate 
    $\NAND_{ij} \in t_C$, where $i$ and $j$ denote the respective ciphertext systems $C_i$ and $C_j$, in decreasing order starting from the last gate in the circuit transcript $t_C$:
    \begin{itemize}
        \item $\algo S$ applies $U_{\NAND}^\dag$ from \expref{Definition}{def:homomorphic-NAND-gate} to systems $C_i C_j C_{ij}$ of $\widetilde{\rho}$ to uncompute system $C_{ij}$.
        \item $\algo S$ repeats the procedure starting from the new outcome state $\widetilde{\rho}$.
        \end{itemize}
    \end{enumerate}
\end{protocol}

Let us now define how to perform the homomorphic $\NAND$ gate in \expref{Construction}{cons:FHE-cd} in more detail.

\begin{definition}[Homomorphic $\NAND$ gate]\label{def:homomorphic-NAND-gate}
Let $q \geq 2$ be a modulus, and let $m$ and $N$ be integers. Let $\vec X,\vec Y,\vec Z \in \Z_q^{(m+1) \times N}$ be arbitrary matrices. We define the homomorphic $\NAND$ gate as the unitary
$$
U_\NAND: \quad \ket{\vec X}_X \otimes \ket{\vec Y}_Y \otimes \ket{\vec Z}_Z \quad \rightarrow \quad  \ket{\vec X}_X \otimes \ket{\vec Y}_Y \otimes \ket{\vec Z + \vec G - \vec X \cdot \vec G^{-1}(\vec Y) \Mod{q}}_Z,
$$
where $\vec G \in \Z_q^{(m+1) \times N}$ is the gadget matrix in Eq.~\eqref{eq:matrix-G}.
\end{definition}

To illustrate the action of our homomorphic $\NAND$ gate, we consider a simple example.

\paragraph{Example.} Consider a pair of two ciphertexts $\ket{\ct_i} \otimes \ket{\ct_j}$ which encrypt two bits $x_i,x_j \in \bit$ as in \expref{Construction}{cons:FHE-cd}. Let $U_{\NAND_{ij}}$ denote the homomorphic $\NAND$ gate applied to systems $C_i$ and $C_j$. Then,
\begin{align*}
U_{\NAND_{ij}}: \quad \ket{\ct_i}_{C_i} \otimes \ket{\ct_j}_{C_j} \otimes \, \ket{\vec 0}_{C_{ij}} \quad \rightarrow \quad  \ket{\ct_{ij}}_{C_i C_jC_{ij}}.
\end{align*}
Here, $\ket{\ct_{ij}}$ is the resulting ciphertext in systems $C_i C_jC_{ij}$. Note that $U_{\NAND_{ij}}$ is reversible in the sense that
\begin{align*}
U_{\NAND_{ij}}^\dag: \quad  \ket{\ct_{ij}}_{C_i C_jC_{ij}} \quad \rightarrow \quad \ket{\ct_i}_{C_i} \otimes \ket{\ct_j}_{C_j} \otimes \, \ket{\vec 0}_{C_{ij}}.
\end{align*}
Let us now analyze how $U_\NAND$ acts on the basis states of a pair of ciphertexts $\ket{\ct_i} \otimes \ket{\ct_j}$ that encode $\LWE$ samples as in \expref{Construction}{cons:FHE-cd}. In the following, $\vec E_i,\vec E_j\sim D_{\Z_q^{(m+1)\times N}, \,\frac{\alpha q}{\sqrt{2}}}$ have a (truncated) discrete Gaussian distribution as part of the superposition. Then,
\begin{align*}
U_{\NAND_{ij}} :\quad &\ket{\vec A \vec S_i + \vec E_i  + x_i \vec G }_{C_i}   \otimes \ket{\vec A  \vec S_j + \vec E_j + x_j \vec G }_{C_j} \otimes \, \ket{\vec 0}_{C_{ij}}\\
\rightarrow \,\, &\ket{\vec A \vec S_i + \vec E_i + x_i \vec G }_{C_i}   \otimes \ket{\vec A  \vec S_j + \vec E_j  + x_j \vec G }_{C_j} \otimes \, \ket{\vec A  \vec S_{ij} + \vec E_{ij} + (1 - x_i x_j) \vec G }_{C_{ij}},
\end{align*}
where introduced the following matrices
\begin{align*}
\vec S_{ij} &:= - \vec S_i \cdot \vec G^{-1} ( \vec A \vec S_j + \vec E_j  + x_j \vec G) - x_i \vec S_i \,\, \Mod{q} \\
\vec E_{ij} &:= - \vec E_i \cdot \vec G^{-1} ( \vec A \vec S_j + \vec E_j + x_j \vec G) - x_i \vec E_j \,\,\Mod{q}.
\end{align*}
Because the initial error terms have the property that $\|\vec E_i\|_\infty,\|\vec E_j\|_\infty \leq \alpha q \sqrt{(m+1)N/2}$, it follows that the resulting error after a single $\NAND$ gate is at most (see also \cite{GSW2013, mahadev2018classical} for more details)
$$
\|\vec E_{ij}\|_\infty \leq \alpha q \sqrt{\frac{(m+1)N}{2}} \cdot (N+1).
$$
In other words, the cumulative error term remains short relative to the modulus $q$ after every application of a homomorphic $\NAND$ gate, exactly as in the Dual-Regev $\FHE$ scheme of Mahadev~\cite{mahadev2018classical}.

\begin{figure}[h!]
\centering
    \begin{circuitikz}[> = latex, scale = 2.7]
        \draw
        (0,0) node[rotate=90, nand port,scale=3] (nand1) {};
        \node at (-0.3,-1.7) {$\ket{\ct_1}_{C_1}$};
        \node at (0.35,-1.7) {$\ket{\ct_2}_{C_2}$};
        
        \draw
        (2,0) node[rotate=90, nand port,scale=3] (nand2) {};

        \node at (1.7,-1.7) {$\ket{\ct_3}_{C_3}$};
        \node at (2.35,-1.7) {$\ket{\ct_4}_{C_4}$};
    
        \draw
        (1,2) node[rotate=90, nand port,scale=3] (nand3) {};
        
        \draw (nand1.out) -- (nand3.in 1);
        \draw (nand2.out) -- (nand3.in 2);

        \node at (-0.25,0.5) {$\ket{\ct_{12}}_{C_1 C_2 C_{12}}$};
        
        \node at (2.25,0.5) {$\ket{\ct_{34}}_{C_3 C_4 C_{34}}$};
        
        \node at (1,2.38) {$\ket{\ct_{12,34}}_{C_1 C_ 2 C_3 C_4 C_{12} C_{34} C_{12,34}}$};

        \node at (0,-0.7) {$U_\NAND$};
        
        \node at (1,1.35) {$U_\NAND$};
        
        \node at (2,-0.7) {$U_\NAND$};
        
         \node at (3.7,-1.6) {$C_{\mathsf{in}}= C_1 C_2 C_3 C_4$};
         
        \node at (3.7,0.5) {$C_\aux = C_{12} C_{34}$};
        
        \node at (3.7,2.38) {$C_{\mathsf{out}} = C_{12,34}$};
    \end{circuitikz}
    \caption{Homomorphic evaluation of a Boolean circuit $C$ composed entirely of three $\NAND$ gates. Here, the input is the quantum ciphertext $\ket{\ct_1} \otimes \ket{\ct_2} \otimes \ket{\ct_3} \otimes \ket{\ct_4}$ which corresponds to an encryption of the plaintext $x = (x_1,\dots,x_4) \in \bit^4$ as in \expref{Construction}{cons:FHE-cd}.
    The resulting ciphertext $\ket{\ct_{12,34}}$ lives on a system $C_1 C_ 2 C_3 C_4 C_{12} C_{34} C_{12,34}$ of which the last system $C_{12,34}$ contains an encryption of $C(x) \in \bit$. 
    }\label{fig:example-quantum-FHE}
\end{figure}

\subsection{Rewinding lemma}

Notice that the procedure $\DualFHECD.\Eval$ in \expref{Construction}{cons:FHE-cd} produces a highly entangled state since the unitary operation $U_\NAND$ induces entanglement between the Gaussian noise terms. In the next lemma, we show that it is possible to \emph{rewind} the evaluation procedure to be able to prove data deletion to a client.

\begin{lemma}[Rewinding lemma]\label{lem:rewinding} 
Let $\lambda \in \N$ be the security parameter. Let $n \in \N$, let $q\geq 2$ be a prime modulus and $m \geq 2 n \log q$. Let $N = (n+1) \cdot \lceil \log q \rceil$ be an integer and let $L$ be an upper bound on the depth of the polynomial-sized Boolean circuit which is to be evaluated. Let $\alpha \in (0,1)$ be a ratio such that      
$$
\sqrt{8(m+1)N}\leq \alpha q \leq \frac{q}{\sqrt{8}(m+1)\cdot  N\cdot (N+1)^L}.
$$
Let $\DualFHECD = (\KeyGen,\Enc,\Dec,\Eval,\Del,\Vrfy)$ be the Dual-Regev (leveled) $\FHE$ scheme with certified deletion in \expref{Construction}{cons:FHE-cd} and let $\Pi$ be the interactive protocol in \expref{Protocol}{prot:rewinding}. Then, the following holds for any parameter $\lambda \in \N$, plaintext $x \in \bit^\ell$ and any polynomial-sized Boolean circuit $C$: 

After the interactive protocol  $\Pi = \ip{\algo S(\ket{\widetilde{\ct}},t_C),\algo R(\sk)}$ between the sender $\algo S$ and receiver $\algo R$ is complete, the sender $\algo S$ is in possession of a quantum state $\rho$ in system $C_{\mathsf{in}}$ that satisfies
$$
\|\rho- \proj{\ct}\|_\tr \leq \negl(\lambda),
$$
where $(\ket{\widetilde{\ct}},t_C) \leftarrow \DualFHECD.\Eval(C,\ket{\ct})$ is the post-evaluation state $\ket{\widetilde{\ct}}$ in systems $C_{\mathsf{in}} C_\aux C_{\mathsf{out}}$ and where $\ket{\ct} \from \DualFHECD.\Enc(\pk,x)$ is the initial state for $(\pk,\sk) \from \DualFHECD.\KeyGen(1^\lambda)$.
\end{lemma}
\begin{proof}
Let $\lambda \in \N$, $x \in \bit^\ell$ be a plaintext and $C$ be any Boolean circuit of $\NAND$-depth $L=\poly(\lambda)$. Let $(\ket{\widetilde{\ct}},t_C) \leftarrow \DualFHECD.\Eval(C,\ket{\ct})$ be the post-evaluation state $\ket{\widetilde{\ct}}$ in systems $C_{\mathsf{in}} C_\aux C_{\mathsf{out}}$ with circuit transcript $t_C$ and let $\rho$ be the outcome of the interactive protocol $\Pi = \ip{\algo S(\ket{\widetilde{\ct}},t_C),\algo R(\sk)}$. Recall that, in \expref{Lemma}{lem:FHE-CD-correctness-decryption}, we established that there exists a negligible $\eps(\lambda)$ such that $\DualFHE.\Dec_{\sk}$ decrypts system $C_{\mathsf{out}}$ of $\ket{\widetilde{\ct}}$ with probability at least $1 - \eps$. By the ''Almost As Good As New Lemma`` (\expref{Lemma}{lem:almost}), performing the operation $U_{\DualFHE.\Dec_{\sk}}$, measuring the ancillary register $M$ and rewinding the computation, results in a mixed state $\widetilde{\rho}$ that is within trace distance $\sqrt{\eps}$ of the post-evaluation state $\ket{\widetilde{\ct}}$. Notice that, by reversing the sequence $U_{t_C}$ of homomorphic $\NAND$ gates according to the transcript $t_C$ with respect to $\ket{\widetilde{\ct}}$, we recover the initial ciphertext $\proj{\ct} = U_{t_C}^\dag \proj{\widetilde{\ct}} \,U_{t_C}$ in system $C_{\mathsf{in}}$. By definition, we also have that $\rho = U_{t_C}^\dag \widetilde{\rho} \,U_{t_C}$. Therefore,
$$
\|\rho - \proj{\ct}\|_\tr = \|U_{t_C}^\dag \widetilde{\rho} \,U_{t_C} - U_{t_C}^\dag \proj{\widetilde{\ct}} \,U_{t_C}\|_\tr = \|\widetilde{\rho} - \proj{\widetilde{\ct}}\|_\tr \leq \sqrt{\eps(\lambda)},
$$
where we used that the trace distance is unitarily invariant. Since $\eps(\lambda) = \negl(\lambda)$, this proves the claim.
\end{proof}

\paragraph{Proof of correctness.}
Let us now verify the correctness of decryption and verification of \expref{Construction}{cons:FHE-cd}.

\begin{lemma}[Compactness and full homomorphism of $\DualFHECD$]\label{lem:FHE-CD-correctness-decryption}
Let $\lambda \in \N$ be the security parameter. Let $n \in \N$, let $q\geq 2$ be a prime and $m \geq 2 n \log q$. Let $N = (n+1) \cdot \lceil \log q \rceil$ and let $L$ be an upper bound on the depth of the polynomial-sized Boolean circuit which is to be evaluated. Let $\alpha \in (0,1)$ be a ratio with    
$$
\sqrt{8(m+1)N}\leq \alpha q \leq \frac{q}{\sqrt{8}(m+1)\cdot  N\cdot (N+1)^L}.
$$
Then, the scheme $\DualFHECD = (\KeyGen,\Enc,\Dec,\Eval,\Del,\Vrfy)$ in \expref{Construction}{cons:FHE-cd} is a compact and fully homomorphic encryption scheme with certified deletion. In other words, for any efficienty (in $\lambda \in \N$) computable circuit $C: \bit^\ell \rightarrow \bit$ and any set of inputs $x = (x_1,\dots,x_\ell) \in \bit^\ell$, it holds that:
$$
\Pr\left[\DualFHECD.\Dec(\sk,\ket{\widetilde{\ct}}) \neq C(x_1,\dots,x_\ell)
\, \bigg| \, \substack{
(\pk,\sk) \leftarrow \DualFHECD.\KeyGen(1^\lambda,1^L)\\
(\vk,\ket{\ct}) \leftarrow \DualFHECD.\Enc(\pk,x)\\
(\ket{\widetilde{\ct}},t_C)\from \DualFHECD.\Eval(C,\ket{\ct}, \pk)
}\right] \leq \negl(\lambda).
$$
\end{lemma}

\begin{proof}
Let $\ket{\ct}$ be the ciphertext output by $\DualFHECD.\Enc(\pk,x)$, where $x \in \bit^\ell$ denotes the plaintext, and let $(\ket{\widetilde{\ct}},t_C) \leftarrow \DualFHECD.\Eval(C,\ket{\ct})$ be the output of the evaluation procedure.
Let us first consider the case when $t_C = \emptyset$, i.e. not a single $\NAND$ gate has been applied to the ciphertext. In this case, the claim follows from the fact that the truncated discrete Gaussian $D_{\Z_q^{(m+1)\times N},\frac{\alpha q}{\sqrt{2}}}$ is supported on $\{\vec X \in \Z_q^{(m+1)\times N} \, : \, \|\vec X\|_\infty \leq \alpha q\sqrt{N(m+1)/2}\}$.
Recall that $\DualFHECD.\Dec(\sk,\ket{\widetilde{\ct}})$ measures the ciphertext $\ket{\widetilde{\ct}}$ in the computational basis with outcome $\vec C= (\vec C_1,\dots,\vec C_\ell)$, where $\vec C_i \in \Z_q^{(m+1)\times N}$ is a matrix, and outputs $x' \leftarrow \DualFHE.\Dec(\sk,\vec C)$. By our choice of parameters, each error term satisfies
$$\| \vec E_i\|_\infty \leq \alpha q \sqrt{\frac{N(m+1)}{2}} < \frac{q}{4\sqrt{(m+1)N}}, \quad \forall i \in [\ell].$$ Hence, decryption correctness is preserved if $t_C = \emptyset$. Let us now consider the case when $t_C \neq \emptyset$, i.e. the Boolean circuit $C$ consists of at least one $\NAND$ gate which has been applied to the ciphertext $\ket{\ct}$. In this case, the cumulative error in system $C_{\mathsf{out}}$ after $L$ applications of $U_\NAND$ in \expref{Definition}{def:homomorphic-NAND-gate} is at most $\alpha q \sqrt{(m+1)N/2}(N+1)^L$, which is less than $\frac{q}{4\sqrt{(m+1)N}}$ by our choice of parameters. Therefore, the procedure $\DualFHE.\Dec_\sk$ decrypts a computational basis state in system $C_{\mathsf{out}}$ of the state $\ket{\widetilde{\ct}}$ correctly with probability at least $1 - \negl(\lambda)$. Furthermore, because the procedure $\DualFHECD.\Dec$ is independent of the circuit $C$ and its depth $L$, the scheme $\DualFHECD$ is compact. This proves the claim.
\end{proof} 

Let us now verify the correctness of verification of the scheme $\DualFHECD$ in \expref{Construction}{cons:FHE-cd} according to \expref{Definition}{def:correctness-verification-HE}. We show the following.

\begin{lemma}[Correctness of verification]\label{lem:FHE-CD-correctness-verification}
Let $\lambda \in \N$ be the security parameter. Let $n \in \N$, let $q\geq 2$ be a prime modulus and $m \geq 2 n \log q$. Let $N = (n+1) \cdot \lceil \log q \rceil$ be an integer and let $L$ be an upper bound on the depth of the polynomial-sized Boolean circuit which is to be evaluated. Let $\alpha \in (0,1)$ be a ratio with
$$
\sqrt{8(m+1)N}\leq \alpha q \leq \frac{q}{\sqrt{8}(m+1)\cdot  N\cdot (N+1)^L}.
$$
Then, the Dual-Regev $\FHE$ scheme $\DualFHECD = (\KeyGen,\Enc,\Dec,\Eval,\Del,\Vrfy)$ with certified deletion in \expref{Construction}{cons:FHE-cd} satisfies verification correctness. In other words, for any $\lambda \in \N$, any plaintext $x \in \bit^\ell$ and any polynomial-sized Boolean circuit $C$ entirely composed of $\NAND$ gates:
$$
\Pr \left[\Verify(\vk, \pi) =  1 \, \bigg| \, \substack{
    (\pk,\sk) \leftarrow \KeyGen(1^\lambda)\\
    (\vk,\ket{\ct}) \leftarrow \Enc(\pk,x)\\
    \pi \from \Del(\ket{\ct})
    }\right] \geq 1- \negl(\lambda).
$$
\end{lemma}
\begin{proof}
Consider a bit $x \in \bit$ and a public key $\pk$ given by $\vec A = [\bar{\vec A} | \bar{\vec A} \cdot \bar{\vec x} \Mod{q}] \in \Z_q^{(m+1)\times n}$, for $\bar{\vec x} \rand \bit^{m}$. By the Leftover Hash Lemma (\expref{Lemma}{lem:LHL}), the distribution of $\vec A$ is within negligible total variation distance of the uniform distribution over $\Z_q^{(m+1) \times n}$. \expref{Lemma}{lem:full-rank} implies that the columns of $\vec A$ generate $\Z_q^n$ with overwhelming probability.
We consider the ciphertext $\ket{\ct}$ output by $\Enc(\pk,x)$, where
$$
\ket{\ct} \leftarrow \vec X_q^{x \cdot \vec g_1 } \ket{\hat\psi_{\vec y_1}} \otimes \dots \otimes \vec X_q^{x \cdot \vec g_N } \ket{\hat\psi_{\vec y_N}},
$$
and where $(\vec g_1,\dots,\vec g_N)$ are the columns of the gadget matrix $\vec G \in \Z_q^{(m+1)\times N}$ in Eq.~\eqref{eq:matrix-G}.
Given our choice,
$$
\sqrt{8(m+1)N}\leq \alpha q \leq \frac{q}{\sqrt{8}(m+1)\cdot  N\cdot (N+1)^L},
$$
\expref{Corollary}{cor:switching} implies that the Fourier transform of $\ket{\ct}$ is within negligible trace distance of the state
$$
\ket{\widehat{\ct}} =\sum_{\substack{\vec x_1 \in \Z_q^{m+1}:\\ \vec A \vec x_1 = \vec y_1 \Mod{q}}}\rho_{\frac{1}{\alpha}}(\vec x_1) \, \omega_q^{\ip{\vec x_1,x \cdot \vec g_1}} \,\ket{\vec x_1} \otimes \dots \otimes \sum_{\substack{\vec x_N \in \Z_q^{m+1}:\\ \vec A \vec x_N = \vec y_N \Mod{q}}}\rho_{\frac{1}{\alpha}}(\vec x_N) \, \omega_q^{\ip{\vec x_N,x \cdot \vec g_N}} \,\ket{\vec x_N}.
$$
From \expref{Lemma}{lem:tailboundII}, it follows that the distribution of computational basis measurement outcomes is within negligible total variation distance of the sample
$$
\pi = (\pi_1,\dots,\pi_N) \sim D_{\Lambda_q^{\vec y_1}(\vec A),\frac{1}{\sqrt{2}\alpha}} \times \dots \times D_{\Lambda_q^{\vec y_N}(\vec A),\frac{1}{\sqrt{2}\alpha}},
$$
where $\| \pi_i\| \leq \sqrt{m+1}/\sqrt{2}\alpha$ for every $i \in [N]$. This proves the claim.
\end{proof}

We now show that our scheme $\DualFHECD$ in \expref{Construction}{cons:FHE-cd} is \emph{extractable} according to \expref{Definition}{def:extractable-FHE}.

\begin{lemma}[Extractability of $\DualFHECD$]
Let $\lambda \in \N$ be the security parameter. Let $n \in \N$, let $q\geq 2$ be a prime modulus and $m \geq 2 n \log q$. Let $N = (n+1) \cdot \lceil \log q \rceil$ and let $L$ be an upper bound on the depth of the polynomial-sized Boolean circuit which is to be evaluated. Let $\alpha \in (0,1)$ be a noise ratio with
$$
\sqrt{8(m+1)N}\leq \alpha q \leq \frac{q}{\sqrt{8}(m+1)\cdot  N\cdot (N+1)^L}.
$$
Then, the Dual-Regev $\FHE$ scheme $\Sigma=\DualFHECD$ with certified deletion in \expref{Construction}{cons:FHE-cd} is extractable. In other words, for any efficiently computable circuit $C: \bit^\ell \rightarrow \bit$ and any input $x\in \bit^\ell$:
    $$
\Pr\left[ y \neq C(x_1,\dots,x_\ell)
\, \bigg| \, \substack{
(\pk,\sk) \leftarrow \KeyGen(1^\lambda,1^L)\\
(\vk,\ket{\ct}) \leftarrow \Enc(\pk,x)\\
(\ket{\widetilde{\ct}},t_C)\from \Eval(C,\ket{\ct}, \pk)\\
(\rho,y)\from \Extract\ip{\algo S(\ket{\widetilde{\ct}},t_C),\algo R(\sk)}
}\right] \leq \negl(\lambda), \quad \text{and}
$$
$$
\Pr \left[\Vrfy(\vk,\pi) = \bot \, \bigg| \, \substack{
(\pk,\sk) \leftarrow \KeyGen(1^\lambda,1^L)\\
(\vk,\ket{\ct}) \leftarrow \Enc(\pk,x)\\
(\ket{\widetilde{\ct}},t_C)\from \Eval(C,\ket{\ct}, \pk)\\
(\rho,y)\from \Extract\ip{\algo S(\ket{\widetilde{\ct}},t_C),\algo R(\sk)}\\
\pi \from \Del(\rho)
}\right] \, \leq \, \negl(\lambda).\quad\quad
$$
\end{lemma}
\begin{proof}
Let $C: \bit^\ell \rightarrow \bit$ be an efficiently computable circuit and let $x\in \bit^\ell$ be any input.
Let $(\rho,y)\from \Extract\ip{\algo S(\ket{\widetilde{\ct}},t_C),\algo R(\sk)}$ denote the outcome of the interactive protocol between the sender $\algo S$ and the receiver $\algo R$, where $(\ket{\widetilde{\ct}},t_C)\from \Eval(C,\ket{\ct}, \pk)$ is the post-evaluation state and $\ct \from \Enc(\pk,x)$ is the initial ciphertext for $(\pk,\sk) \from \KeyGen(1^\lambda)$. Recall that the receiver $\algo R$ reversibly performs the decryption procedure $\Dec$ (with the secret key $\sk$ hard-coded) during the execution of the protocol $\Pi = \ip{\algo S(\ket{\widetilde{\ct}},t_C),\algo R(\sk)}$ in \expref{Protocol}{prot:rewinding}. Therefore, it follows that the measurement outcome $y$ is equal to $C(x_1,\dots,x_\ell)$ with overwhelming probability due \expref{Lemma}{lem:FHE-CD-correctness-decryption}. This shows the first property.

To show the second property, we can use the Rewinding Lemma (\expref{Lemma}{lem:rewinding}) to argue that after the interactive protocol  $\Pi = \ip{\algo S(\widetilde{\ct},t_C),\algo R(\sk)}$ between the sender $\algo S$ and receiver $\algo R$ is complete, the sender $\algo S$ is in possession of a quantum state $\rho$ in system $C_{\mathsf{in}}$ that satisfies
$$
\|\rho- \proj{\ct}\|_\tr \leq \negl(\lambda).
$$
Therefore, the claim follows immediately from the verification correctness of $\Sigma$ shown in \expref{Lemma}{lem:FHE-CD-correctness-verification}.
\end{proof}

\subsection{Proof of security}\label{sec:FHE-CD-security}

Let us now analyze the security of our $\FHE$ scheme with certified deletion in \expref{Construction}{cons:FHE-cd}. Note that the results in this section all essentially carry over from \expref{Section}{sec:DualRegev-PKE-CD}, where we analyzed the security of our Dual-Regev $\PKE$ scheme with certified deletion.

\paragraph{$\INDCPA$ security of $\DualFHECD$.}
We first prove that our  scheme $\FHECD$ in \expref{Construction}{cons:FHE-cd}  satisfies the notion $\INDCPA$ security according to \expref{Definition}{def:ind-cpa}. The proof is identical to the proof of $\INDCPA$-security of our $\DualPKE$ scheme in \expref{Theorem}{thm:Dual-Regev-PKE-CPA}. We add it for completeness.

\begin{theorem}\label{thm:FHECD-CPA} Let $n \in \N$, let $q \geq 2$ be a modulus, let $m \geq 2n \log q$ and let $N = (n+1)\lceil\log q \rceil$, each parameterized by the security parameter $\lambda \in \N$. Let $\alpha \in (0,1)$ be a noise ratio parameter such that $\sqrt{8(m+1)N} \leq \frac{1}{\alpha} \leq \frac{q}{\sqrt{8(m+1)N}}$. Then, the scheme $\DualFHECD$ in \expref{Construction}{cons:FHE-cd} is $\INDCPA$-secure assuming the quantum hardness of (decisional) $\LWE_{n,q,\beta q}^{(m+1)N}$, for any $\beta \in (0,1)$ with $\alpha/\beta= \lambda^{\omega(1)}$.
\end{theorem}

\begin{proof}
Let $\Sigma = \DualFHECD$. We need to show that, for any $\QPT$ adversary $\algo A$, it holds that
$$
\Adv_{\Sigma,\algo A}(\lambda) := |\Pr[\Exp^{\mathsf{ind\mbox{-}cpa}}_{\Sigma,\algo A,\lambda}(0)=1] - \Pr[\Exp^{\mathsf{ind\mbox{-}cpa}}_{\Sigma,\algo A,\lambda}(1) = 1] |
 \leq \negl(\lambda).$$
Consider the experiment $\Exp^{\mathsf{ind\mbox{-}cpa}}_{\Sigma,\algo A,\lambda}(b)$ between the adversary $\algo A$ and a challenger taking place as follows:
\begin{enumerate}
    \item The challenger generates a pair $(\pk,\sk) \from \KeyGen(1^\lambda)$, and sends $\pk$ to $\algo A$.
    \item $\algo A$ sends a distinct plaintext pair $(m_0,m_1) \in \bit^\ell \times \bit^\ell$ to the challenger.
    \item The challenger computes $(\vk,\ct_b) \leftarrow \DualFHECD.\Enc(\pk,m_b)$, and sends $\ket{\ct_b}$ to $\algo A$.
    \item $\algo A$ outputs a guess $b' \in \bit$, which is also the output of the experiment.
\end{enumerate}
Recall that the procedure $\Enc(\pk,m_b)$ outputs a pair $(\vk,\ket{\ct_b})$, where 
$$
\left(\vec A \in \Z_q^{(m+1) \times n},(\vec y_1|\dots|\vec y_N)  \in \Z_q^{n \times N} \right) \leftarrow \vk
$$
is the verification key and where the ciphertext $\ket{\ct_b}$ is within negligible trace distance of
\begin{align}\label{eq:ct-DualRegevPKECD-security}
\sum_{\vec S \in \Z_q^{n \times N}} \sum_{\vec E \in \Z_q^{(m+1)\times N}} \rho_{\alpha q}(\vec E) \, \omega_q^{-\Tr[\vec S^T \vec Y]} \ket{\vec A\cdot \vec S + \vec E + m_b \cdot \vec G \Mod{q}}.
\end{align}
Here, $\vec Y \in \Z_q^{n\times N}$ is the matrix composed of the columns $\vec y_1,\dots,\vec y_N$.
Let $\beta \in (0,1)$ be any parameter with $\alpha/\beta= \lambda^{\omega(1)}$. Then, it follows from
\expref{Theorem}{thm:pseudorandom-SLWE} that, under the (decisional) $\LWE_{n,q,\beta q}^{(m+1)N}$ assumption, $\ket{\ct_b}$ is computationally indistinguishable from the state
\begin{align}\label{eq:random-state}
\sum_{\vec U \in \Z_q^{(m+1)\times N}} \omega_q^{\Tr[\vec U^T \bar{\vec X}]}\ket{\vec U}, \quad \,\, \bar{\vec X} = (\bar{\vec x}_1,\dots,\bar{\vec x}_N)\sim D_{\Lambda_q^{\vec y_1}(\vec A),\frac{1}{\sqrt{2}\alpha}} \times \dots \times D_{\Lambda_q^{\vec y_N}(\vec A),\frac{1}{\sqrt{2}\alpha}}.
\end{align}
Here $(\bar{\vec x}_1,\dots,\bar{\vec x}_N)$ refer to the columns of the matrix $ \bar{\vec X} \in \Z_q^{(m+1)\times N}$.
Finally, because the state in Eq.~\eqref{eq:random-state} is completely independent of the bit $b \in \bit$, it follows that
$$
\Adv_{\Sigma,\algo A}(\lambda) := |\Pr[\Exp^{\mathsf{ind\mbox{-}cpa}}_{\Sigma,\algo A,\lambda}(0)=1] - \Pr[\Exp^{\mathsf{ind\mbox{-}cpa}}_{\Sigma,\algo A,\lambda}(1) = 1] |
 \leq \negl(\lambda).$$
This proves the claim.
\end{proof}

\paragraph{$\INDCPACD$ security of \DualFHECD.} 
Let us now analyze the security of our Dual-Regev homomorphic encryption scheme $\DualFHECD$ in \expref{Construction}{cons:FHE-cd}. We prove that it satisfies \emph{certified deletion security} assuming the \emph{Strong Gaussian-Collapsing (SGC) Conjecture} (see \expref{Conjecture}{conj:SGC}). This is a strengthening of the Gaussian-collapsing property which we proved under the (decisional) \LWE assumption (see \expref{Theorem}{thm:GaussCollapse}). The proof is similar to the proof of \expref{Theorem}{thm:Dual-Regev-PKE-CD}. We add it for completeness.

\begin{theorem}\label{thm:FHE-CD-security}Let $\lambda \in \N$ be the security parameter. Let $n \in \N$, let $q\geq 2$ be a prime modulus and $m \geq 2 n \log q$. Let $N = (n+1) \cdot \lceil \log q \rceil$ be an integer and let $L$ be an upper bound on the depth of the polynomial-sized Boolean circuit which is to be evaluated. Let $\alpha \in (0,1)$ be a noise ratio such that$$
\sqrt{8(m+1)N}\leq \alpha q \leq \frac{q}{\sqrt{8}(m+1)\cdot  N\cdot (N+1)^L}.
$$
Then, the Dual-Regev homomorphic
encryption scheme $\DualFHECD$ in \expref{Construction}{cons:FHE-cd} is $\INDCPACD$-secure assuming the Strong Gaussian-Collapsing property $\mathsf{SGC}_{n,(m+1),q,\frac{1}{\alpha}}^N$ from \expref{Conjecture}{conj:SGC}.
\end{theorem}

\begin{proof}
Let $\Sigma = \DualFHECD$. We need to show that, for any $\QPT$ adversary $\algo A$, it holds that
$$
\Adv_{\Sigma,\algo A}^{\mathsf{he}\mbox{-}\mathsf{cert}\mbox{-}\mathsf{del}}(\lambda) := |\Pr[\Exp^{\mathsf{he}\mbox{-}\mathsf{cert}\mbox{-}\mathsf{del}}_{\Sigma,\algo A,\lambda}(0)=1] - \Pr[\Exp^{\mathsf{he}\mbox{-}\mathsf{cert}\mbox{-}\mathsf{del}}_{\Sigma,\algo A,\lambda}(1) = 1] |
 \leq \negl(\lambda).$$
We consider the following sequence of hybrids:

\begin{description}
\item $\mathbf{H_0:}$ This is the experiment $\Exp^{\mathsf{he\mbox{-}cert\mbox{-}del}}_{\Sigma,\algo A,\lambda}(0)$ between $\algo A$ and a challenger:
\begin{enumerate}
\item The challenger samples a random matrix $\bar{\vec A} \rand \Z_q^{n\times m}$ and a vector $\bar{\vec x} \rand \bit^{m}$ and chooses $\vec A = [\bar{\vec A} | \bar{\vec A} \cdot \bar{\vec x} \Mod{q}]^T$.
The challenger chooses the secret key $\sk \leftarrow (-\bar{\vec x}, 1) \in \Z_q^{m+1}$ and the public key $\pk \leftarrow \vec A \in \Z_q^{(m+1) \times n}$.

\item $\algo A$ sends a distinct plaintext pair $(m_0,m_1) \in \bit \times \bit$ to the challenger. (Note: Without loss of generality, we can just assume that $m_0 = 0$ and $m_1=1$).
    
\item The challenger runs $(\ket{\psi_{\vec y_i}},\vec y_i) \leftarrow \mathsf{GenPrimal}(\vec A^T,\sigma)$ in \expref{Algorithm}{alg:GenPrimal}, for $i \in [N]$, and outputs
$$
\left(\vk \leftarrow (\vec A \in \Z_q^{(m+1) \times n},(\vec y_1|\dots|\vec y_N)  \in \Z_q^{n \times N}), \quad \ket{\ct_0} \leftarrow \ket{\psi_{\vec y_1}} \otimes \dots \otimes \ket{\psi_{\vec y_N}} \right).
$$

\item At some point in time, $\algo A$ returns a certificate $\pi = (\pi_1,\dots,\pi_N)$ to the challenger.

\item The challenger outputs $\top$, if $\vec A^T \cdot \pi_i = \vec y_i \Mod{q}$ and $\| \pi_i \| \leq \sqrt{m+1}/\sqrt{2}\alpha$ for $i \in [N]$, and outputs $\bot$, otherwise. If $\pi$ passes the test with outcome $\top$, the challenger sends $\sk$ to $\algo A$.
    
\item $\algo A$ outputs a guess $b' \in \bit$, which is also the output of the experiment.
\end{enumerate}

\item $\mathbf{H_1:}$ This is same experiment as in $\mathbf{H_0}$, except that (in Step 3) the challenger prepares the ciphertext in the Fourier basis rather than the standard basis. In other words, $\algo A$ receives the pair
$$
\left(\vk \leftarrow (\vec A \in \Z_q^{(m+1) \times n},(\vec y_1,\dots,\vec y_N)  \in \Z_q^{n \times N}), \quad \ket{\ct_0}  \leftarrow \FT_q \ket{\psi_{\vec y_1}} \otimes \dots \otimes \FT_q \ket{\psi_{\vec y_N}} \right).
$$

\item $\mathbf{H_2:}$ This experiment is an $N$-fold variant of $\mathsf{StrongGaussCollapseExp}_{\algo H,\algo D,\lambda}(0)$ in \expref{Conjecture}{conj:SGC}:
\begin{enumerate}
\item The challenger samples a random matrix $\bar{\vec A} \rand \Z_q^{n\times m}$ and a vector $\bar{\vec x} \rand \bit^{m}$ and chooses $\vec A = [\bar{\vec A} | \bar{\vec A} \cdot \bar{\vec x} \Mod{q}]$ and $\vec t = (-\bar{\vec x}, 1) \in \Z_q^{m+1}$.

\item The challenger runs $(\ket{\hat\psi_{\vec y_i}},\vec y_i) \leftarrow \mathsf{GenDual}(\vec A^T,\sigma)$ in \expref{Algorithm}{alg:GenDual}, for $i \in [N]$, and sends the following tiplet to the adversary $\algo A$:
$$
\left(\ket{\hat\psi_{\vec y_1}} \otimes \dots \otimes \ket{\hat\psi_{\vec y_N}}, \quad \vec A^T \in \Z_q^{n \times (m+1)},\quad \vec Y=(\vec y_1|\dots|\vec y_N)  \in \Z_q^{n \times N} \right).
$$

\item At some point in time, $\algo A$ returns a certificate $\pi$ to the challenger.

\item The challenger outputs $\top$, if $\vec A^T \cdot \pi_i = \vec y_i \Mod{q}$ and $\| \pi_i \| \leq \sqrt{m+1}/\sqrt{2}\alpha$ for $i \in [N]$, and outputs $\bot$, otherwise. If $\pi$ passes the test with outcome $\top$, the challenger sends $\sk$ to $\algo A$.
    
\item $\algo A$ outputs a guess $b' \in \bit$, which is also the output of the experiment.
\end{enumerate}

\item $\mathbf{H_3:}$ This is an $N$-fold variant of the experiment in $\mathsf{StrongGaussCollapseExp}_{\algo H,\algo D,\lambda}(1)$ in \expref{Conjecture}{conj:SGC}; it is the same as $\mathbf{H_2}$, except that the states $\ket{\hat\psi_{\vec y_1}} \otimes \dots \otimes \ket{\hat\psi_{\vec y_N}}$ (in Step 2) are measured in the computational basis before they are sent to $\algo A$. 

\item $\mathbf{H_4:}$ This is same experiment as $\mathbf{H_3}$, except that (in Step 2) the challenger additionally applies the Pauli operators $\vec Z_q^{\vec g_1} \otimes \dots \otimes \vec Z_q^{\vec g_N}$ to the states $\ket{\hat\psi_{\vec y_1}} \otimes \dots \otimes \ket{\hat\psi_{\vec y_N}}$ before they are measured in the computational basis, where $(\vec g_1,\dots,\vec g_N)$ are the rows of the gadget matrix $\vec G \in \Z_q^{(m+1)\times N}$ in Eq.~\eqref{eq:matrix-G}.

\item $\mathbf{H_5:}$ This is same experiment as $\mathbf{H_4}$, except that (in Step 2) $\algo A$ receives the triplet 
$$
\left(\vec Z_q^{ \vec g_1}\ket{\hat\psi_{\vec y_1}} \otimes \dots \otimes \vec Z_q^{\vec g_N}\ket{\hat\psi_{\vec y_N}},\quad\vec A^T \in \Z_q^{n \times (m+1)},\quad\vec y=(\vec y_1|\dots|\vec y_N) \in \Z_q^{n \times N} \right).
$$

\item $\mathbf{H_6:}$ This is same experiment as $\mathbf{H_5}$, except that (in Step 2) the challenger prepares the quantum states in the Fourier basis instead. In other words, $\algo A$ receives the triplet 
$$
\left(\FT_q^\dag \vec Z_q^{ \vec g_1}\ket{\hat\psi_{\vec y_1}} \otimes \dots \otimes \FT_q^\dag \vec Z_q^{\vec g_N}\ket{\hat\psi_{\vec y_N}},\quad\vec A^T \in \Z_q^{n \times (m+1)},\quad\vec y=(\vec y_1|\dots|\vec y_N) \in \Z_q^{n \times N} \right).
$$

\item $\mathbf{H_7:}$ This is the experiment $\Exp^{\mathsf{he\mbox{-}cert\mbox{-}del}}_{\Sigma,\algo A,\lambda}(1)$.
\end{description}

We now show that the hybrids are indistinguishable.

%% H0 vs H1
\begin{claim}
$$
 \Pr[\Exp^{\mathsf{he\mbox{-}cert\mbox{-}del}}_{\Sigma,\algo A,\lambda}(0)=1] = \Pr[\mathbf{H_1} = 1].$$
\end{claim}
\begin{proof}
Without loss of generality, we can assume that the challenger applies the inverse Fourier transform before sending the ciphertext to $\algo A$. Therefore, the success probabilities are identical in $\mathbf{H_0}$ and $\mathbf{H_1}$. 
\end{proof}

%% H2 vs H1
\begin{claim}
$$
\Pr[\mathbf{H_1} = 1] = \Pr[\mathbf{H_2} = 1].$$
\end{claim}
\begin{proof}
Because the challenger in $\mathbf{H_1}$ always sends the ciphertext $\ket{\ct_0}$ corresponding to $m_0=0$ to the adversary $\algo A$, the two hybrids $\mathbf{H_1}$ and $\mathbf{H_2}$ are identical.
\end{proof}

%% H3 vs H2
\begin{claim} Under the Strong Gaussian-Collapsing property $\mathsf{SGC}_{n,(m+1),q,\frac{1}{\alpha}}^N$, it holds that
$$
 | \Pr[\mathbf{H_2} = 1] - \Pr[\mathbf{H_3} = 1] |
 \leq \negl(\lambda).$$
\end{claim}
\begin{proof}
This follows from \expref{Conjecture}{conj:SGC}.
\end{proof}

%% H4 vs H3
\begin{claim}
$$
\Pr[\mathbf{H_3} = 1] = \Pr[\mathbf{H_4} = 1].$$
\end{claim}
\begin{proof}
Because the challenger measures the state $\ket{\hat\psi_{\vec y_1}} \otimes \dots \otimes \ket{\hat\psi_{\vec y_N}}$ in Step 2 in the computational basis, applying the phase operators $\vec Z_q^{\vec g_1} \otimes \dots \otimes \vec Z_q^{\vec g_N}$ before the measurement does not affect the outcome.
\end{proof}

%% H4 vs H5
\begin{claim} Under the Strong Gaussian-Collapsing property $\mathsf{SGC}_{n,(m+1),q,\frac{1}{\alpha}}^N$, it holds that
$$
 | \Pr[\mathbf{H_4} = 1] - \Pr[\mathbf{H_5} = 1] |
 \leq \negl(\lambda).$$
\end{claim}
\begin{proof}
This follows from \expref{Conjecture}{conj:SGC} since, without loss of generality, we can assume that the challenger applies the phase operators $\vec Z_q^{\vec g_1} \otimes \dots \otimes \vec Z_q^{\vec g_N}$ before sending the states $\ket{\hat\psi_{\vec y_1}} \otimes \dots \otimes \ket{\hat\psi_{\vec y_N}}$ to $\algo A$ as input.
\end{proof}

%% H5 vs H6
\begin{claim}
$$
\Pr[\mathbf{H_5} = 1] = \Pr[\mathbf{H_6} = 1].
$$
\end{claim}

\begin{proof}
Without loss of generality, we can assume that the challenger applies the Fourier transform to the state $\vec Z_q^{ \vec g_1}\ket{\hat\psi_{\vec y_1}} \otimes \dots \otimes \vec Z_q^{\vec g_N}\ket{\hat\psi_{\vec y_N}}$ before sending it to the adversary $\algo A$. Therefore, the success probabilities in $\mathbf{H_5}$ and $\mathbf{H_6}$ are identical. 
\end{proof}

%% H6 vs H7
\begin{claim}
$$
|\Pr[\mathbf{H_6} = 1] -\Pr[\Exp^{\mathsf{pk\mbox{-}cert\mbox{-}del}}_{\Sigma,\algo A,\lambda}(1)=1]| \leq \negl(\lambda).$$
\end{claim}
\begin{proof}
From \expref{Lemma}{lem:XZ-conjugation}, we have $\FT_q \vec X_q^{\vec v} = \vec Z_q^{\vec v} \FT_q$, for all $\vec v \in \Z_q^m$. Hence, in $\mathbf{H_6}$, we can instead assume that the challenger runs $(\ket{\psi_{\vec y_i}},\vec y_i) \leftarrow \mathsf{GenPrimal}(\vec A^T,1/\alpha)$ in \expref{Algorithm}{alg:GenPrimal}, for $i \in [N]$, and then sends the following to $\algo A$:
$$
\left(\vk \leftarrow (\vec A \in \Z_q^{(m+1)\times n}, (\vec y_1|\dots|\vec y_N)  \in \Z_q^{n \times N}), \quad \ket{\ct_1}  \leftarrow \vec X_q^{\vec g_1}\ket{\psi_{\vec y_1}} \otimes \dots \otimes \vec X_q^{\vec g_N}\ket{\psi_{\vec y_N}} \right).
$$
From \expref{Corollary}{cor:switching}, it follows that the states $\FT_q^\dag \vec Z_q^{\vec v} \ket{\hat\psi_{\vec y}}$ and $\vec X_q^{\vec v} \ket{{\psi_{\vec y}}}$ are within negligible trace distance, for all $\vec v \in \Z_q^m$. Because the challenger in $\mathbf{H_7}$ always sends $\ket{\ct_1}$ corresponding to $m_1=1$ to the adversary $\algo A$, it follows that the distinguishing advantage between $\mathbf{H_6}$ and $\mathbf{H_7}=\Exp^{\mathsf{he\mbox{-}cert\mbox{-}del}}_{\Sigma,\algo A,\lambda}(1)$ is negligible.
\end{proof} 
Because the hybrids $\mathbf{H_0}$ and $\mathbf{H_7}$ are indistinguishable, this implies that
$$
\Adv_{\Sigma,\algo A}^{\mathsf{he\mbox-}\mathsf{cert}\mbox{-}\mathsf{del}}(\lambda)\leq \negl(\lambda).$$
\end{proof}

\appendix

\bibliographystyle{alpha}
\bibliography{references}

\end{document}